\documentclass{article}[12pt]
\usepackage[utf8]{inputenc}
\usepackage[english]{babel}
\usepackage[utf8]{inputenc}

\usepackage{lscape}
\usepackage{amssymb}
\usepackage{amsmath}
\usepackage{mathtools}
\usepackage{stmaryrd}
\usepackage[hypcap=false]{caption}
\usepackage{multirow}
\usepackage{hhline}
\usepackage{mathrsfs}
\usepackage{tabularx}
\usepackage[table]{xcolor}
\usepackage{tablestyles}

\usepackage{amsthm}
\usepackage{makeidx}
\usepackage[color,all]{xy}
\usepackage{xcolor,colortbl}

\usepackage[hmargin={2cm,2cm},vmargin={2cm,2cm}]{geometry}

\usepackage[pdftex]{graphicx}
\usepackage{comment}
\usepackage{tikz}
\usetikzlibrary{shapes.misc,shadows}

\usepackage[colorlinks=true,linkcolor=black]{hyperref}

\hypersetup{urlcolor=blue, citecolor=red}

\newcommand{\rr}{\mathbb{R}}
\newcommand{\rz}{\mathbb{R} \backslash \{0\}}
\newcommand{\rpm}{\mathbb{R}_{\pm}}
\newtheorem{theorem}{Theorem}[section]

\newtheorem{proposition}{Proposition}

\theoremstyle{definition} 

\newtheorem{definition}[theorem]{Definition}
\newtheorem{remark}{Remark}

\numberwithin{equation}{section}
\numberwithin{proposition}{section}
\numberwithin{example}{section}
\numberwithin{corollary}{section}
\numberwithin{remark}{section}
\title{Darboux families and the classification of real four-dimensional indecomposable coboundary Lie bialgebras}
\author{J. de Lucas and D. Wysocki}

\date{}
\begin{document}

\maketitle
 \begin{abstract} 
 This work introduces a new concept, the so-called {\it Darboux family}, which is employed to determine, to analyse geometrically, and to classify up to Lie algebra automorphisms, in a relatively easy manner, coboundary Lie bialgebras on real four-dimensional indecomposable Lie algebras.  The Darboux family notion can be considered as a generalisation of the Darboux polynomial for a vector field. The classification of $r$-matrices and solutions to classical Yang-Baxter equations for real four-dimensional indecomposable Lie algebras is also given in detail. Our methods can further be applied to general, even higher-dimensional, Lie algebras. As a byproduct, a method to obtain matrix representations of certain Lie algebras with a non-trivial center is developed.  
 \end{abstract}
\section{Introduction}

	Lie bialgebras \cite{CP94,Dr83, Dr87,KS04}  
	appeared as a tool to study integrable systems \cite{Fa84, Fa87}. A {\it Lie bialgebra} is a Lie algebra $\mathfrak{g}$ along with a Lie bracket on its dual space $\mathfrak{g}^*$ that amounts to a cocycle in a Chevalley-Eilenberg cohomology of $\mathfrak{g}$. Lie bialgebras have also applications to quantum gravity  \cite{BCH00, BHM13,MS11, MS03} and other research fields \cite{CP94,KS04,Op98}. In differential geometry, they occur in the problem of classifying Poisson Lie groups \cite{BBM12,CP94}.

Although much research has been devoted to the classification  of Lie bialgebras up to Lie algebra automorphisms, there are still many open problems. More specifically, Lie bialgebras on two- and three-dimensional Lie algebras $\mathfrak{g}$ have been completely classified \cite{FJ15, Go00}. Particular instances of Lie bialgebras with $\dim \mathfrak{g}>3$, e.g. for a semi-simple $\mathfrak{g}$, have also been studied \cite{ARH17,BHP99,BLT17,Ku94,LT17,Op98,Op00}. Employed techniques are rarely non-algebraic (cf. \cite{LW20}) and they are not very effective to analyse  Lie bialgebras when $\dim \mathfrak{g}>3$ \cite{ARH17, CP94,FJ15,Go00,LW20}. This motivates the search for new approaches to the study and classification of Lie bialgebras \cite{LW20}. In particular, we are here interested in geometric approaches.

A {\it coboundary Lie bialgebra} is a particular type of Lie bialgebra that is characterised by means of an $r$-{\it matrix}, namely a bivector on a Lie algebra $\mathfrak{g}$ that is a solution to its {\it modified classical Yang-Baxter equation} (mCYBE) \cite{CP94,Go00}. This work introduces novel geometric techniques to classify, up to Lie algebra automorphisms, coboundary Lie bialgebras and their $r$-matrices on a fixed Lie algebra $\mathfrak{g}$. 

In particular, we introduce a generalisation, the so-called {\it Darboux family}, of the notion of Darboux polynomial for a polynomial vector field. On a finite-dimensional vector space $E$, a polynomial function is a polynomial expression on a set of linear coordinates on $E$. A {\it Darboux polynomial}, $P$, for a polynomial vector field $X$ on $E$ is a polynomial function on $E$ so that $XP=fP$ for a certain polynomial $f$ on $E$, the so-called {\it cofactor} of $P$ relative to $X$ \cite{Gu11,LV13}. It is worth noting that, geometrically, one cannot define intrinsically what a polynomial on a general manifold is, as the explicit form of a function on a manifold  depends on the chosen coordinate system. 

More generally, let us consider a $q$-dimensional Lie algebra of vector fields, let us say $V:=\langle X_1,\ldots,X_q\rangle$, on a  manifold $M$. A {\it Darboux family} for $V$ is an $s$-dimensional vector space $\mathcal{A}:=\langle f_1,\ldots,f_s\rangle$ of smooth functions on $M$ such that $X_\alpha f_\beta=\sum_{\gamma=1}^sg_{\alpha\beta\gamma}f_\gamma$ for certain smooth functions $g_{\alpha\beta\gamma}\in C^\infty(M)$ with $\alpha=1,\ldots,q$ and $\beta,\gamma=1,\ldots,s$. In this work, 
{\it Darboux families} are studied and employed to classify up to Lie algebra automorphisms the $r$-matrices and coboundary Lie bialgebras on real four-dimensional Lie algebras that are {\it indecomposable} \cite{SW14}, i.e. they cannot be written as a direct sum of two {\it proper} ideals, namely ideals that are different of the zero and the total Lie algebra. Nevertheless, our methods can also be applied to any other Lie algebra. 

Let ${\rm Aut}(\mathfrak{g})$ stand for the Lie group of Lie algebra automorphisms of a  Lie algebra $\mathfrak{g}$. Then, ${\rm Aut}(\mathfrak{g})$ naturally acts on the space $\Lambda^2\mathfrak{g}$ of {\it bivectors} \cite{AM78} of $\mathfrak{g}$ and, more specifically, on the space $\mathcal{Y}_{\mathfrak{g}}\subset\Lambda^2\mathfrak{g}$ of solutions to  the mCYBE on $\mathfrak{g}$ (see \cite{CP94,LW20}). The classes of equivalent $r$-matrices (up to Lie algebra automorphisms of $\mathfrak{g}$) are given by the orbits of the action of ${\rm Aut}(\mathfrak{g})$ on $\mathcal{Y}_{\mathfrak{g}}$. Characterising and analysing such orbits is in general complicated even when the form of ${\rm Aut}(\mathfrak{g})$ is explicitly known (see \cite{FJ15} for a typical algebraic approach to Lie bialgebras on three-dimensional Lie algebras). Most techniques in the literature are algebraic \cite{FJ15,Go00}. Instead, we use here a more geometrical approach. Let us sketch our main ideas. We use  a simple method  (see Proposition \ref{prop:derbiv} and Remark \ref{Re:DerAlg}) to determine the Lie algebra of fundamental vector fields, $V_\mathfrak{g}$, of the action of ${\rm Aut}(\mathfrak{g})$ on $\Lambda^2\mathfrak{g}$. Our technique does not require to know the explicit form of ${\rm Aut}(\mathfrak{g})$. Instead, it is enough to know the space of derivations on $\mathfrak{g}$, which can be derived by solving a linear system of algebraic equations. We show that the orbits of the connected part of ${\rm Aut}(\mathfrak{g})$ containing its neutral element, ${\rm Aut}_c(\mathfrak{g})$, are the integral connected submanifolds of the {\it generalised distribution}, $\mathscr{E}_\mathfrak{g}$, spanned by the vector fields of  $V_{\mathfrak{g}}$ (see \cite{He62,La18,Na66,St80,Su73,Va94} for general results on generalised distributions). In fact, as $V_{\mathfrak{g}}$ is a finite-dimensional Lie algebra, $\mathscr{E}_\mathfrak{g}$ is integrable \cite{La18}. Finding the integral connected submanifolds of a generalised distribution, its so-called {\it strata}, is more complicated than finding the leaves of standard distributions because strata cannot always be determined by a family of common first-integrals for the elements of $V_{\mathfrak{g}}$, as it happens in the case of distributions (cf. \cite{AM78,Ki04}). We determine them here via our Darboux families, which is  relatively easy as illustrated by our examples. Once the strata  of $\mathscr{E}_\mathfrak{g}$ in $\mathcal{Y}_{\mathfrak{g}}$ have been obtained, the use of the action on such strata of an element of each connected part of ${\rm Aut}(\mathfrak{g})$ allows us to determine the orbits of ${\rm Aut}(\mathfrak{g})$ on $\mathcal{Y}_\mathfrak{g}$, which gives us the desired classification of $r$-matrices. Note that our method do rely on the determination of a single element of each connected component of ${\rm Aut}(\mathfrak{g})$ and, more relevantly, it does not need the explicit action of the whole ${\rm Aut}(\mathfrak{g})$ on $\Lambda^2\mathfrak{g}$. 

It is known that two $r$-matrices for $\mathfrak{g}$ that are not equivalent up to a Lie algebra automorphism may give rise to coboundary Lie bialgebras that are equivalent up to a Lie algebra automorphism (cf. \cite{CP94,LW20}). Let $(\Lambda^2\mathfrak{g})^{\mathfrak{g}}$ be the space of bivectors of $\mathfrak{g}$ that are invariant relative to the action of elements of $\mathfrak{g}$ via the algebraic Schouten bracket \cite{CP94,LW20}. In our work, we prove that the orbits of ${\rm Aut}(\mathfrak{g})$ on $\mathcal{Y}_{\mathfrak{g}}$ that project onto the same orbit of the natural action of ${\rm Aut}(\mathfrak{g})$ on  $\Lambda^2\mathfrak{g}/(\Lambda^2\mathfrak{g})^{\mathfrak{g}}$ (see \cite{LW20}) are exactly the $r$-matrices that lead to equivalent (up to Lie algebra automorphisms) coboundary Lie bialgebras on $\mathfrak{g}$. 

Our methods are computationally affordable for the study and classification of general coboundary Lie bialgebras with three- and four-dimensional $\mathfrak{g}$. Our techniques are quite probably appropriate for looking into and classifying Lie bialgebras on any five-dimensional $\mathfrak{g}$ and, possibly, for other particular instances of higher-dimensional Lie algebras $\mathfrak{g}$. Indeed, our procedures lead to the classification of real coboundary Lie bialgebras, up to Lie algebra automorphisms, on any four-dimensional indecomposable Lie algebra $\mathfrak{g}$ (see \cite{SW14} for a classification of indecomposable Lie algebras up to dimension six). Our results are summarised in Table
\ref{Tab:g_orb_1}. This is a remarkable advance relative to other classification works in the literature \cite{ARH17,FJ15,Go00}. Indeed, the most complicated classification in the literature so far is probably the classification of Lie bialgebras on symplectic four-dimensional Lie algebras \cite{ARH17}. Additionally, other four-dimensional Lie bialgebras have been partially studied in the mathematics and physics literature (see \cite{BCH00,BH96,BH97,BHOS93,BCGST92,BHP99,Ku94,Op98,Op00} and references therein). As a byproduct of our research, a method for the matrix representation of a class of finite-dimensional Lie algebras with a non-trivial center, which cannot be represented through the matrices of the adjoint representation, is given. This is interesting not only for our purposes, but also in many other works where such a representation is employed in practical calculations, e.g. \cite{GBGH19}.

The structure of the paper goes as follows. Section \ref{Se:CLA} surveys the main notions on Lie bialgebras and their derivations, $\mathfrak{g}$-modules, and the notation to be used.  In Section \ref{Sec:StSus}, the theory of generalised distributions is discussed. A method to obtain a matrix representation of a class of Lie algebras with non-trivial center is given in Section \ref{Se:MatrRep}. Section \ref{Se:DarFam} introduces Darboux families and shows how they can be used to classify $r$-matrices on the Lie algebra $\mathfrak{s}_{4,1}$ (according to \v{S}nobl and Winternitz's notation in \cite{SW14}), up to Lie algebra automorphisms thereof. Section \ref{Se:GSrMatDarFam} analyses several geometric properties of solutions to mCYBEs and it also provides some hints on the use Darboux families to study the equivalence up to Lie algebra automorphisms of $r$-matrices and coboundary Lie bialgebras. In Section \ref{Sec:Cla}, Darboux families are applied to studying and classifying coboundary Lie bialgebras on real four-dimensional indecomposable Lie algebras. Section 8 summarises our results and presents some further work in progress.

\section{Fundamentals on Lie bialgebras and their derivations}\label{Se:CLA}
 Let us provide a brief account on the notions of Lie bialgebras, Schouten brackets, and $\mathfrak{g}$-modules to be used hereafter (see \cite{CP94,KS04,Va94} for further details). Some simple results on the geometric properties of the action of ${\rm Aut}(\mathfrak{g})$ on spaces of $k$-vectors are provided. Our approach is more geometric than in standard works on Lie bialgebras. We hereafter assume that $\mathfrak{g}$ and $E$ are a finite-dimensional real Lie algebra and a finite-dimensional real vector space, respectively. Meanwhile, $GL(E)$ and $\mathfrak{gl}(E)$ stand for the Lie group of automorphisms and the Lie algebra of endomorphisms on $E$, respectively. 

Let $\mathcal{V}^m M$ be the vector space of $m$-vector fields on a manifold $M$. The {\it Schouten-Nijenhuis bracket} \cite{Ma97,Va94} on $\mathcal{V}M:=\oplus_{m\in \mathbb{Z}}\mathcal{V}^m M$ is the unique bilinear map $[\cdot, \cdot]: \mathcal{V}M \times \mathcal{V}M \to \mathcal{V}M$ satisfying that 
$$
[f,g]=0,
$$
\vskip -0.6cm
$$
[X_1\wedge\ldots\wedge X_s,f]:=\iota_{df}X_1\wedge\ldots\wedge X_s:=\sum_{i=1}^s(-1)^{i+1}  (X_if)X_1\wedge\ldots\wedge\widehat{X_i}\wedge\ldots\wedge X_s,
$$ 
\vskip -0.7cm
 \begin{equation}\label{multi_sn}
	[X_1 \wedge \ldots \wedge X_s, Y_1 \wedge \ldots \wedge Y_l]\!\! := \!\!\!\!\sum_{\substack{i=1,\ldots,s\\ j=1,\ldots,l}}\!\!\!(-1)^{i+j} [X_i, Y_j] \wedge X_1 \wedge \ldots  \wedge\widehat{X}_i \wedge \ldots\wedge  X_s \wedge Y_1\wedge \ldots \wedge\widehat{Y}_j \wedge \ldots \wedge Y_l,
	\end{equation}
\vskip -0.2cm
\noindent where $f,g\in C^\infty(M)$,  $s,l\in \mathbb{N}$, the $X_1,\ldots,X_s,Y_1,\ldots,Y_l$ are arbitrary  vector fields on $M$,     an omitted vector field $X$ is denoted by the hat symbol $\widehat{X}$, and $[X_i,Y_j]$ is the Lie bracket\footnote{We denote similar structures with the same symbol, but their meaning is clear from the context.} of $X_i$ and $Y_j$ (see \cite{Ma97} for more details). 
Remarkably, $[\mathcal{X},\mathcal{Y}]\in \mathcal{V}^{s+l-1}$ for $\mathcal{X} \in \mathcal{V}^s M$ and $\mathcal{Y} \in \mathcal{V}^l M$. 
The space, $\mathcal{V}^LG$, of left-invariant elements of $\mathcal{V}G$ for a Lie group $G$ is closed relative to $[\cdot,\cdot]$. In particular, left-invariant vector fields on $G$, i.e. ${\mathcal{V}^L}^1G$, span a finite-dimensional Lie algebra called the Lie algebra of $G$, which can be identified with $T_eG$ \cite{AM78}. Vice versa, every abstract finite-dimensional Lie algebra $\mathfrak{g}$ can be thought of as the Lie algebra of left-invariant vector fields of a Lie group \cite{DK00}. Meanwhile, $\mathcal{V}^LG$ can be identified with the Grassmann algebra $\Lambda\mathfrak{g}$, namely the algebra relative to the exterior product spanned by all the multivectors of the Lie algebra, $\mathfrak{g}$, of $G$ \cite{LGLV19}. Moreover, $[\cdot,\cdot]$ can be restricted to $\mathcal{V}^LG$ leading to the \textit{algebraic Schouten bracket} on $\Lambda\mathfrak{g}$ \cite{LW20,Va94}. For simplicity, we will call it the {\it Schouten bracket} on $\Lambda \mathfrak{g}$. Recall that the Grassmann algebra $\Lambda E$ of a vector space $E$ satisfies that $\Lambda E=\bigoplus_{m\in \mathbb{Z}}\Lambda^m E$, where $\Lambda^m E$ is the vector space of $m$-vectors of $E$.

A \textit{Lie bialgebra} is a pair $(\mathfrak{g}, \delta)$, where $\mathfrak{g}$ admits a Lie bracket $[\cdot, \cdot]_\mathfrak{g}$, whilst $\delta: \mathfrak{g} \to \Lambda^2 \mathfrak{g}$, the \textit{cocommutator}, is a linear map, its transpose 
$\delta^*: \Lambda^2\mathfrak{g}^* \to \mathfrak{g}^*$ is a Lie bracket on $\mathfrak{g}^*$, and 
\begin{equation}\label{cocycle_cond}
\delta([v_1,v_2]_\mathfrak{g}) = [v_1, \delta(v_2)] +[\delta(v_1),v_2], \qquad \forall v_1,v_2 \in \mathfrak{g}.
\end{equation}

A \textit{Lie bialgebra homomorphism} is a Lie algebra homomorphism $\phi:\mathfrak{g}\rightarrow \mathfrak{h}$ between Lie bialgebras $(\mathfrak{g},\delta_{\mathfrak{g}})$ and $(\mathfrak{h},\delta_{\mathfrak{h}})$ such that $(\phi\otimes \phi)\circ\delta_\mathfrak{g}=\delta_\mathfrak{h}\circ \phi$. A \textit{coboundary Lie bialgebra} is a Lie bialgebra $(\mathfrak{g}, \delta_r)$ such that $\delta_r(v) := [v, r] $ for every $v \in \mathfrak{g}$ and  some $r \in \Lambda^2 \mathfrak{g}$, a so-called {\it $r$-matrix}. To characterise $r$-matrices, we use the following notions and Theorem \ref{thm1}. The standard identification of an abstract Lie algebra $\mathfrak{g}$ with the Lie algebra of left-invariant vector fields on a Lie group $G$   allows us to understand the tensor algebra $\mathfrak{T}(\mathfrak{g})$ of $\mathfrak{g}$ as the tensor algebra, $\mathfrak{T}^L(G)$, of left-invariant tensor fields on $G$. This gives rise to a Lie algebra representation  $\textrm{ad}: v\in \mathfrak{g} \mapsto {\rm ad}_v\in \mathfrak{gl}\left(\mathfrak{T}(\mathfrak{g})\right)$, where ${\rm ad}_v(w):=\mathcal{L}_vw$ for every $w\in \mathfrak{T}(\mathfrak{g})$ and $\mathcal{L}_vw$ is the Lie derivative of $w$ relative to  $v$, which are understood as elements in $\mathfrak{T}^L(G)$ in the natural way. Note that the geometric notation $\mathcal{L}_vw$ is conciser than algebraic ones (cf. \cite{CP94}). An element $q \in \mathfrak{T}(\mathfrak{g})$ is called \textit{$\mathfrak{g}$-invariant} if $\mathcal{L}_vq=0$ for all $v \in \mathfrak{g}$. We write $(\mathfrak{T}(\mathfrak{g}))^{\mathfrak{g}}$ for the set of $\mathfrak{g}$-invariant elements of $\mathfrak{T}(\mathfrak{g})$. The map ${\rm ad}:\mathfrak{g}\rightarrow \mathfrak{gl}(\mathfrak{T}(\mathfrak{g}))$ admits a  restriction ${\rm ad}:\mathfrak{g}\rightarrow \mathfrak{gl}(\Lambda^m\mathfrak{g})$. We recall that  $(\Lambda^m\mathfrak{g})^{\mathfrak{g}}$ stands for the space of $\mathfrak{g}$-invariant $m$-vectors. Let us recall the following well-known result.

\begin{theorem}\label{thm1} 
	The map $\delta_r: v \in \mathfrak{g} \mapsto [v, r]  \in \Lambda^2\mathfrak{g}$, for $r\in \Lambda^2\mathfrak{g}$, is a cocommutator if and only if 
	$[r, r] \in (\Lambda^3 \mathfrak{g})^{\mathfrak{g}}$.
\end{theorem}
We call $[r, r] \in (\Lambda^3 \mathfrak{g})^{\mathfrak{g}}$ the \textit{modified classical Yang-Baxter equation} (mCYBE) of $\mathfrak{g}$, while $[r, r] = 0$ is referred to as the \textit{classical Yang-Baxter equation} (CYBE) of $\mathfrak{g}$ and its solutions amount to left-invariant Poisson bivectors on any Lie group $G$ with Lie algebra isomorphic to $\mathfrak{g}$ \cite{Va94}. 

Note that two $r$-matrices $r_1,r_2\in \Lambda^2\mathfrak{g}$ satisfy that $\delta_{r_1}=\delta_{r_2}$ if and only if  $r_1-r_2\in(\Lambda^2\mathfrak{g})^{\mathfrak{g}}$.
Then, what really matters to the determination of coboundary Lie bialgebras is not $r$-matrices, but their equivalence classes in the quotient space $\Lambda^2_R\mathfrak{g}:=\Lambda^2\mathfrak{g}/(\Lambda^2\mathfrak{g})^\mathfrak{g}$.

	A {\it $\mathfrak{g}$-module} is a pair $(V,\rho)$, where  $\rho:v\in \mathfrak{g}\mapsto \rho_v\in \mathfrak{gl}(V)$ is a Lie algebra morphism. A $\mathfrak{g}$-module $(V,\rho)$ will be represented just by $V$, while $\rho_v(x)$, for any $v\in \mathfrak{g}$ and $x\in V$, will be written simply as $vx$  if $\rho$ is understood from the context. If ${\rm ad}: v \in \mathfrak{g}\mapsto [v,\cdot]_\mathfrak{g} \in \mathfrak{gl}(\mathfrak{g})$ stands for the adjoint representation of $\mathfrak{g}$, then the fact that each $[v,\cdot]_\mathfrak{g}$, with $v\in\mathfrak{g}$, is a derivation of the Lie algebra $\mathfrak{g}$ (relative to its Lie bracket $[\cdot,\cdot]_\mathfrak{g}$) allows us to ensure that $(\mathfrak{g},{\rm ad})$ is a $\mathfrak{g}$-module  \cite{FH91}. The map ${\rm ad}$ can be considered as a mapping of the form ${\rm ad}:\mathfrak{g}\rightarrow \mathfrak{der}(\mathfrak{g})$, where $\mathfrak{der}(\mathfrak{g})$ is the Lie algebra of derivations on $\mathfrak{g}$. As a second relevant example of $\mathfrak{g}$-module, consider the Lie group ${\rm Aut}(\mathfrak{g})$ of the Lie algebra automorphisms of $\mathfrak{g}$ (see \cite{SW73} for details on its Lie group structure) and its Lie algebra, which is denoted by $\mathfrak{aut}(\mathfrak{g})$. The tangent map at the identity map on $\mathfrak{g}$, let us say ${\rm id}_\mathfrak{g}\in {\rm Aut}(\mathfrak{g})$, to the injection $\iota:{\rm Aut}(\mathfrak{g})\hookrightarrow GL(\mathfrak{g})$ induces a Lie algebra morphism $\widehat {\rm ad}:\mathfrak{aut}(\mathfrak{g})\simeq {\rm T}_{{\rm id}_\mathfrak{g}}{\rm Aut}(\mathfrak{g})\rightarrow \mathfrak{gl}(\mathfrak{g})\simeq {\rm T}_{{\rm id}_\mathfrak{g}}GL(\mathfrak{g})$ and $(\mathfrak{g},\widehat{\rm ad})$ becomes an $\mathfrak{aut}(\mathfrak{g}$)-module.

In view of the properties of the algebraic Schouten bracket, each $\mathfrak{g}$ gives rise to a $\mathfrak{g}$-module  $(\Lambda\mathfrak{g},{\rm ad})$, where ${\rm ad}:v\in \mathfrak{g}\mapsto [v,\cdot]\in \mathfrak{gl}(\Lambda \mathfrak{g})$ (cf. \cite{Va94}). This fact can be viewed as a consequence of  \cite[Proposition 2.1]{LW20}. To grasp this result and related ones,  recall that every $T\in \mathfrak{gl}(E)$ gives rise to the mappings  $\Lambda^mT:\lambda \in \Lambda^mE\mapsto 0\in\Lambda^mE$ for $m\leq 0$, and the maps $\Lambda^mT\in \mathfrak{gl}(\Lambda^mE)$, for $m>0$, given by the restriction to $\Lambda^m E$ of 
$
\Lambda^mT:={{T\otimes {\rm id}\otimes\ldots\otimes {\rm id}}}(m-{\rm operators})+\ldots+{{{\rm id}\otimes\ldots\otimes{\rm id}\otimes T}}(m-{\rm operators}),
$ where ${\rm id}$ is the identity on $E$. Moreover,  $\Lambda T:=\bigoplus_{m\in \mathbb{Z}}\Lambda^mT\in\mathfrak{gl}(\Lambda E)$. If $T$ is considered as an element of $GL(E)$, we define $\Lambda^m T:=T\otimes\ldots \otimes T(m-{\rm operators})$ for $m\geq 1$ and $\Lambda^m T$ is the identity on $\Lambda^mE$ for $m\leq 0$. Finally, $\Lambda T:=\bigoplus_{m\in \mathbb{Z}}\Lambda^mT$.

The Lie group ${\rm Aut}(\mathfrak{g})$ gives rise to a Lie group action $T\in {\rm Aut}(\mathfrak{g})\mapsto \Lambda^mT\in GL(\Lambda^m\mathfrak{g})$. Moreover, this gives rise to an infinitesimal Lie group action $\rho:\mathfrak{aut}(\mathfrak{g})\rightarrow \mathfrak{gl}(\Lambda^m\mathfrak{g})$ such that 
$$
[\rho(d)](w):=\frac{d}{dt}\bigg|_{t=0}[\Lambda^m\exp(td)](w)=(\Lambda^md)(w),\qquad \forall d\in \mathfrak{aut}(\mathfrak{g}),\qquad \forall w\in \Lambda^m\mathfrak{g}.
$$
We write $V_{\mathfrak{g}}$ for the Lie algebra of fundamental vector fields of this Lie group action for $m=2$. Recall that $\mathfrak{aut}(\mathfrak{g})$ is indeed the space of derivations on $\mathfrak{g}$ (cf. \cite{LW20}). The following result was proved in \cite{LW20}.
\begin{proposition}\label{proporb}
	The dimension of the orbit $\mathcal{O}_w$ of the action of ${\rm Inn}(\mathfrak{g})$ on $\Lambda^m\mathfrak{g}$ through $w\in \Lambda^m\mathfrak{g}$ is 
	$
	\dim  {\rm Im}\,\Theta^m_w,
	$
	where $\Theta^m_w:v\in \mathfrak{inn}(\mathfrak{g})\mapsto [v, w]  \in \Lambda^m\mathfrak{g}$.
\end{proposition}
In this work, we will use the following rather straightforward generalisation of Proposition \ref{proporb}.
\begin{proposition}\label{prop:derbiv}
	The dimension of the orbit $\mathscr{O}_w$ of the action of ${\rm Aut}(\mathfrak{g})$ on $\Lambda^m\mathfrak{g}$ through $w\in \Lambda^m\mathfrak{g}$ is 
	$
	\dim  {\rm Im}\,\Upsilon^m_w,
	$
	where $\Upsilon^m_w:d\in \mathfrak{der}(\mathfrak{g})\mapsto (\Lambda^md)(w) \in \Lambda^m\mathfrak{g}$.
\end{proposition}
\begin{proof}
	The orbit of $w\in\Lambda^m \mathfrak{g}$ relative to ${\rm Aut}(\mathfrak{g})$ is given by the points $ (\Lambda^mT)( w)$ for every $T\in {\rm Aut}(\mathfrak{g})$. Define $\exp(td ) =:T_t$, with $t\in \mathbb{R}$, for $d\in\mathfrak{der}(\mathfrak{g})$. Then, the tangent space at $w$ of $\mathscr{O}_w$  is spanned by the tangent vectors
	\begin{equation}\label{Eq:Red}
	\frac{d}{dt}\bigg|_{t=0}[\Lambda^m\exp(td)](w)=(\Lambda^md)(w)\in T_w\Lambda^m\mathfrak{g}\simeq \Lambda^m\mathfrak{g}.
	\end{equation}
	Then, $\dim \mathscr{O}_w=\dim \mathfrak{der}(\mathfrak{g})-\dim \mathfrak{g}_w$, where $G_w$ is the isotropy group of $w\in \Lambda^m\mathfrak{g}$ relative to the action of ${\rm Aut}(\mathfrak{g})$ and $\mathfrak{g}_w$ is the Lie algebra of $G_w$. Moreover, $\mathfrak{g}_w$ is given by those $d\in\mathfrak{der}(\mathfrak{g})$ such that  $(\Lambda^md)(w)=0$.
	This amounts to $d\in \ker\,\Upsilon^{m}_w$.
	Hence, $\dim \mathscr{O}_w = \dim \mathfrak{der}(\mathfrak{g})-\dim \mathfrak{g}_w=\dim {\rm Im}\, \Upsilon_{w}^{m}$.
\end{proof}
\begin{remark}\label{Re:DerAlg}
As a byproduct, Proposition \ref{prop:derbiv} shows that the fundamental vector fields of the natural Lie group action of ${\rm Aut}(\mathfrak{g})$ on $\Lambda^2 \mathfrak{g}$ are spanned by $X^2_v(w) = \sum_{i=1}^s[\Upsilon^2_w]_i \partial_{x_i}$, where $x_1, \ldots, x_s$ is any linear coordinate system on $\Lambda^2 \mathfrak{g}$, the  $w$ is any point in $\Lambda^2\mathfrak{g}$,  and $v$ belongs to   $\mathfrak{der}(\mathfrak{g})$. Moreover, the coordinates of such vector fields are the coordinates of the extensions to $\Lambda^2\mathfrak{g}$ of the derivations of $\mathfrak{g}$, which can easily be obtained as the solutions, $d$, to the linear problem
$$
d([e_1,e_2])=[d(e_1),e_2]+[e_1,d(e_2)],\qquad \forall e_1,e_2\in \mathfrak{g}.
$$
\end{remark}

\section{Generalised distributions}\label{Sec:StSus}

In this work we want to show that the problem of determining equivalent $r$-matrices up to Lie algebra automorphisms can be significantly simplified by studying the strata of the so-called {\it generalised distributions} \cite{La18,St80,Su73,Va94}. Let us detail some useful facts on these geometric entities. Unless otherwise stated, we assume all objects to be smooth and globally defined. Hereafter, $\mathfrak{X}(M)$ stands for the Lie algebra of vector fields on $M$.

A {\it generalised distribution} (also called a {\it Stefan--Sussmann distribution}) on a manifold $M$ is a correspondence $\mathcal{D}$ attaching each $x\in M$ to a subspace $\mathcal{D}_x\subset T_xM$. We call {\it rank} of $\mathcal{D}$ at $x$ the dimension of $\mathcal{D}_x$. A generalised distribution need not have the same rank at every point of $M$. If $\mathcal{D}$ has the same rank at every point of $M$, then  $\mathcal{D}$ is said to be {\it regular} or $\mathcal{D}$ is simply  called a {\it distribution}. Otherwise, $\mathcal{D}$ is said to be {\it singular}. A generalised distribution $\mathcal{D}$ on $M$ is \emph{involutive} if every two vector fields taking values in $\mathcal{D}$ satisfy that their Lie bracket takes values in $\mathcal{D}$ as well. 

A \emph{stratification}, let us say $\mathcal{F}$, on a manifold $M$ is a partition of $M$ into connected disjoint immersed submanifolds $\{\mathcal{F}_k\}_{k\in I}$, where $I$ is a certain set of indices, i.e. $M=\bigcup_{k\in I}\mathcal{F}_k$ and   the submanifolds $\{\mathcal{F}_k\}_{k\in I}$ satisfy $\mathcal{F}_{k}\cap \mathcal{F}_{k'}=\emptyset$ for $k\neq k'$ and $k,k'\in I$. The connected immersed submanifolds $\mathcal{F}_k$, with $k\in I$, are called the {\it strata} of the stratification. A stratification is \emph{regular} if its strata are immersed submanifolds of the same dimension, whilst it is \emph{singular} otherwise. Regular stratifications are called {\it foliations} and their strata are called {\it leaves}.  The tangent space to a stratum, $\mathcal{F}_k$, of a stratification passing through a point $x\in M$ is a subspace $\mathcal{D}_x\subset T_xM$. All the subspaces $\mathcal{D}_x\subset T_xM$ for every point $x\in M$ give rise to a {generalised distribution}  $\mathcal{D}:=\bigcup_{x\in M}\mathcal{D}_x$ on $M$. All the leaves of a foliation have the same dimension and, therefore, the generalised distribution formed by the tangent spaces at every point to its leaves is regular. Meanwhile, a singular stratification gives rise to a singular generalised distribution. 

In this work, we are specially interested in generalised distributions generated by finite-dimensional Lie algebras of vector fields, the so-called {\it Vessiot--Guldberg Lie algebras} \cite{LS20}.  More specifically, let $V$ be a Vessiot--Guldberg Lie algebra, the vector fields of $V$ span a generalised distribution $\mathcal{D}^V$ given by
 $$
	\mathcal{D}_x^V:=\{X_x:X\in V\}\subset T_xM,\qquad \forall x\in M.
	$$

Since the space of vector fields tangent to the strata of a stratification are closed under Lie brackets, the Lie bracket of vector fields on $M$ taking values in a distribution $\mathcal{D}$ can be restricted to each one of its strata. A generalised distribution $\mathcal{D}$ on $M$ is \emph{integrable} if there exists a stratification $\mathcal{F}$ on $M$ such that each stratum $\mathcal{F}_k$ thereof satisfies $T_x\mathcal{F}_k=\mathcal{D}_x$ for every $x\in \mathcal{F}_k$. A relevant question is whether a generalised distribution on $M$ is integrable of not. For regular distributions, the Frobenius theorem holds \cite{Fr77,La18}.

\begin{theorem} 
	If $\mathcal{D}$ is a distribution on a manifold $M$, then $\mathcal{D}$ is integrable if and only if it is involutive.
\end{theorem}

Involutivity is a natural necessary condition for a generalised distribution to be {\it integrable}  because the set of vector fields tangent to any stratum of a stratification is involutive. Nevertheless, if a generalised distribution is not regular, its involutiveness does not necessarily implies its integrability. If a generalised distribution $\mathcal{D}$ on $M$ is {\it analytical}, i.e. for every $x\in M$ there exists a family of analytical vector fields taking values in $\mathcal{D}$ and spanning $\mathcal{D}_{x'}$ for every $x'$ in an open neighbourhood of $x$, one has the following proposition. 

 \begin{theorem}\emph{\textbf{(Nagano \cite{Na66} and   \cite{La18})}}\label{theonagano}
	Let $M$ be a real analytic manifold, and let $V$ be a sub-Lie algebra of analytic vector fields on $M$. Then, the induced analytic distribution $\mathcal{D}^V$ is integrable.
\end{theorem}

If a generalised distribution is not analytical, the Stefan-Sussmann's theorem establishes additional conditions to involutivity to ensure integrability. 
We are more interested in the following result that ensures integrability in a particular case of relevance to us.

\begin{theorem}{\bf (Hermann \cite{He62} and   \cite{La18})}\label{Th:VGID} Let $M$ be a smooth manifold. If $V$ is a  finite-dimensional Lie subalgebra of $\mathfrak{X}(M)$, then the distribution $\mathcal{D}^V$ is integrable.
\end{theorem}
 
One relevant problem concerns the determination of the form of the strata of a stratification. If a stratification is regular on an open subset $U \subset M$, then one can define locally a set of functionally independent functions whose level sets are the strata of the stratification. Indeed, this amounts to the integration of a standard distribution. If a stratification is singular, the previous statement is not longer true and other methods are needed. The following facts are useful to understand further parts of this work. First, given a point $x\in M$ and a Vessiot--Guldberg Lie algebra $V$ on $M$, the stratum $\mathcal{F}_k$ of the generalised distribution $\mathcal{D}^V$ passing through $x$ is given by the points of the form \cite{St80}
\begin{equation}
\label{eq:dec}
\exp(X_1)\circ \exp(X_2)\circ \ldots \circ \exp(X_s)x,
\end{equation}
where $s$ is any natural number, $X_1,\ldots,X_s$ are any vector fields of  $V$, and each $\exp(X)$ stands for the local diffeomorphism on $M$ induced by the vector field $X$ on $M$.  Let $f$ be a function on $M$ such that if $X\in V$, then there exists a function $f_X\in C^\infty(M)$ such that $Xf=f_Xf$. This can be considered as a not necessarily polynomial analogue of a Darboux function for a vector field. Let us assume that $f=0$ at a point of $\mathcal{F}_k$. It immediately follows from (\ref{eq:dec}) that $f$ vanishes on the whole $\mathcal{F}_k$. Hence, two points $x_1,x_2\in M$ such that $f(x_1)\neq f(x_2)$ cannot belong to the same $\mathcal{F}_k$. As shown in following sections, $f$ does not need to be a constant of motion of the vector fields in $V$. This is specially relevant when determining the strata of an integrable singular generalised foliation, as their strata cannot always be determined locally as the zeroes of a family of common first-integrals of the vector fields of $V$. Meanwhile, this can  always be achieved locally for the leaves of an integrable distribution \cite{AM78,Va94}.

\section{Obtaining a matrix representation for a Lie algebra with nontrival center}\label{Se:MatrRep}

If a Lie algebra $\mathfrak{g}$ has a nontrivial center, then the image of the map ${\rm ad}: \mathfrak{g} \ni v \mapsto [v, \cdot] \in \mathfrak{gl}(\mathfrak{g})$ is a matrix Lie algebra that is not isomorphic to $\mathfrak{g}$. Hence, the image of {\rm ad} does not give a matrix algebra representation of $\mathfrak{g}$, which may give rise to a problem as matrix Lie algebra representation are quite practical in computations \cite{GBGH19}. The aim of this section is to provide a method to obtain a matrix Lie algebra isomorphic to $\mathfrak{g}$ when its center, $\mathcal{Z}(\mathfrak{g})$, is not trivial, namely $\mathcal{Z}(\mathfrak{g})\neq 0$, and $\mathfrak{g}$ satisfy some additional conditions. Our method is to be employed in the rest of our work during calculations.

Let $\{e_1,\ldots, e_s, e_{s+1}, \ldots, e_n\}$ be a basis of $\mathfrak{g}$ such that $\{e_1,\ldots, e_s\}$ form a basis for  $\mathcal{Z}(\mathfrak{g})$. Let $c_{ij}^k$, with $i,j,k=1,\ldots,n$, be the structure constants of $\mathfrak{g}$ in the given basis, i.e. $[e_i,e_j]=\sum_{k=1}^nc_{ij}^ke_k$ for $i,j=1,\ldots,n$. Define then $\tilde{\mathfrak{g}} := \langle e_1, \ldots, e_n, e\rangle$. Let us determine the conditions ensuring that $\tilde{\mathfrak{g}}$ is a Lie algebra relative to a Lie bracket, $[\cdot,\cdot]_{\bar{\mathfrak{g}}}$, that is an extension of the Lie bracket on $\mathfrak{g}$, i.e. $[e',e'']_{\bar{\mathfrak{g}}}=[e',e'']_{\mathfrak{g}}$ for every $e',e''\in \mathfrak{g}$, and $[e, e_i]_{\bar{\mathfrak{g}}} = \alpha_i e_i$ for $i=1,\ldots,n$ so that 
\begin{equation}\label{cond1}
\alpha_1 \cdot \alpha_2 \cdot \ldots \cdot \alpha_s \neq 0.
\end{equation}
The meaning of the last condition will be clear in a while. If  $[\cdot,\cdot]_{\bar{\mathfrak{g}}}$ is to be a Lie bracket, its Jacobi identity leads to the following conditions 
\begin{equation}\label{cond2}
\begin{gathered}
[e, [e_i, e_j]_{\bar{\mathfrak{g}}}]_{\bar{\mathfrak{g}}} = [[e,e_i]_{\bar{\mathfrak{g}}},e_j]_{\bar{\mathfrak{g}}} + [e_i, [e,e_j]_{\bar{\mathfrak{g}}}]_{\bar{\mathfrak{g}}} \iff \sum_{k=1}^n\alpha_k c_{ij}^ke_k =\sum_{k=1}^n (\alpha_i + \alpha_j) c_{ij}^ke_k,\qquad \forall i,j=1,\ldots,n, \\
[e, [e, e_j]_{\bar{\mathfrak{g}}}]_{\bar{\mathfrak{g}}} = [[e,e]_{\bar{\mathfrak{g}}},e_j]_{\bar{\mathfrak{g}}} + [e, [e,e_j]_{\bar{\mathfrak{g}}}]_{\bar{\mathfrak{g}}}=[e, [e,e_j]_{\bar{\mathfrak{g}}}]_{\bar{\mathfrak{g}}},\qquad \forall j=1,\ldots,n.
\end{gathered}
\end{equation}
Thus, previous conditions can be reduced to requiring that, for those indices $i,j,k$ satisfying $c_{ij}^k \neq 0$, one gets 
$$
\alpha_i + \alpha_j = \alpha_k.
$$
If the latter condition is satisfied, a new Lie algebra $\tilde{\mathfrak{g}}$ arises. Note that $\mathfrak{g}$ is an ideal of $\tilde{\mathfrak{g}}$. This leads to a Lie algebra morphism $\mathscr{R}:v\in {\mathfrak{g}}\mapsto \mathscr{R}_v:=[v,\cdot]_{\tilde{\mathfrak{g}}}\in \mathfrak{gl}(\tilde{\mathfrak{g}})$. If $\mathscr{R}_v=0$, then $[v,e_i]_{\tilde{\mathfrak{g}}}=0$ for $i=1,\ldots,n$, which means that $v\in \mathcal{Z}(\mathfrak{g})$. Thus, $v = \sum_{i=1}^s \lambda_i e_i$ for certain constants $\lambda_1,\ldots,\lambda_s$ and $0 = [e,v]_{\tilde{\mathfrak{g}}} = \sum_{i=1}^s \lambda_i \alpha_i e_i$. Then, condition (\ref{cond1}) yields $\lambda_1 = \ldots =\lambda_s = 0$ and it turns out that $v=0$.
Hence, the elements $\mathscr{R}_v$, with $v\in \mathfrak{g}$, span a matrix Lie algebra of endomorphisms on $\tilde{\mathfrak{g}}$ isomorphic to $\mathfrak{g}$. 

It is clear that conditions (\ref{cond1}) and (\ref{cond2}) do not need to be satisfied for a general Lie algebra $\mathfrak{g}$ with a nontrivial center. 
Nevertheless, we used in this work that this method works for all indecomposable real four-dimensional Lie algebras with non-trivial center (see \cite{SW14} for a complete list of such Lie algebras). This will be enough for our purposes. 

Let us illustrate our method. Consider the Lie algebra $\mathfrak{s}_{ 1} = \langle e_1, e_2, e_3, e_4\rangle$ with non-vanishing commutation relations $[e_4, e_2] = e_1, [e_4, e_3] = e_3$. Let us construct a new Lie algebra $\tilde{\mathfrak{g}}=\langle e_1,e_2, e_3,e_4,e\rangle$ following (\ref{cond1}) and (\ref{cond2}). This gives rise to the following system of equations:
$$
\alpha_4 + \alpha_2 = \alpha_1, \quad \alpha_4 + \alpha_3 = \alpha_3, \quad \alpha_1 \neq 0.
$$
We have that $\mathcal{Z}(\mathfrak{s}_{1})=\langle e_1\rangle$. The previous system, under the corresponding restriction (\ref{cond1}),  has a solution $\alpha_4 =0, \alpha_3 \in \mathbb{R}, \alpha_2 = \alpha_1, \alpha_1 \in \mathbb{R}\backslash\{0\}$. In particular, set $\alpha_2=\alpha_1 = 1$ and $\alpha_3 = 0$. Thus, the endomorphisms $\mathscr{R}_{e_i}$ on $\tilde{\mathfrak{g}}$ read
{\footnotesize
$$
\mathscr{R}_{e_1} = \left(
\begingroup
\setlength\arraycolsep{4pt}
\begin{array}{ccccc}
0 & 0 & 0 & 0 & -1 \\
0 & 0 & 0 & 0 & 0 \\
0 & 0 & 0 & 0 & 0 \\
0 & 0 & 0 & 0 & 0 \\
0 & 0 & 0 & 0 & 0 
\end{array}
\endgroup
\right),
\mathscr{R}_{e_2} = \left(
\begingroup
\setlength\arraycolsep{4pt}
\begin{array}{ccccc}
0 & 0 & 0 & -1 & 0 \\
0 & 0 & 0 & 0 & -1 \\
0 & 0 & 0 & 0 & 0 \\
0 & 0 & 0 & 0 & 0 \\
0 & 0 & 0 & 0 & 0 
\end{array}
\endgroup
\right),
\mathscr{R}_{e_3} = \left(
\begingroup
\setlength\arraycolsep{4pt}
\begin{array}{ccccc}
0 & 0 & 0 & 0 & 0 \\
0 & 0 & 0 & 0 & 0 \\
0 & 0 & 0 & -1 & 0 \\
0 & 0 & 0 & 0 & 0 \\
0 & 0 & 0 & 0 & 0 
\end{array}
\endgroup
\right),
\mathscr{R}_{e_4} = \left(
\begingroup
\setlength\arraycolsep{4pt}
\begin{array}{ccccc}
0 & 1 & 0 & 0 & 0 \\
0 & 0 & 0 & 0 & 0 \\
0 & 0 & 1 & 0 & 0 \\
0 & 0 & 0 & 0 & 0 \\
0 & 0 & 0 & 0 & 0 
\end{array}
\endgroup
\right).
$$
}
As desired, the above matrices have the same non-vanishing commutation relations as $\mathfrak{s}_{ 1}$, i.e.
$$
[\mathscr{R}_{e_4},\mathscr{R}_{e_2}]=\mathscr{R}_{e_1},\qquad [\mathscr{R}_{e_4},\mathscr{R}_{e_3}]=\mathscr{R}_{e_3}.
$$

Let us prove that not every Lie algebra $\mathfrak{g}$ with non-trivial center admits an extended Lie algebra $\tilde{\mathfrak{g}}$ of the form required by our method. Consider the Lie algebra $\mathfrak{s}_{6,231} = \langle e_1, \ldots, e_6\rangle$ with nonzero commutation relations (see \cite{SW14} for details):
$$
[e_2, e_3] = e_1, \quad [e_5, e_1] = e_1, \quad [e_5, e_2] = e_2, \quad [e_6, e_1] = e_1, \quad [e_6, e_3] = e_3, \quad [e_6, e_5] = e_4.
$$
Thus, $\mathcal{Z}(\mathfrak{s}_{4,231}) = \langle e_4\rangle$. The corresponding system of equations (\ref{cond2}) reads
$$
\alpha_2 + \alpha_3 = \alpha_1, \quad \alpha_5 + \alpha_1 = \alpha_1, \quad \alpha_5 + \alpha_2 = \alpha_2, \quad \alpha_6 + \alpha_1 = \alpha_1, \quad \alpha_6 + \alpha_3 = \alpha_3, \quad \alpha_6 + \alpha_5 = \alpha_4.
$$
Then, $\alpha_5 = 0, \alpha_6 = 0$ and $\alpha_4 = \alpha_5=0$, which contradicts the assumption of our method (\ref{cond2}), namely $\alpha_4 \neq 0$. 

Consider the 7-dimensional Lie algebra $\mathfrak{g} := \langle e_1, \ldots, e_7\rangle$ with $\mathcal{Z}(\mathfrak{g}) = \langle e_7 \rangle$ and nonzero commutation relations (cf. \cite[p. 492, case $7_I$]{Se93})
\begin{equation*}
\begin{gathered}
[e_1, e_2] = e_3, \quad [e_1, e_3] = e_4, \quad [e_1, e_4] = e_5, \quad [e_1, e_6] = e_7, \quad [e_2, e_3] = e_6,\\
[e_2, e_4] = e_7, \quad [e_2, e_5] = e_7, \quad [e_2, e_6] = e_7, \quad [e_3, e_4] = -e_7.
\end{gathered}
\end{equation*}
The corresponding system (\ref{cond2}) reads
$$
\begin{array}{rlrlrl}
{\rm (i)} & \alpha_1 + \alpha_2 = \alpha_3, & {\rm (ii)} & \alpha_1 + \alpha_3 = \alpha_4, & {\rm (iii)} & \alpha_1 + \alpha_4 = \alpha_5, \\
{\rm (iv)} & \alpha_1 + \alpha_6 = \alpha_7, & {\rm (v)} & \alpha_2 + \alpha_3 = \alpha_6, & {\rm (vi)} & \alpha_2 + \alpha_4 = \alpha_7, \\
{\rm (vii)} & \alpha_2 + \alpha_5 = \alpha_7, & {\rm (viii)} & \alpha_2 + \alpha_6= \alpha_7, & {\rm (ix)} & \alpha_3 + \alpha_4 = \alpha_7.
\end{array}
$$
From (vii) and (viii), we  get $\alpha_5 = \alpha_6$; from (vi) and (ix), we get $\alpha_2 = \alpha_3$; and from (iv) and (viii), we get $\alpha_1 = \alpha_2$. Thus, $\alpha_1 = \alpha_2 = \alpha_3$. From (i), it follows that $\alpha_1 = 0$. Thus, $\alpha_2 =  \alpha_3 = 0$. From (ii), we get $\alpha_4 = 0$ and (iii) yields $\alpha_5 = 0$. From (iv), one gets $\alpha_6 = \alpha_7$. From (ix), we get $\alpha_7  =0$, and thus $\alpha_6 = 0$. Therefore, $\alpha_1 = \ldots = \alpha_7 = 0$ is the only possibility, but it does not satisfy (\ref{cond1}).

\section{Darboux families and the Lie algebra $\mathfrak{s}_{1}$}\label{Se:DarFam}

Let us classify solutions to the mCYBE for the Lie algebra $\mathfrak{s}_{1}$ up to its Lie algebra automorphisms. As a byproduct, we introduce a generalisation of the concept of Darboux polynomial of a vector field.

The space of derivations on $\mathfrak{s}_1$ can straightforwardly be obtained. It is indeed the solution of a linear algebra problem that can be easily solved by hand calculation and/or via any mathematical computation program. The same could be achieved for any four, five, or even some other higher-dimensional Lie algebra. In particular, derivations on $\mathfrak{s}_1$ take the form
\begin{equation}\label{Eq:Der1}
\mathfrak{der}(\mathfrak{s}_1)=\left\{\left(
\begin{array}{cccc}
\mu_{11} & \mu_{12} & 0 & \mu_{14} \\
0 & \mu_{11} & 0 & \mu_{24} \\
0 & 0 & \mu_{33} & \mu_{34} \\
0 & 0 & 0 & 0
\end{array}
\right):  \mu_{11}, \mu_{12}, \mu_{33}, \mu_{14}, \mu_{24}, \mu_{34} \in \mathbb{R}\right\}
\end{equation}
in the basis $\{e_{1},e_{2},e_{3},e_{4}\}$ appearing in Table \ref{Tab:StruCons}. 
By  Proposition \ref{prop:derbiv} and, more specifically, Remark \ref{Re:DerAlg}, the Lie algebra $V_{\mathfrak{s}_1}$ of fundamental vector fields of the natural Lie group action of ${\rm Aut}(\mathfrak{s}_{1})$ on $\Lambda^2\mathfrak{s}_{1}$ is spanned by the basis (over the reals)
\begin{equation}\label{fundv_s1}
\begin{gathered}
X_1=2x_1\partial_{x_1}+x_2\partial_{x_2}+x_3\partial_{x_3}+x_4\partial_{x_4}+x_5\partial_{x_5},\qquad X_2=x_4\partial_{x_2}+x_5\partial_{x_3},\qquad X_3= -x_5\partial_{x_1}-x_6\partial_{x_2},\\
X_4=x_3\partial_{x_1}-x_6\partial_{x_4},\qquad X_5=x_2\partial_{x_2}+x_4\partial_{x_4}+x_6\partial_{x_6},\qquad X_6=x_3\partial_{x_2}+x_5\partial_{x_4}.
\end{gathered}
\end{equation}
As $X_1,\ldots,X_6$ close on a finite-dimensional Lie algebra,  Theorem \ref{Th:VGID} shows that they span an integrable generalised distribution $\mathcal{D}^{V_{\mathfrak{s}_1}}$ on $\Lambda^2\mathfrak{s}_{1}$.  In view of (\ref{eq:dec}) the strata of $\mathcal{D}^{V_{\mathfrak{s}_1}}$ are the orbits of the action of the connected part of the identity of ${\rm Aut}(\mathfrak{g})$, let us say ${\rm Aut}_c(\mathfrak{s}_{1})$, on $\Lambda^2\mathfrak{s}_{1}$.

We define $M(p)$ to be a matrix whose entry $(i,j)$ is the $j$-coefficient of $X_i$  at $p\in \Lambda^2\mathfrak{s}_1$ in the basis $\partial_{x_1}|_p,\ldots,\partial_{x_6}|_p$, namely
$$
M(p):=\left(
\begin{array}{cccccc}
2x_1 & x_2 & x_3 & x_4 & x_5 & 0 \\
0 & x_4 & x_5 & 0 & 0 & 0 \\
-x_5 & -x_6 & 0 & 0 & 0 & 0 \\
x_3 & 0 & 0 & -x_6 & 0 & 0 \\
0 & x_2 & 0 & x_4 & 0 & x_6 \\
0 & x_3 & 0 & x_5 & 0 & 0 
\end{array}
\right),\qquad p:=(x_1,\ldots,x_6)\in \Lambda^2\mathfrak{s}_{1}.
$$
The rank of  $\mathcal{D}^{V_{\mathfrak{s}_1}}$ at $p\in \Lambda^2\mathfrak{s}_{1}$ is equal to the rank of $M(p)$. It is simple to calculate that the mCYBE  for $\mathfrak{s}_{1}$ is given by the common zeroes of the functions on $\Lambda^2\mathfrak{s}_1$ of the form
\begin{equation}\label{mCYBEs}
f_1:=x_3 x_4, \quad f_2:=x_3 x_6, \quad f_3:=x_5^2.
\end{equation}
It is immediate that the space of solutions to the mCYBE, let us say $\mathcal{Y}_{\mathfrak{s}_1}$, is not a submanifold in $\Lambda^2\mathfrak{s}_{1}$. 
Nevertheless, 
\begin{equation}\label{tangency}
X_if_j=\sum_{k=1}^3\lambda_{ij}^kf_k,\qquad i=1,\ldots,6,\quad
j=1,2,3,
\end{equation}
for certain constants $\lambda_{ij}^k$, with $i=1,\ldots,6$ and $j,k=1,2,3$. Relations (\ref{tangency}) show that the integral curves of any vector field in $V_{\mathfrak{s}_1}$  passing through $w\in \mathcal{Y}_{\mathfrak{s}_1}$ is contained in $\mathcal{Y}_{\mathfrak{s}_1}$. In fact, the derivative of the functions $f_1,f_2,f_3$ along an integral curve, $\Psi(t)$, of a vector field $X_j\in V_{\mathfrak{s}_1}$ such that $\Psi(0)\in \mathcal{Y}_{\mathfrak{s}_1}$ reads
$$
\frac{df_i(\Psi(t))}{dt}=(X_j f_i)(\Psi(t))=\sum_{k=1}^3\lambda_{ji}^kf_k(\Psi(t)),\qquad i=1,2,3.
$$
Hence, the values of the $f_i(\Psi(t))$, with $i=1,2,3$, can be understood as the solutions to a linear system of first-order ordinary differential equations in normal form with constant coefficients and zero initial condition since $f_i(\Psi(0))=0$ for $i=1,2,3$. Hence, $f_1,f_2,f_3$ vanish along $\Psi(t)$ and, since  (\ref{eq:dec}) shows that the integral curves for all the vector fields in $V$ connect all the points within the same strata of $\mathcal{D}^{V_{\mathfrak{s}_1}}$, one obtains that the functions $f_1,f_2,f_3$ are zero on any strata of $\mathcal{D}^{V_{\mathfrak{s}_1}}$ containing a point where $f_1,f_2,f_3$ vanish. 
It is worth noting that $f_1,f_2,f_3$ are not constants of motion common to all the vector fields of $V_{\mathfrak{s}_1}$. 
We hereafter call $\langle f_1,f_2,f_3\rangle$ a {\it Darboux family for the Lie algebra of vector fields} $\langle X_1,\ldots,X_6\rangle $ on $\Lambda^2\mathfrak{s}_{1}$. More generally, we propose the following definition and we extend previous results to a more general realm. 

\begin{definition}
We call an $s$-dimensional linear space of functions $\mathcal{A}:=\langle f_1,\ldots,f_s\rangle $ on a manifold $M$ a {\it Darboux family} for a Vessiot--Guldberg Lie algebra $V$ on $M$ if, for every $X\in V$ and $f_j$, with $j=1,\ldots,s$, one can write
$$
Xf_j=\sum_{i=1}^sh_{j X}^i f_i,
$$
for a  certain family of smooth functions,  $h_{jX}^i$, with $i=1,\ldots,s$, on $M$, the so-called {\it co-factors} of $f_j$ relative to $X$ and the basis $\{f_1,\ldots,f_s\}$. The subset $\mathcal{S}_{\mathcal{A}}:= \{p \in M: f(p)=0, \forall f \in \mathcal{A}\}$ is called the {\it locus} of the Darboux family $\mathcal{A}$.
\end{definition}

It is worth stressing that we require the functions $f_1,\ldots,f_s$ and $h^i_{jX}$, for $i,j=1,\ldots,s$  and $X\in V$, in the above definition to be smooth. 
If all the functions $h_{jX}^i$ are equal to zero for every $X \in V$ and $i = 1, \ldots, s$, then $f_j $ becomes a constant of motion for the vector fields of $V$ on $M$. If the $h_{jX}^i$ are constants for every $X \in V$ and $i,j = 1, \ldots, s$, then we say that the Darboux family is {\it linear}. In this case, $V$ gives rise to a Lie algebra representation $\rho:X\in V\mapsto D_X\in {\rm End}(\mathcal{A})$, where $D_X$ stands for the action of the vector field $X$ on the space of functions $\mathcal{A}$. 

The vector fields of $V$ span a generalised distribution, which is integrable and leads  to a stratification of $M$ by some disjoint immersed submanifolds which may have different dimensions. Darboux families are interesting to us because the integral curves of vector fields in $V$ passing through the set $\mathcal{S}_\mathcal{A}$ of common zeroes for the elements of  $\mathcal{A}$ remain within it. Let us prove and analyse this fact, which represents a rather simple generalisation of the argument given for $\mathfrak{s}_1$. Consider a point in the locus $\mathcal{S}_{\mathcal{A}}$ and consider a basis $Y_1,\ldots,Y_q$ of the Lie algebra $V$ of vector fields. Consider an integral curve $\Psi(t)$ of a vector field of $V$ passing through a point   $\Psi(0)\in \mathcal{S}_\mathcal{A}$. Then, the time derivative  of the functions $f_1(\Psi(t)),\ldots,f_s(\Psi(t))$ is 
$$
\frac{df_i(\Psi(t))}{dt}=(Xf_i)\Psi(t)=\sum_{i=1}^sh_{iX}^j(\Psi(t))f_j(\Psi(t)),\qquad i=1,\ldots,s.
$$
Hence, the values of the $f_i(\Psi(t))$ can be understood as the solutions to the linear system of first-order ordinary differential equations in normal form with $t$-dependent coefficients 
$$
\frac{du_i}{dt}=\sum_{j=1}^sh^j_{iX}(\Psi(t))u_j,\qquad i=1,\ldots,s,
$$
with initial condition $u_1=\ldots=u_s=0$. By the existence and uniqueness theorem, the solution to the previous Cauchy problem is $u_1(t)=\ldots=u_s(t)=0$. Therefore, the functions $f_1,\ldots,f_s$ vanish on the integral curves of $X$. From this and the decomposition (\ref{eq:dec}), we can infer that the functions $f_1,\ldots,f_s$ vanish on the strata of the distribution $\mathcal{D}^V$ containing a point within $\mathcal{S}_\mathcal{A}$. In other words, the strata of $\mathcal{D}^V$ containing some point of $\mathcal{S}_A$ are fully contained in $\mathcal{S}_A$.

In view of the above, Darboux families can be used to reduce the determination of the strata in $M$ of the generalised distribution $\mathcal{D}^V$ to determining the strata within  ${\mathcal{S}_\mathcal{A}}$ and out of  ${\mathcal{S}_\mathcal{A}}$. This will be specially interesting to obtain the strata of $\mathcal{D}^V$ at points where the generalised distribution is not regular and it may not exist a constant of motion common to all the vector fields in $V$ that can be used to obtain the strata of $\mathcal{D}^V$. Note also that if $\langle f_1,\ldots,f_s\rangle $ form a Darboux family for $V$, then the set of common zeroes of $\langle f_1,\ldots,f_s\rangle $ must be the sum (as subsets of $M$) of a collection of strata of the generalised distribution $\mathcal{D}^V$.

Let us study how to use Darboux families to study solutions to  mCYBEs and CYBEs. Let us start by a general result.

 \begin{theorem}\label{Th:ConsDar32} Let $\mathcal{A}^{(3)}_{\mathfrak{g}}$ be a Darboux family for $V_{\mathfrak{g}}^{(3)}$ on $\Lambda^3\mathfrak{g}$. Then, the space of functions
 \begin{equation}\label{Eq:Ind}
 \mathcal{A}: = \left\{g \in C^{\infty}(\Lambda^2 \mathfrak{g}): \, g = f \circ [\cdot, \cdot], \, f\in \mathcal{A}_{\mathfrak{g}}^{(3)} \right\}
 \end{equation}
 is a Darboux family for $V_{\mathfrak{g}}$ on $\Lambda^2\mathfrak{g}$. If $\mathcal{A}_{\mathfrak{g}}^{(3)}$ is a linear Darboux family, then $\mathcal{A}$ is also a linear Darboux family.
 \end{theorem}
 \begin{proof}
Since each $X^{(2)} \in V_{\mathfrak{g}}$ is a fundamental field of the action of ${\rm Aut}(\mathfrak{g})$ on $\Lambda^2 \mathfrak{g}$, its flow is given by the one-parametric group of diffeomorphisms $\Lambda^2 T_t$ induced by a certain one-parametric group of Lie algebra automorphisms $T_t \in {\rm Aut}(\mathfrak{g})$, $t \in \mathbb{R}$. Every $g \in \mathcal{A}$ is of the form $g(r) = f([r, r])$ for some $f \in \mathcal{A}_{\mathfrak{g}}^{(3)}$ and every $r\in \Lambda^2\mathfrak{g}$. Then,
 \begin{multline*}
 (X^{(2)}g)(r) = \frac{d}{dt}\bigg\vert_{t = 0} g(\Lambda^2 T_t (r)) = \frac{d}{dt}\bigg\vert_{t = 0} f([\Lambda^2 T_t (r), \Lambda^2 T_t (r)]) \\= \frac{d}{dt}\bigg\vert_{t = 0} f(\Lambda^3 T_t [r, r]) = (X^{(3)} f)([r, r]) = \sum_{i = 1}^{r} h^i([r,r])f_i([r, r]),
 \end{multline*}
 where $f_1, \ldots, f_s$ form a basis of the Darboux family $\mathcal{A}_{\mathfrak{g}}^{(3)}$, the functions $h^1,\ldots,h^s$ are the cofactors of $f$ relative to $X^{(3)}$ and the basis $\{f_1,\ldots,f_s\}$, and $X^{(3)}$ is the fundamental vector field of the action of ${\rm Aut}(\mathfrak{g})$ on $\Lambda^3\mathfrak{g}$ induced by the one-parameter group $\{T_t\}_{t\in \mathbb{R}}$ of Lie algebra automorphisms of $\mathfrak{g}$. Since $f_1([r, r]),\ldots,f_s([r,r])$ belong to $\mathcal{A}$ and $h^1([r,r]),\ldots,h^s([r,r])$ are functions on $\Lambda^2\mathfrak{g}$, the $\mathcal{A}$ forms a Darboux family for $V_{\mathfrak{g}}$ on $\Lambda^2 \mathfrak{g}$.
 
 If $\mathcal{A}^{(3)}_{\mathfrak{g}}$ is a linear Darboux family, then $h^1,\ldots,h^s$ are constants. Therefore, the cofactors $h^1([r,r]),\ldots,h^s([r,r])$ are also constants, which makes $\mathcal{A}$ into a linear Darboux family for $V_{\mathfrak{g}}$ on $\Lambda^2\mathfrak{g}$.
 \end{proof}
 
 In short, the following proposition shows that the solutions to the mCYBE on $\mathfrak{g}$ can be considered as the locus of a Darboux family relative to $V_{\mathfrak{g}}$. A similar result can be applied to the CYBE on any Lie algebra $\mathfrak{g}$. 

 \begin{proposition}  Let $((\Lambda^3\mathfrak{g})^\mathfrak{g})^\circ$ be the annihilator of $(\Lambda^3\mathfrak{g})^\mathfrak{g}$, i.e. the subspace of elements of $(\Lambda^3\mathfrak{g})^*$ vanishing on $(\Lambda^3\mathfrak{g})^\mathfrak{g}$.
The functions
\begin{equation}\label{Eq:ParFam}
f_\upsilon:r\in \Lambda^2\mathfrak{g}\mapsto \upsilon([r,r])\in \mathbb{R},\qquad \upsilon\in ((\Lambda^3\mathfrak{g})^\mathfrak{g})^\circ,
\end{equation}
span a linear Darboux family, $\mathcal{A}$, for the Lie algebra $V_\mathfrak{g}$ of fundamental vector fields of the action of ${\rm Aut}(\mathfrak{g})$ on $\Lambda^2\mathfrak{g}$. Its locus is the set of solutions to the mCYBE for $\mathfrak{g}$. Moreover, the components of the CYBE in a coordinate system given by a basis of $\Lambda^2\mathfrak{g}^*$ span also a linear Darboux family for $V_{\mathfrak{g}}$ on $\Lambda^2\mathfrak{g}$.  
\end{proposition}
\begin{proof} Let us prove that the annihilator of $(\Lambda^3\mathfrak{g})^{\mathfrak{g}}$ is a Darboux family for $V_{\mathfrak{g}}^{(3)}$ on $\Lambda^3\mathfrak{g}$. In fact, if $X^{(3)} \in V_{\mathfrak{g}}^{(3)}$, then its flow is given by a one-parameter group of diffeomorphisms of the form $\Lambda^3 T_t$ for certain Lie algebra automorphisms $T_t\in {\rm Aut}(\mathfrak{g})$ with $t\in \mathbb{R}$. Consider a function $f \in \Lambda^3 \mathfrak{g}^*$ such that $f(w) = 0$ for every $w \in (\Lambda^3\mathfrak{g})^{\mathfrak{g}}$, i.e. $f\in ((\Lambda^3\mathfrak{g})^\mathfrak{g})^\circ$. Then,
 $$
 (X^{(3)}f)(w) = \frac{d}{dt}\bigg\vert_{t = 0} f((\Lambda^3 T_t)( w)) = \frac{d}{dt}\bigg\vert_{t = 0} f\left(w_t \right) =
 0,
 $$
where we have used  that $(\Lambda^3\mathfrak{g})^{\mathfrak{g}}$ is closed under the extension to $\Lambda^3\mathfrak{g}$ of Lie algebra automorphisms of $\mathfrak{g}$ and thus $(\Lambda^3T_t)(w) = :w_t\in (\Lambda^3\mathfrak{g})^{\mathfrak{g}}$. Hence, $X^{(3)}f$ vanishes on $(\Lambda^3\mathfrak{g})^{\mathfrak{g}}$. Since $X^{(3)}$ is a linear vector field on $\Lambda^3\mathfrak{g}$ and $f$ is a linear function on $\Lambda^3\mathfrak{g}$, one obtains that $X^{(3)}f$ is a linear function on $\Lambda^3\mathfrak{g}$ vanishing on $(\Lambda^3\mathfrak{g})^{\mathfrak{g}}$. Hence, $X^{(3)}f$ is a linear combination of elements of a basis of $((\Lambda^3\mathfrak{g})^{\mathfrak{g}})^\circ$ and $((\Lambda^3\mathfrak{g})^{\mathfrak{g}})^\circ$ becomes a linear Darboux family of $V_{\mathfrak{g}}^{(3)}$ on $\Lambda^3\mathfrak{g}$. By Theorem \ref{Th:ConsDar32}, the space of functions (\ref{Eq:ParFam}) 
becomes a linear Darboux family for $V_{\mathfrak{g}}$ on $\Lambda^2\mathfrak{g}$. The locus of (\ref{Eq:ParFam}) is the space of solutions of the mCYBE for $\mathfrak{g}$. 

Since every $X^{(3)}\in V^{(3)}_{\mathfrak{g}}$ is linear vector field on a linear coordinate system on $\Lambda^3 \mathfrak{g}$, the space $\Lambda^3\mathfrak{g}^*$ is also a Darboux family for $V_{\mathfrak{g}}^{(3)}$ on $\Lambda^3\mathfrak{g}$. By Theorem \ref{Th:ConsDar32}, one has that the induced  by (\ref{Eq:Ind})  family of functions on $\Lambda^2\mathfrak{g}$ is a linear Darboux family relative to $V_{\mathfrak{g}}$ on $\Lambda^2\mathfrak{g}$. The Darboux family is indeed spanned by the components of the CYBE of $\mathfrak{g}$, i.e. $[r,r]=0$, and its locus is its set of solutions.
\end{proof}

A natural question is how to determine Darboux families. Here, we provide several results about them that will be useful to as so as to derive new Darboux families from known ones.

\begin{proposition} If $\mathcal{A}_1$ and $\mathcal{A}_2$ are Darboux families for the same Vessiot--Guldberg Lie algebra $V$ on a manifold $M$, then the sum $\mathcal{A}_1+\mathcal{A}_2$ is a Darboux family for $V$ on $M$ as well.
\end{proposition} 

The proof of the previous result is immediate. Obviously, the locus of $\mathcal{A}_1+\mathcal{A}_2$ is contained in the locus of $\mathcal{A}_1$ and $\mathcal{A}_2$. This will be interesting in next sections to distinguish between  different strata of the generalised distribution spanned by the vector fields of $V_{\mathfrak{g}}$ on $\mathcal{Y}_{\mathfrak{g}}$. 

In next sections, it will be relevant to study the linear Darboux families for a certain $V_{\mathfrak{g}}$ given by a one-dimensional $\mathcal{A}\subset (\Lambda^2\mathfrak{g})^*$. This is due to the fact that, as seen in Section \ref{Sec:Cla}, the sum of such Darboux families as vector spaces will generate new useful Darboux families.  We will call such a linear one-dimensional Darboux family a {\it brick}. Bricks are one-dimensional representations of $V_\mathfrak{g}$ on $\Lambda^2\mathfrak{g}^*$. 

Note that, for every linear Darboux family for $V_{\mathfrak{g}}$, bricks  are easy to obtain since, in view of (\ref{Eq:Red}), they are given by the intersection of an eigenvector space for each endomorphism of the form $\Lambda^2d\in \mathfrak{gl}(\Lambda^2\mathfrak{g})$ with  $d\in \mathfrak{der}(\mathfrak{g})$.

Let us give another interesting proposition, which is an immediate consequence of the fact that the vector fields of $V_{\mathfrak{g}}$ are linear relative to a linear coordinate set on $\Lambda^2\mathfrak{g}$.

\begin{proposition} The space $\Lambda^2\mathfrak{g}^*$ is a linear Darboux family of functions on  $\Lambda^2\mathfrak{g}$ relative to the Lie algebra $V_{\mathfrak{g}}$.
\end{proposition}

Since $\Lambda^2 \mathfrak{g}^*$ is a linear Darboux family of $V_{\mathfrak{g}}$, there exists a linear representation of $\mathfrak{aut}(\mathfrak{g})$ on $
\Lambda^2\mathfrak{g}^*$. Its irreducible representations also give rise to Darboux families, which can be summed (as linear spaces) to the Darboux family, or potentially elements thereof, associated with the mCYBE to determine new Darboux families. The locus of such sums will be interested to as when they contain elements of  $\mathcal{Y}_{\mathfrak{g}}$. This will be employed to obtain the strata of $V_{\mathfrak{g}}$ within $\mathcal{Y}_{\mathfrak{g}}$ and, therefore, inequivalent $r$-matrices relative to the action of ${\rm Aut}_c(\mathfrak{g})$ on $\Lambda^2\mathfrak{g}$. This in turn will be used to obtain families of inequivalent $r$-matrices and coboundary cocommutators for real four-dimensional indecomposable Lie algebras \cite{SW14} in Section \ref{Sec:Cla}.

\section{Geometric structure of solutions to mCYBEs and Darboux families}\label{Se:GSrMatDarFam}

Let us study several details on the geometry of the space of solutions to modified and non-modified classical Yang-Baxter equations, the problem of classification up to Lie algebra automorphisms of coboundary Lie bialgebras, and Darboux families. As previously, we hereafter write $V_{\mathfrak{g}}$ for the Lie algebra of fundamental vector fields of the action of ${\rm Aut}(\mathfrak{g})$ on $\Lambda^2\mathfrak{g}$, whilst $\mathscr{E}_\mathfrak{g}$ stands for the distribution on $\Lambda^2\mathfrak{g}$ spanned  by the vector fields of $V_{\mathfrak{g}}$. We recall that we have defined $\Lambda^2_R\mathfrak{g}:=\Lambda^2\mathfrak{g}/(\Lambda^2\mathfrak{g})^\mathfrak{g}$ and $\pi_\mathfrak{g}:w\in \Lambda^2\mathfrak{g}\mapsto [w]\in \Lambda^2_R\mathfrak{g}$ stands for the canonical projection onto the quotient space  $\Lambda_R^2\mathfrak{g}$. As previously, we hereafter write $\mathcal{Y}_\mathfrak{g}$ for the space of $r$-matrices of $\mathfrak{g}$.

The action of ${\rm Aut}(\mathfrak{g})$ on $\Lambda^2\mathfrak{g}$  induces an action of ${\rm Aut}(\mathfrak{g})$ on $\Lambda^2_R\mathfrak{g}$ of the form (see \cite[Lemma 8.3]{LW20} for details)
$$
\Psi:T\in {\rm Aut}(\mathfrak{g})\mapsto [\Lambda^2T]\in { GL}(\Lambda^2_R\mathfrak{g}), \qquad [\Lambda^2T]([w]):=[\Lambda^2T(w)],\qquad \forall w\in \Lambda^2\mathfrak{g}.
$$
The above result leads to the following consequence
\begin{equation}\label{Eq:Equiv}
\pi_{\mathfrak{g}}((\Lambda^2T)(w))=[\Psi(T)](\pi_{\mathfrak{g}}(w)),\qquad \forall w\in \Lambda^2\mathfrak{g},\qquad \forall T\in {\rm Aut}(\mathfrak{g}).
\end{equation}

Recall that two actions of a Lie group $G$, let us say $\Phi_i:G\times N_i\rightarrow N_i$ with $i=1,2$, are {\it equivariant} relative to  $\phi:N_1\rightarrow N_2$ if $\phi\circ \Phi_1(g,x)=\Phi_2(g,\phi(x))$ for every $g\in G$ and $x\in N_1$. It can be proved that if $\Phi_1$ and $\Phi_2$ are equivariant, the Lie algebra of fundamental vector fields of $\Phi_1$ projects, via $\phi_*$, onto the Lie algebra of fundamental vector fields of $\Phi_2$. Moreover, the orbits of $\Phi_1$ project, via $\phi$, onto the orbits of $\Phi_2$  (see \cite{AM78} for details). 

Expression (\ref{Eq:Equiv}) shows that the actions of ${\rm Aut}(\mathfrak{g})$ on $\Lambda^2\mathfrak{g}$ and $\Lambda^2_R\mathfrak{g}$  are equivariant relative to $\pi_\mathfrak{g}:\Lambda^2\mathfrak{g}\rightarrow \Lambda^2_R\mathfrak{g}$. Then, the vector fields of  $V_\mathfrak{g}$ project via $\pi_{\mathfrak{g}*}$ onto the Lie algebra of fundamental vector fields of the action of ${\rm Aut}(\mathfrak{g})$ on $\Lambda^2_R\mathfrak{g}$. We write $\mathscr{E}^R_{\mathfrak{g}}$ for the generalised distribution spanned by the fundamental vector fields of the action of ${\rm Aut}(\mathfrak{g})$ on $\Lambda^2_R\mathfrak{g}$. This means that the strata of the generalised distribution $\mathscr{E}_{\mathfrak{g}}$ project onto the strata of the generalised distribution $\mathscr{E}^R_\mathfrak{g}$. 

As the elements of ${\rm Aut}(\mathfrak{g})$ map solutions to the mCYBE of $\mathfrak{g}$ onto solutions of the same equation, the orbits of ${\rm Aut}(\mathfrak{g})$ containing a solution to the mCYBE consist of solutions to the mCYBE. Similarly, the orbits of the action of ${\rm Aut}(\mathfrak{g})$ on $\Lambda^2\mathfrak{g}$ containing a solution to the CYBE consist of solutions to the CYBE.   Two $r$-matrices are equivalent up to Lie algebra automorphisms of $\mathfrak{g}$ if and only if they belong to the same orbit of ${\rm Aut}(\mathfrak{g})$ within $\mathcal{Y}_\mathfrak{g}$. 

Recall that  ${\rm Aut}(\mathfrak{g})/{\rm Aut}_c(\mathfrak{g})$, where ${\rm Aut}_c(\mathfrak{g})$ is the connected component of the identity of the Lie group ${\rm Aut}(\mathfrak{g})$, is discrete and therefore countable \cite{Bo05}. In other words, the Lie group ${\rm Aut}(\mathfrak{g})$ is a numerable sum (as subsets) of disjoint and connected subsets of ${\rm Aut}(\mathfrak{g})$  diffeomorphic to ${\rm Aut}_c(\mathfrak{g})$. The   strata of $\mathscr{E}_\mathfrak{g}$ coincide with the orbits of ${\rm Aut}_c(\mathfrak{g})$. Therefore, the orbits of ${\rm Aut}(\mathfrak{g})$ are given by a numerable sum of  strata of $\mathscr{E}_\mathfrak{g}$. Moreover, each particular orbit of ${\rm Aut}(\mathfrak{g})$ is an immersed submanifold in $\Lambda^2\mathfrak{g}$ of a fixed dimension. Hence, each one of the orbits of ${\rm Aut}_c(\mathfrak{g})$, whose sum gives rise to an orbit of ${\rm Aut}(\mathfrak{g})$, must have the same dimension. Similarly, the orbits of ${\rm Aut}(\mathfrak{g})$ on $\Lambda^2_R\mathfrak{g}$ are given by a numerable sum (as subsets) of strata of $\mathscr{E}^R_{\mathfrak{g}}$ of the same dimension, which are orbits of the action of ${\rm Aut}_c(\mathfrak{g})$ on $\Lambda^2_R\mathfrak{g}$. 

It can be proved that two coboundary cocommutators can be equivalent up to Lie algebra automorphisms even when their associated $r$-matrices are not. For instance, 
the zero cocommutator $\delta_0:\mathfrak{g}\rightarrow \mathfrak{g}\wedge\mathfrak{g}$ can take the form $0=\delta_0=[r,\cdot]$ for $r=0$ or for any other $r_1\in (\Lambda^2\mathfrak{g})^\mathfrak{g}\backslash\{0\}$. Nevertheless, the $0$ and $r_1$ cannot be connected by the action of an element of ${\rm Aut}(\mathfrak{g})$ on $\Lambda^2\mathfrak{g}$, since the action of ${\rm Aut}(\mathfrak{g})$ on $\Lambda^2\mathfrak{g}$ is linear and $0\in \Lambda^2\mathfrak{g}$ is an orbit. In spite of that, the cocommutators generated by $0$ and $r_1$ are the same. Let us give the conditions ensuring that two coboundary cocommutators on $\mathfrak{g}$ are equivalent up to a Lie algebra automorphism of $\mathfrak{g}$.

\begin{proposition} Two coboundary cocommutators $\delta_i:v\in \mathfrak{g}\mapsto [r_i,v]\in \mathfrak{g}\wedge\mathfrak{g}$, with $i=1,2$ and $r_i\in \mathcal{Y}_\mathfrak{g}$, are equivalent under a Lie algebra automorphism of $\mathfrak{g}$ if and only if there exists $T\in {\rm Aut}(\mathfrak{g})$ such that $(\Lambda^2T)(r_1)$ and $r_2$ belong to the same equivalence class in $\Lambda^2_R\mathfrak{g}$.
\end{proposition}
\begin{proof} It is clear that the transformation of $\delta_1$ by $T$ reads, by the properties of the Schouten bracket, as follows
$$
\Lambda^2T\circ \delta_1\circ T^{-1}=\Lambda^2T\circ [r_1,T^{-1}(\cdot)]=[\Lambda^2Tr_1,\cdot].
$$
Hence,
$$
\delta_2=\Lambda^2T\circ \delta_1\circ T^{-1}\Leftrightarrow [\Lambda^2Tr_1,\cdot]=[r_2,\cdot] \Leftrightarrow [\Lambda^2Tr_1 - r_2,\cdot]=0,
$$
and the last condition amounts to the fact that $\Lambda^2Tr_1-r_2 \in (\Lambda^2 \mathfrak{g})^{\mathfrak{g}}$, which implies that $\Lambda^2Tr_1,r_2$ belong to the same equivalence class in $\Lambda_R^2\mathfrak{g}$. 
\end{proof}

Note that two $r$-matrices are equivalent (relative to the action of an element of ${\rm Aut}(\mathfrak{g})$ on $\Lambda^2\mathfrak{g}$) if and only if they belong to the same family of strata of the distribution $\mathscr{E}_\mathfrak{g}$ giving rise to an orbit of ${\rm Aut}(\mathfrak{g})$ in $\mathcal{Y}_\mathfrak{g}$. Nevertheless, two $r$-matrices give rise to two equivalent coboundary cocommutators if and only if they belong to the same orbit of ${\rm Aut}(\mathfrak{g})$ on $\Lambda^2_R\mathfrak{g}$. In other words, we have proven the following proposition.

\begin{proposition}\label{Pro:ClassificationRmatrix} There exists a one-to one correspondence between the $r$-matrices for $\mathfrak{g}$ that are equivalent up to an element of ${\rm Aut}_c(\mathfrak{g})$ and the strata of the generalised distribution $\mathscr{E}_\mathfrak{g}$ within $\mathcal{Y}_\mathfrak{g}$. Moreover, there exists a one-to-one correspondence between the families of equivalent (up to Lie algebra automorphisms of $\mathfrak{g}$) coboundary cocommutators of a Lie algebra $\mathfrak{g}$ and the orbits of the action of ${\rm Aut}(\mathfrak{g})$ on $\pi_{\mathfrak{g}}(\mathcal{Y}_{\mathfrak{g}})\subset \Lambda^2_R\mathfrak{g}$. Every such an orbit in $\pi_{\mathfrak{g}}(\mathcal{Y}_{\mathfrak{g}})$ is the sum (as subsets) of a numerable collection of strata of the same dimension of the generalised distribution $\mathscr{E}_\mathfrak{g}^R$.
\end{proposition}

In practice, we shall obtain each orbit of ${\rm Aut}(\mathfrak{g})$ on $\mathcal{Y}_\mathfrak{g}$ as a numerable family of strata of $\mathscr{E} _\mathfrak{g}$ in $\mathcal{Y}_\mathfrak{g}$. Such orbits represent the families of $r$-matrices that are equivalent up to an element of ${\rm Aut}(\mathfrak{g})$. Then, we will derive all strata in $\mathcal{Y}_\mathfrak{g}$ that map onto the same space in $\Lambda_R^2\mathfrak{g}$. This last task will give us, along with the orbits of ${\rm Aut}(\mathfrak{g})$ on $\mathcal{Y}_\mathfrak{g}$, the equivalent classes of  coboundary cocommutators on $\mathfrak{g}$ up to Lie algebra automorphisms of $\mathfrak{g}$. 

Let us study how Darboux families can be employed to obtain the orbits of ${\rm Aut}(\mathfrak{g})$ on $\Lambda^2\mathfrak{g}$ and $\Lambda_R^2\mathfrak{g}$. In particular, Darboux families are employed here to identify the strata of $\mathscr{E}_\mathfrak{g}$.

\begin{proposition}\label{Pro:UseProp} The locus, $\ell_{\mathfrak{g}}$, of a Darboux family $\mathcal{A}^E_\mathfrak{g}$ relative to the Lie algebra $V_\mathfrak{g}$ on $\Lambda^2\mathfrak{g}$ is a sum (as subsets) of the orbits of the action of ${\rm Aut}_c(\mathfrak{g})$ on $\Lambda^2\mathfrak{g}$ containing a point in $\ell_{\mathfrak{g}}$. If $\ell_{\mathfrak{g}}$ is a connected submanifold in $\Lambda^2\mathfrak{g}$ of dimension given by the rank of the generalised distribution $\mathscr{E}_{\mathfrak{g}}$, then $\ell_{\mathfrak{g}}$ is a   strata of the generalised distribution $\mathscr{E}_{\mathfrak{g}}$. 
\end{proposition}
\begin{proof} For each point $p$ in $\ell_{\mathfrak{g}}$ the orbit $\mathcal{O}_p$ of ${\rm Aut}_c(\mathfrak{g})$ passing through $p$ has a tangent space spanned by the vector fields of $V_{\mathfrak{g}}$. By the definition of a Darboux family for $V_{\mathfrak{g}}$ and the fact that the functions of $\mathcal{A}^E_\mathfrak{g}$ vanish  at $p$, one has that all functions of $\mathcal{A}^E_{\mathfrak{g}}$ are equal to zero on $\mathcal{O}_p$ (see Section \ref{Se:DarFam}). Hence, $\mathcal{O}_p$ is contained in $\ell_{\mathfrak{g}}$. Then, $\ell_{\mathfrak{g}}$ is the sum (as subsets) of orbits of ${\rm Aut}_c(\mathfrak{g})$ passing through the points of  $\ell_{\mathfrak{g}}$.

Recall that the strata of $\mathscr{E}_\mathfrak{g}$ are the orbits of ${\rm Aut}_c(\mathfrak{g})$ acting on $\Lambda^2\mathfrak{g}$. Hence, $\ell_{\mathfrak{g}}$ is a sum of strata of $\mathscr{E}_\mathfrak{g}$. 
If the rank of $\mathscr{E}_\mathfrak{g}$ on $\ell_\mathfrak{g}$ is equal to $\dim \ell_{\mathfrak{g}}$, where $\ell_{\mathfrak{g}}$ is assumed to be a connected submanifold, then $\ell_{\mathfrak{g}}$ is locally generated around any point $p\in \ell_{\mathfrak{g}}$ by the action of ${\rm Aut}_c(\mathfrak{g})$ on that point. Since $\ell_1$ is connected, it is wholly generated by the action of ${\rm Aut}_c(\mathfrak{g})$ and it becomes a strata of $\mathscr{E}_{\mathfrak{g}}$.
\end{proof}

Note that we are interested only in those loci of Darboux families of $V_{\mathfrak{g}}$ contained in the space $\mathcal{Y}_\mathfrak{g}$ of solutions to the mCYBE of $\mathfrak{g}$.

It is interesting that necessary conditions for the existence of a Lie algebra automorphism connecting two $r$-matrices can be given using  $\Lambda^3\mathfrak{g}$. Consider, for instance, the following proposition, whose proof is straightforward.

\begin{proposition}\label{Prop:NecCon} If $r_1$ and $r_2$ are two equivalent $r$-matrices for a Lie algebra $\mathfrak{g}$, then, $[r_1,r_1]$ and $[r_2,r_2]$ belong to the same orbit of ${\rm Aut}(\mathfrak{g})$ acting on $(\Lambda^3\mathfrak{g})^{\mathfrak{g}}$. 
\end{proposition}

The necessary condition in  Proposition \ref{Prop:NecCon} is not sufficient. For instance, it may happen that two $r$-matrices are solutions to the classical Yang-Baxter equation without being equivalent (as seen in Table \ref{Tab:g_orb_2}). From a practical point of view, Proposition \ref{Prop:NecCon} shows that the determination of the orbits of the action of ${\rm Aut}(\mathfrak{g})$ on $\Lambda^3\mathfrak{g}$ may be of interest. Following the reasoning for studying the equivalence of $r$-matrices in $\Lambda^2\mathfrak{g}$, the  orbits of the action of ${\rm Aut}(\mathfrak{g})$ on $\Lambda^3\mathfrak{g}$ can be obtained by obtaining the strata of the distribution, $\mathscr{E}^{(3)}_{\mathfrak{g}}$, spanned by the fundamental vector fields $V^{(3)}_{\mathfrak{g}}$ of the action of ${\rm Aut}_c(\mathfrak{g})$ on $\Lambda^3\mathfrak{g}$. Such strata can be obtained through Darboux families for $V^{(3)}_{\mathfrak{g}}$ in an analogous way to the method employed to study $\mathscr{E}_{\mathfrak{g}}$. Then, the action of one element of each connected component of ${\rm Aut}(\mathfrak{g})$ on the strata of  $V^{(3)}_{\mathfrak{g}}$ in $(\Lambda^2\mathfrak{g})^{\mathfrak{g}}$ gives the orbits of ${\rm Aut}(\mathfrak{g})$ within $(\Lambda^3\mathfrak{g})^{\mathfrak{g}}$. 

\begin{proposition} If two $r$-matrices $r_1,r_2$ are equivalent, then the rank of $r_1,r_2$, as bilinear mappings on $\mathfrak{g}^*$, are the same.
\end{proposition}
\begin{proof}
Recall that $r$-matrices can be considered as the antisymmetric bilinear  mappings on $\mathfrak{g}^*$ associated with them. The action of ${\rm Aut}(\mathfrak{g})$ on an $r$-matrix does not change its rank as a bilinear antisymmetric mapping. Hence, equivalent $r$-matrices must have the same rank. 
\end{proof}

Although $r$-matrices whose bilinear antisymmetric mappings on $\mathfrak{g}^*$ have different rank are not equivalent up to Lie algebra automorphisms of $\mathfrak{g}$, it may happen that $r$-matrices related to the same rank are not  equivalent. In fact, although there always exists in this latter case  an  $A\in GL(\mathfrak{g})$ acting on $\Lambda^2\mathfrak{g}$ mapping one element of this space into another of the same order ($A$ can be obtained by mapping both bivectors into a canonical form), $A$ does not necessarily belong to ${\rm Aut}(\mathfrak{g})$. Hence, both bivectors need not  be equivalent up to Lie algebra automorphisms of $\mathfrak{g}$. Many examples of this will appear in Table \ref{Tab:g_orb_2}.
 
\section{Classification of coboundary Lie bialgebras on four-dimensional indecomposable Lie algebras}\label{Sec:Cla}

A Lie algebra $\mathfrak{g}$ is {\it decomposable} when it can be written as the direct sum of two of its proper ideals. Otherwise, we say that $\mathfrak{g}$ is {\it indecomposable}. Winternitz and \v{S}nobl classified in \cite[Part 4]{SW14} all indecomposable Lie algebras up to dimension six (see \cite[p. 217]{SW14} for comments on previous classifications) . The aim of this section is to present the classification up to Lie algebra automorphisms of coboundary Lie bialgebras on real four-dimensional indecomposable Lie algebras via Darboux families and the Winternitz--\v{S}nobl (W\v{S}) classification. Moreover, the equivalence of $r$-matrices and solutions to CYBEs for four-dimensional real  indecomposable Lie algebras are also analysed.

We hereafter assume $\mathfrak{g}$ to be a real four-dimensional and indecomposable Lie algebra. For the sake of completeness, the structure constants for every $\mathfrak{g}$ in a basis $\{e_1,e_2,e_3,e_4\}$ are given in Table \ref{Tab:StruCons}. We hereafter endow $\Lambda^2\mathfrak{g}$  with a basis $\{e_{12}:=e_1\wedge e_2,e_{13}:=e_1\wedge e_3,e_{14}:=e_1\wedge e_4,e_{23}:=e_2\wedge e_3,e_{24}:=e_2\wedge e_4,e_{34}:=e_3\wedge e_4\}$, while its dual basis is denoted by $\{x_1,x_2,x_3,x_4,x_5,x_6\}$. Additionally, $\Lambda^3\mathfrak{g}$ is endowed with a basis $\{e_{123},e_{124},e_{134},e_{234}\}$. Moreover, due to the bilinearity and symmetry of the Schouten bracket on $\Lambda^2\mathfrak{g}$, it is enough to know  $[e_{ij},e_{kl}]$, with $1\leq i<k\leq 4$ or $i=k\in \{1,\ldots,4\},1\leq j\leq l\leq 4$, to determine the Schouten bracket on the whole $\Lambda^2\mathfrak{g}$. In Tables \ref{Tab:Schoten_1_2}--\ref{Tab:Schoten_1_3}, we summarise several relevant Schouten brackets between the basis elements of $\mathfrak{g}$, $\Lambda^2\mathfrak{g}$, and $\Lambda^3\mathfrak{g}$ that can easily be obtained via Table \ref{Tab:StruCons} and, eventually, by means of a symbolic computation program. Anyway, we add such calculations to easily follow the proofs of our following results. The classes of Lie algebras in Table \ref{Tab:StruCons} may contain several non-isomorphic Lie algebras depending on several parameters, e.g. $\alpha,\beta$. To simplify the notation, we will not detail specific values of the parameters when we refer to properties of a whole class of Lie algebras and/or it is clear the particular case we are discussing from the context. 

	\begin{table}[h]
\centering		
\begin{tabular}{|c||c|c|c|c|c|c|c|}
	\hline
	Lie algebra & $[e_1,e_2]$ & $[e_1,e_3]$ & $[e_1,e_4]$ & $[e_2,e_3]$ & $[e_2,e_4]$ & $[e_3,e_4]$ & Parameters \\
	\hhline{|=#=|=|=|=|=|=|=|}
	$\mathfrak{s}_{1}$ & 0 & 0 & 0 & 0 & $-e_1$ & $-e_3$ & \\
	\hline
	$\mathfrak{s}_{2}$ & 0 & 0 & $-e_1$ & 0 & $-e_1 -e_2$ & $-e_2 -e_3$ & \\
	\hline
	\multirow{2}{*}{$\mathfrak{s}_{3}$} & \multirow{2}{*}{0} & \multirow{2}{*}{0} & \multirow{2}{*}{$-e_1$} & \multirow{2}{*}{0} & \multirow{2}{*}{$-\alpha e_2$} & \multirow{2}{*}{$-\beta e_3$} & $0 < |\beta| \leq |\alpha| \leq 1$, \\
 & & & & & & & $(\alpha, \beta) \neq (-1,-1)$  \\
	\hline
	$\mathfrak{s}_{4}$ & 0 & 0 & $-e_1$ & 0 & $-e_1 -e_2$ & $-\alpha e_3$ & $\alpha \in \rz$ \\
	\hline
	$\mathfrak{s}_{5}$ & 0 & 0 & $-\alpha e_1$ & 0 & $-\beta e_2 + e_3$ & $-e_2 -\beta e_3$ & $\alpha > 0, \beta \in \rr$ \\
	\hline
	$\mathfrak{s}_{6}$ & 0 & 0 & 0 & $e_1$ & $-e_2$ & $e_3$ & \\
	\hline
	$\mathfrak{s}_{7}$ & 0 & 0 & 0 & $e_1$ & $e_3$ & $-e_2$ & \\
	\hline
	$\mathfrak{s}_{8}$ & 0 & 0 & $-(1 + \alpha) e_1$ & $e_1$ & $-e_2$ & $-\alpha e_3$ & $\alpha \in ]-1,1] \backslash \{0\}$ \\
	\hline
	$\mathfrak{s}_{9}$ & 0 & 0 & $-2\alpha e_1$ & $e_1$ & $-\alpha e_2 + e_3$ & $-e_2 -\alpha e_3$ & $\alpha > 0$ \\
	\hline
	$\mathfrak{s}_{10}$ & 0 & 0 & $-2e_1$ & $e_1$ & $-e_2$ & $-e_2 -e_3$ & \\
	\hline
	$\mathfrak{s}_{11}$ & 0 & 0 & $-e_1$ & $e_1$ & $-e_2$ & 0 & \\
	\hline
	$\mathfrak{s}_{12}$ & 0 & $-e_1$ & $e_2$ & $-e_2$ & $-e_1$ & 0 & \\
	\hline
	$\mathfrak{n}_{1}$ & 0 & 0 & 0 & 0 & $e_1$ & $e_2$ & \\
	\hline
\end{tabular}
\caption{Structure constants for real four-dimensional indecomposable Lie algebras according to the W\v{S} classification. Such Lie algebras are denoted by $\mathfrak{s}_{4,k},\mathfrak{n}_{4,1}$ for $1\leq k\leq 12$ in the W\v{S} classification. Nevertheless, the subindex concerning the dimension of the Lie algebra will be removed in our work to simplify the notation. The letters $\mathfrak{n}$ and $\mathfrak{s}$ are used to denote nilpotent and solvable (but not nilpotent) Lie algebras, respectively. To avoid duplicities of Lie algebras, the parameters $\alpha,\beta$ of $\mathfrak{s}_3$ must satisfy additional restrictions that were not detailed explicitly in \cite{SW14} (see \cite[p. 228]{SW14}). Such restrictions will be determined in Subsection \ref{Subsec:3}.}  \label{Tab:StruCons}
	\end{table}

It is interesting to discuss some physical and mathematical applications of the Lie bialgebras given in our list and their relations to some previous works. For instance, the nilpotent Lie algebra $\mathfrak{n}_1$ corresponds to the so-called Galilean Lie algebra, which is deformed  in the study of XYZ models with the quantum Galilei group appearing in  \cite{BCGST92}. Lie bialgebras defined on Lie algebras of the type $\mathfrak{n}_1$ also occur in \cite{Ne16,Op98}. The coboundary Lie bialgebras on the Lie algebra $\mathfrak{s}_6$, which corresponds to the so-called {\it harmonic oscillator algebra}, are obtained  in \cite{BH96} without studying the equivalence up to Lie algebra automorphisms. In that work, the quantisations and $R$-matrices of such a sort of Lie bialgebras are also studied. In particular, equation (3.7) in \cite{BH96} matches the mCYBE obtained in Section \ref{Se:s6}. The work \cite{BH97} proves that all Lie bialgebras on the harmonic oscillator algebra are coboundary ones. This shows that our classification of Lie bialgebras on $\mathfrak{s}_6$ finishes the classification of all Lie bialgebras on the harmonic oscillator algebra. On the other hand, the Lie algebra $\mathfrak{s}_8$ with $\alpha=-1/2$ is isomorphic to the Lie algebra $A\bar{G}_2(1)$ in \cite{Ne16}. The Lie algebra $\mathfrak{s}_7$ provides a central extension of Caley-Klein Lie algebras analysed in \cite[eq. (3.13)]{BHOS93}. Other cases of real four-dimensional indecomposable Lie bialgebras can be found in \cite{BCH00,Op00}, where, for instance, two-dimensional central extended Galilei algebras are studied. 

\subsection{Lie algebra $\mathfrak{s}_{1}$} \label{Sec:s1}

The structure constants for  Lie algebra $\mathfrak{s}_1$ are given in Table \ref{Tab:StruCons}, while relevant Schouten brackets between the elements of the bases of $\mathfrak{s}_1$, $\Lambda^2\mathfrak{s}_1$, $\Lambda^3\mathfrak{s}_1$ to be used hereafter are displayed in Tables  \ref{Tab:Schoten_1_2}--\ref{Tab:Schoten_1_3}. In particular, one sees from such tables that $(\Lambda^2\mathfrak{s}_1)^{\mathfrak{s}_1}=\langle e_{12}\rangle$ and $(\Lambda^3\mathfrak{s}_1)^{\mathfrak{s}_1}=0$.

To obtain the classification of $r$-matrices for $\mathfrak{s}_1$ up to Lie algebra automorphisms thereof, we will rely on the determination of the strata of the distribution $\mathscr{E}_{\mathfrak{s}_1}$ (see Proposition \ref{Pro:ClassificationRmatrix}) via  Darboux families for $V_{\mathfrak{s}_1}$ and a few elements of ${\rm Aut}(\mathfrak{s}_1)$. 

By Proposition \ref{Pro:UseProp}, the locus of a Darboux family for a Lie algebra $V_{\mathfrak{s}_1}$ is the sum (as subsets) of strata of the generalised distribution $\mathscr{E}_{\mathfrak{s}_1}$ spanned by the vector fields of $V_{\mathfrak{s}_1}$. In particular, we are interested in those loci being a submanifold within $\mathcal{Y}_{\mathfrak{s}_1}$ of dimension matching the rank of $\mathscr{E}_\mathfrak{\mathfrak{s}_1}$ on them. This is due to the fact that, in view of Proposition \ref{Pro:UseProp}, their connected components are the orbits of ${\rm Aut}_c(\mathfrak{s}_1)$ on $\mathcal{Y}_{\mathfrak{s}_1}$.
With this aim, we derive the fundamental vector fields of the action of ${\rm Aut}(\mathfrak{s}_1)$ on $\Lambda^2 \mathfrak{s}_1$. By Remark \ref{Re:DerAlg}, the derivations of $\mathfrak{s}_1$ take the form (\ref{Eq:Der1}),
and a basis of $V_{\mathfrak{s}_1}$ is given by (\ref{fundv_s1}).

For an element $r \in \Lambda^2 \mathfrak{s}_1$, we get
$$
[r,r] = 2(-x_2 x_5 + x_3 x_4 - x_4 x_5)e_{123}  -2x_5^2e_{124} + 2(x_3 -x_5)x_6e_{134} + 2x_5 x_6e_{234}.
$$
This expression can easily be derived by using the properties of the algebraic Schouten bracket and the structure constants for $\mathfrak{s}_1$. A similar computation would not be difficult even for many  higher-dimensional Lie algebras. 

Since $(\Lambda^3 \mathfrak{s}_1)^{\mathfrak{s}_1} = 0$, the mCYBE and the CYBE for $\mathfrak{s}_1$ are equal and read
$$
x_3 x_4 = 0, \quad x_3 x_6 = 0, \quad x_5=0.
$$
It is interesting  that, since $(\Lambda^3\mathfrak{s}_2)^{\mathfrak{s}_1}=0$, every $r$-matrix for $\mathfrak{s}_1$ amounts to a left-invariant Poisson bivector on $\mathfrak{s}_1^*$ (cf. \cite{CP94}).

Let us apply to $\mathfrak{s}_1$ the following  procedure to obtain the orbits of ${\rm Aut}_c(\mathfrak{g})$ on $\mathcal{Y}_\mathfrak{g}$ for a Lie algebra $\mathfrak{g}$ of arbitrary dimension, i.e. the strata of $\mathscr{E}_\mathfrak{g}$ within $\mathcal{Y}_\mathfrak{g}$. We start by considering a locus $\ell_1$ of a Darboux family of $V_\mathfrak{g}$ on $\Lambda^2\mathfrak{g}$. As we are interested in determining the strata of $\mathscr{E}_{\mathfrak{g}}$ within $\mathcal{Y}_{\mathfrak{g}}$, we look for an $\ell_1$ containing some solution to the mCYBE, namely $\ell_1\cap \mathcal{Y}_\mathfrak{g}\neq \emptyset$. Recall that  the strata of $\mathscr{E}_\mathfrak{g}$ are completely contained within the locus or completely contained off the locus of every Darboux family for $V_{\mathfrak{g}}$. This divides the strata of $\mathscr{E}_\mathfrak{g}$  on $\mathcal{Y}_\mathfrak{g}$ into those within or outside $\ell_1$. If $\ell_1\cap \mathcal{Y}_\mathfrak{g}$ is a submanifold of dimension given by the rank of $\mathscr{E}_{\mathfrak{g}}$ on it, its connected parts are the strata of $\mathscr{E}_{\mathfrak{g}}$ in $\ell_1\cap \mathcal{Y}_{\mathfrak{g}}$. Otherwise, to obtain the strata within $\ell_1\cap \mathcal{Y}_\mathfrak{g}$, we look for a new Darboux family for $V_\mathfrak{g}$ containing the previous Darboux family so that its new locus, let us say $\ell_2$,  will satisfy $\ell_2\cap \mathcal{Y}_\mathfrak{g}\neq\emptyset$. This process is repeated until we obtain a locus, $\ell_k$, that is a submanifold of dimension equal to the rank of $\mathscr{E}_\mathfrak{g}$ on it. In other words, the connected components of the locus $\ell_k$ are the orbits of ${\rm Aut}_c(\mathfrak{g})$ in $\ell_k$. 

To determine all the strata of $\mathscr{E}_\mathfrak{g}$ on $\mathcal{Y}_\mathfrak{g}\cap (\ell_k)^c$, where $(\ell_k)^c$ stands for the complementary subset of $\ell_k$ in $\Lambda^2\mathfrak{g}$, which is always open, one considers a subspace of the final Darboux family and the previous procedure is applied again iteratively to the obtained Darboux family on $\Lambda^2\mathfrak{g}\backslash\ell_k$. At the end, we obtain another family of strata of $\mathscr{E}_{\mathfrak{g}}$ and the procedure can be repeated to search for remaining strata of $\mathscr{E}_{\mathfrak{g}}$ in $\mathcal{Y}_{\mathfrak{g}}$. The above-described process, for all four-dimensional indecomposable Lie algebras, allows us to obtain the orbits of ${\rm Aut}_c(\mathfrak{g})$ on $\mathcal{Y}_\mathfrak{g}$.

To represent schematically our procedure, we use a diagram, which is hereafter called a {\it Darboux tree}. Every Darboux family of the above-mentioned method is described by a set of boxes going from an edge of the diagram on the left to one of the edges on the right. The squared boxes of the type $f=0$, for a certain function $f$, give the generating functions of the Darboux family while the squared boxes of the type $f\neq 0$,  give the conditions restricting the manifold in $\Lambda^2\mathfrak{g}$ where the Darboux family and $V_{\mathfrak{g}}$ are restricted to. The connected parts of the loci of all the Darboux families represented in a Darboux tree give rise to a decomposition of $\mathcal{Y}_{\mathfrak{g}}$ into orbits of ${\rm Aut}_c(\mathfrak{g})$. To help understanding the calculations, some oval boxes with additional information, e.g. explaining why other possibilities are not considered, are given. It is worth noting that bricks where employed to generate Darboux families. In particular, the Darboux tree of $\mathfrak{s}_1$ is given below. In this case, it is immediate that $x_5$ and $x_6$ are bricks for $V_{\mathfrak{s}_1}$. We use this fact to generate Darboux families of the following Darboux tree.

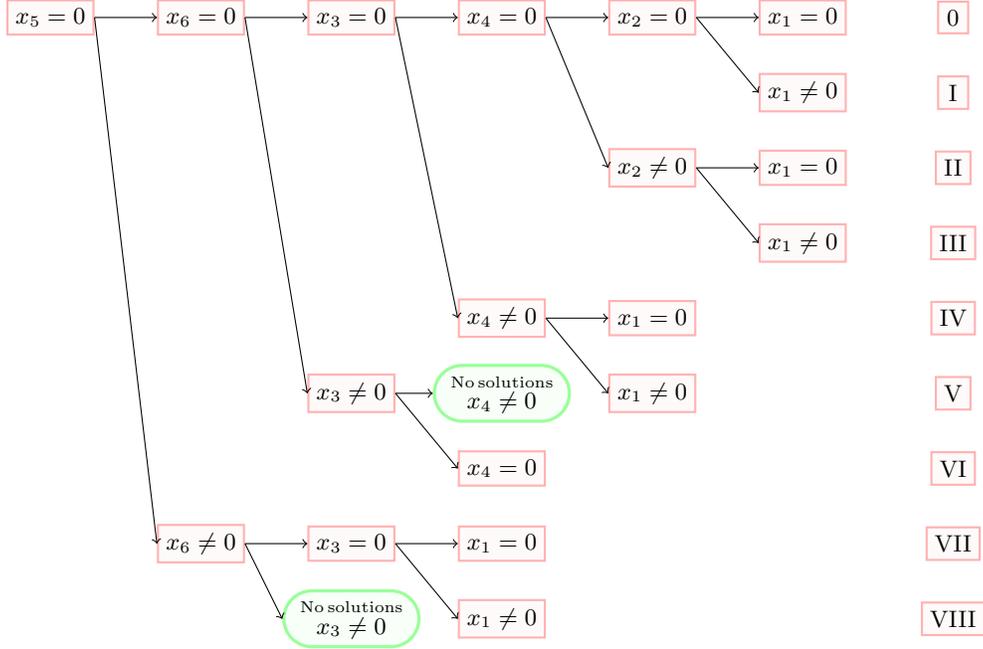
\begin{center}
{\small
\begin{tikzpicture}[
roundnode/.style={rounded rectangle, draw=green!40, fill=green!3, very thick, minimum size=2mm},
squarednode/.style={rectangle, draw=red!30, fill=red!2, thick, minimum size=4mm}
]
\node[squarednode] (brick) at (0,0) {$x_5=0$};

\node[squarednode] (u)  at (2,0) {$x_6=0$};
\node[squarednode] (d)  at (2,-7) {$x_6 \neq 0$};

\node[squarednode] (uu)  at (4,0) {$x_3=0$};
\node[squarednode] (ud)  at (4,-5) {$x_3 \neq 0$};
\node[squarednode] (du)  at (4,-7) {$x_3=0$};
\node[roundnode] (dd)  at (4,-8) {$\stackrel{\tiny{\rm No \, solutions}}{x_3 \neq 0}$};

\node[squarednode] (uuu)  at (6,0) {$x_4=0$};
\node[squarednode] (uud)  at (6,-4) {$x_4 \neq 0$};
\node[roundnode] (udu)  at (6,-5) {$\stackrel{\tiny{\rm No \, solutions}}{x_4 \neq 0}$};
\node[squarednode] (udd)  at (6,-6) {$x_4=0$};
\node[squarednode] (duu)  at (6,-7) {$x_1=0$};
\node[squarednode] (dud)  at (6,-8) {$x_1\neq 0$};

\node[squarednode] (uuuu)  at (8,0) {$x_2=0$};
\node[squarednode] (uuud)  at (8,-2) {$x_2 \neq 0$};
\node[squarednode] (uudu)  at (8,-4) {$x_1=0$};
\node[squarednode] (uudd)  at (8,-5) {$x_1 \neq 0$};

\node[squarednode] (uuuuu)  at (10,0) {$x_1 = 0$};
\node[squarednode] (uuuud)  at (10,-1) {$x_1 \neq 0$};
\node[squarednode] (uuudu)  at (10,-2) {$x_1 = 0$};
\node[squarednode] (uuudd)  at (10,-3) {$x_1 \neq 0$};

\node[squarednode] (0) at (12,0) {0};
\node[squarednode] (I) at (12,-1) {I};
\node[squarednode] (II) at (12,-2) {II};
\node[squarednode] (III) at (12,-3) {III};
\node[squarednode] (IV) at (12,-4) {IV};
\node[squarednode] (V) at (12,-5) {V};
\node[squarednode] (VI) at (12,-6) {VI};
\node[squarednode] (VII) at (12,-7) {VII};
\node[squarednode] (VIII) at (12,-8) {VIII};

\draw[->] (brick.east) -- (u.west);
\draw[->] (brick.east) -- (d.west);

\draw[->] (u.east) -- (uu.west);
\draw[->] (u.east) -- (ud.west);
\draw[->] (d.east) -- (du.west);
\draw[->] (d.east) -- (dd.west);

\draw[->] (uu.east) -- (uuu.west);
\draw[->] (uu.east) -- (uud.west);
\draw[->] (ud.east) -- (udu.west);
\draw[->] (ud.east) -- (udd.west);
\draw[->] (du.east) -- (duu.west);
\draw[->] (du.east) -- (dud.west);

\draw[->] (uuu.east) -- (uuuu.west);
\draw[->] (uuu.east) -- (uuud.west);
\draw[->] (uud.east) -- (uudu.west);
\draw[->] (uud.east) -- (uudd.west);

\draw[->] (uuuu.east) -- (uuuuu.west);
\draw[->] (uuuu.east) -- (uuuud.west);
\draw[->] (uuud.east) -- (uuudu.west);
\draw[->] (uuud.east) -- (uuudd.west);

\end{tikzpicture}
}
\end{center}


The connected components of the loci of the Darboux families of the previous Darboux tree give the orbits of ${\rm Aut}_c(\mathfrak{s}_1)$ on $\mathcal{Y}_{\mathfrak{s}_1}$. Results are given in Table \ref{Tab:g_orb_1}. Let us study the equivalence of $r$-matrices up to action of the whole ${\rm Aut}(\mathfrak{s}_1)$. To do so, we use that the group of automorphisms of $\mathfrak{s}_1$ reads
$$
{\rm Aut}(\mathfrak{s}_1) = \left\{
\left( 
\begin{array}{cccc}
T_2^2 & T^2_1 & 0 & T^4_1 \\
0 & T^2_2 & 0 & T^4_2 \\
0 & 0 & T_3^3 & T^4_3 \\
0 & 0 & 0 & 1
\end{array}
\right): T_2^2,  T_3^3 \in \rz, \, T^2_1, T^4_1, T^4_2, T^4_3 \in \mathbb{R}
\right\}.
$$
This Lie group is easy to be derived from the structure constants of $\mathfrak{s}_1$. In reality, it is enough for our purposes to consider an element of each connected component of ${\rm Aut}(\mathfrak{s}_1)$. It is remarkable that Lie algebra automorphism groups can rather easily be obtained for Lie algebras of relatively high dimension and/or satisfying special properties, e.g. for semisimple Lie algebras \cite{Mu52}. In our case, it is immediate that ${\rm Aut}(\mathfrak{s}_1)$ has four such components. One element of each connected component of ${\rm Aut}(\mathfrak{s}_1)$ and its extension  to $\Lambda^2\mathfrak{s}_1$ are given by 
\begin{equation}\label{aut_s1_con}
{\small T_{\lambda_1,\lambda_2}:=\left(
\begin{array}{cccc}
\lambda_1 & 0 & 0 & 0 \\
0 & \lambda_1 & 0 & 0 \\
0 & 0 & \lambda_2 & 0 \\
0 & 0 & 0 & 1
\end{array}
\right) \Longrightarrow 
\Lambda^2T_{\lambda_1,\lambda_2}:=\left(
\begin{array}{cccccc}
1 & 0 & 0 & 0 &0&0 \\
0 & \lambda_1 \lambda_2 & 0 & 0 &0&0 \\
0 & 0 & \lambda_1 & 0 &0&0 \\
0 & 0 & 0 & \lambda_1 \lambda_2&0&0 \\
0 & 0 & 0 & 0&\lambda_1&0 \\
0 & 0 & 0 & 0&0&\lambda_2 \\
\end{array}
\right),\qquad \lambda_1, \lambda_2 \in \{\pm 1\}}.
\end{equation}
The orbits of ${\rm Aut}(\mathfrak{s}_1)$ in $\mathcal{Y}_{\mathfrak{s}_1}$ are given by the action of $\Lambda^2T_{\lambda_1,\lambda_2}$ on the orbits of ${\rm Aut}_c(\mathfrak{s}_1)$ on $\mathcal{Y}_{\mathfrak{s}_1}$.
By using all mappings $\Lambda^2T_{\lambda_1,\lambda_2}$, we can verify whether some of the orbits of ${\rm Aut}_c(\mathfrak{s}_1)$ on $\mathcal{Y}_{\mathfrak{s}_1}$ can be  connected among themselves by a Lie algebra automorphism of $\mathfrak{s}_1$. Our results are summarised in Table \ref{Tab:g_orb_1}.

Recall that $(\Lambda^2\mathfrak{s}_1)^{\mathfrak{s}_1}=\langle e_{12}\rangle$. Consequently, all orbits of ${\rm Aut}(\mathfrak{s}_1)$ on $\mathcal{Y}_{\mathfrak{s}_1}$ mapping onto the same space in $\Lambda^2_R\mathfrak{s}_1$ via $\pi_{\mathfrak{s}_1}$ give equivalent coboundary coproducts up to the action of elements of ${\rm Aut}(\mathfrak{s}_1)$. In particular, we get five classes of coboundary coproducts induced by the following classes of $r$-matrices: a) ${\rm I}_{\pm}$, b) II, III$_{\pm}$, c) IV, V$_\pm$, d) VI, and e) VII, VIII$_{\pm}$. We recall that all given $r$-matrices are solutions to the CYBE.

\subsection{Lie algebra $\mathfrak{s}_{2}$}
Let us apply the formalism given in the previous section to  Lie algebra $\mathfrak{s}_2$. Lie algebra $\mathfrak{s}_{2}$ has structure constants given in Table \ref{Tab:StruCons}. Relevant Schouten brackets between the bases of elements of $\mathfrak{s}_2$, $\Lambda^2\mathfrak{s}_2$, and $\Lambda^3\mathfrak{s}_2$ are given in Tables    \ref{Tab:Schoten_1_2}--\ref{Tab:Schoten_1_3}. Remarkably,  $(\Lambda^2\mathfrak{s}_2)^{\mathfrak{s}_2}=0$ and $(\Lambda^3\mathfrak{s}_2)^{\mathfrak{s}_2}=0$.

By Remark \ref{Re:DerAlg}, it follows that $$\mathfrak{der}(\mathfrak{s}_2)=
\left\{\left(
\begin{array}{cccc}
\mu_{11} & \mu_{12} & \mu_{13} & \mu_{14} \\
0 & \mu_{11} & \mu_{12} & \mu_{24} \\
0 & 0 & \mu_{11} & \mu_{34} \\
0 & 0 & 0 & 0
\end{array}
\right): \mu_{11}, \mu_{12}, \mu_{13}, \mu_{14}, \mu_{24}, \mu_{34} \in \mathbb{R}\right\},
$$
which gives rise, by lifting these derivations to $\Lambda^2\mathfrak{s}_2$, to the basis of $V_{\mathfrak{s}_2}$ of the form
\begin{equation*}
\begin{gathered} 
X_1 = 2x_1 \partial_{x_1} + 2x_2 \partial_{x_2} + x_3 \partial_{x_3} + 2x_4 \partial_{x_4} + x_5 \partial_{x_5} +  x_6 \partial_{x_6}, \quad
X_2 = x_2 \partial_{x_1} + x_4 \partial_{x_2} + x_5 \partial_{x_3} + x_6 \partial_{x_5}, \\
X_3 = -x_4 \partial_{x_1} + x_6 \partial_{x_3}, \quad
X_4 = -x_5 \partial_{x_1} - x_6 \partial_{x_2}, \quad
X_5 = x_3 \partial_{x_1} - x_6 \partial_{x_4}, \quad
X_6 = x_3 \partial_{x_2} + x_5 \partial_{x_4}.
\end{gathered}
\end{equation*}

Meanwhile, for $r\in \Lambda^2\mathfrak{s}_2$, one has
$$
[r,r]=2(2x_1x_6 - 2x_2x_5 + x_2x_6 + 2x_3x_4 - x_4x_5)e_{123} + 2(x_3x_6 - x_5^2)e_{124} - 2x_5x_6e_{134} -2x_6^2e_{234}.
$$

Since $(\Lambda^3 \mathfrak{s}_2)^{\mathfrak{s}_2}=0$, the  mCYBE and the CYBE are the same and read
$$
x_3x_4 = 0, \quad x_5 = 0, \quad x_6 = 0.
$$

With the previous information, we are  ready to obtain the classification of orbits of the action of Aut$_c(\mathfrak{s}_2)$ on the subset of $r$-matrices $\mathcal{Y}_{\mathfrak{s}_2}\subset \Lambda^2\mathfrak{s}_2$ by using Darboux families. We shall start by using bricks. The only brick of $\mathfrak{s}_2$ is $x_6$. The procedure is accomplished as in the previous section and it is summarised in the following Darboux tree. 

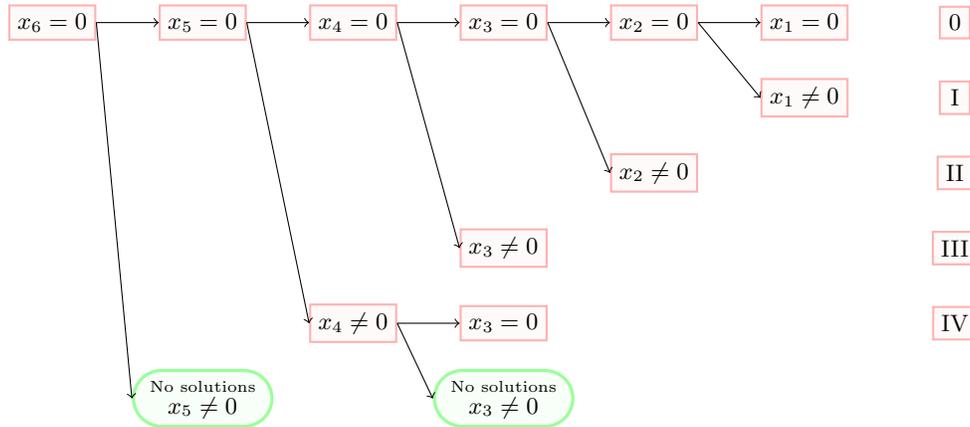
\begin{center}
 {\small
\begin{tikzpicture}[
roundnode/.style={rounded rectangle, draw=green!40, fill=green!3, very thick, minimum size=2mm},
squarednode/.style={rectangle, draw=red!30, fill=red!2, thick, minimum size=4mm}
]
\node[squarednode] (brick) at (0,0) {$x_6=0$};

\node[squarednode] (u)  at (2,0) {$x_5=0$};
\node[roundnode] (d)  at (2,-5) {$\stackrel{\tiny{\rm No\,\, solutions}}{x_5 \neq 0}$};

\node[squarednode] (uu)  at (4,0) {$x_4=0$};
\node[squarednode] (ud)  at (4,-4) {$x_4 \neq 0$};

\node[squarednode] (uuu)  at (6,0) {$x_3=0$};
\node[squarednode] (uud)  at (6,-3) {$x_3 \neq 0$};
\node[squarednode] (udu)  at (6,-4) {$x_3=0$};
\node[roundnode] (udd)  at (6,-5) {$\stackrel{\tiny{\rm No\,\, solutions}}{x_3 \neq 0}$};

\node[squarednode] (uuuu)  at (8,0) {$x_2=0$};
\node[squarednode] (uuud)  at (8,-2) {$x_2 \neq 0$};

\node[squarednode] (uuuuu)  at (10,0) {$x_1 = 0$};
\node[squarednode] (uuuud)  at (10,-1) {$x_1 \neq 0$};

\node[squarednode] (0) at (12,0) {0};
\node[squarednode] (I) at (12,-1) {I};
\node[squarednode] (II) at (12,-2) {II};
\node[squarednode] (III) at (12,-3) {III};
\node[squarednode] (IV) at (12,-4) {IV};

\draw[->] (brick.east) -- (u.west);
\draw[->] (brick.east) -- (d.west);

\draw[->] (u.east) -- (uu.west);
\draw[->] (u.east) -- (ud.west);

\draw[->] (uu.east) -- (uuu.west);
\draw[->] (uu.east) -- (uud.west);
\draw[->] (ud.east) -- (udu.west);
\draw[->] (ud.east) -- (udd.west);

\draw[->] (uuu.east) -- (uuuu.west);
\draw[->] (uuu.east) -- (uuud.west);

\draw[->] (uuuu.east) -- (uuuuu.west);
\draw[->] (uuuu.east) -- (uuuud.west);

\end{tikzpicture}}
\end{center}

As  commented in Subsection \ref{Sec:s1}, to obtain the orbits of ${\rm Aut}(\mathfrak{s}_2)$ on $\mathcal{Y}_{\mathfrak{s}_2}$, it is enough to derive the extension to $\Lambda^2\mathfrak{s}_2$ of a single element of each connected component of Aut$(\mathfrak{s}_2)$.
The automorphisms of the Lie algebra $\mathfrak{s}_{2}$ can easily be obtained. They  read
$$
{\rm Aut}(\mathfrak{s}_2)=\left\{\left(
\begin{array}{cccc}
T_2^2 & T^3_2 & T^3_1 & T^4_1 \\
0 & T^2_2 & T^3_2 & T^4_2 \\
0 & 0 & T^2_2 & T^4_3 \\
0 & 0 & 0 & 1
\end{array}
\right): T^2_2 \in \rz, \, T_{1}^3,T_2^3,T_1^4,T_2^4,T_3^4 \in \mathbb{R}\right\}.
$$
Then, one element of ${\rm Aut}(\mathfrak{s}_2)$ for each of its connected components and their extensions to $\Lambda^2\mathfrak{s}_2$ read
\begin{equation}
\label{aut_s2_con}
T_\lambda:=\left(
\begin{array}{cccc}
\lambda & 0 & 0 & 0 \\
0 & \lambda & 0 & 0 \\
0 & 0 & \lambda & 0 \\
0 & 0 & 0 & 1
\end{array}
\right)\Longrightarrow \Lambda^2T_\lambda={\small
\left(
\begin{array}{cccccc}
1 & 0 & 0 & 0 &0&0 \\
0 & 1 & 0 & 0 &0&0 \\
0 & 0 & \lambda & 0 &0&0 \\
0 & 0 & 0 & 1&0&0 \\
0 & 0 & 0 & 0&\lambda&0 \\
0 & 0 & 0 & 0&0&\lambda \\
\end{array}
\right)},\qquad \lambda \in \{\pm 1\}.
\end{equation}
The nine connected parts of the submanifolds $0$, I$_\pm$, II$_\pm$, III, IV$_\pm$ are the orbits of Aut$_c(\mathfrak{s}_2)$ on $\mathcal{Y}_{\mathfrak{s}_2}$. 
It is simple to see which of them are further connected through an element of  Aut$(\mathfrak{s}_2)$. In result,
the orbits of Aut$(\mathfrak{s}_2)$ on $\mathcal{Y}_{\mathfrak{s}_2}$ are given by the eight submanifolds $0,{\rm I}_-,{\rm I}_+,{\rm II}_-,{\rm II}_+, {\rm III}, {\rm IV}_-, {\rm IV}_+$ given in Table \ref{Tab:g_orb_1}.  Since $(\Lambda^2\mathfrak{s}_2)^{\mathfrak{s}_2}=0$, each such a submanifold gives a family of equivalent coboundary coproducts that is non-equivalent to $r$-matrices within remaining submanifolds. This gives all possible classes of non-equivalent coboundary coproducts. As in the previous subsection, all $r$-matrices are solutions to the CYBE in $\mathfrak{s}_2$.

\subsection{Lie algebra $\mathfrak{s}_{3}$}\label{Subsec:3}

The structure constants of the Lie algebras of type $\mathfrak{s}_3$ are given in Table \ref{Tab:StruCons}. It was noted in \cite[p. 228]{SW14}  that additional restrictions must be imposed on the parameters $\alpha,\beta$ given in Table \ref{Tab:StruCons} to avoid the repetition of isomorphic Lie algebras within the Lie algebra class $\mathfrak{s}_3$. Such restrictions were not explicitly detailed in \cite{SW14}, but it is immediate that isomorphic cases within Lie algebras of the class $\mathfrak{s}_3$ can be classified by the adjoint action of $e_4$ on $\langle e_1,e_2,e_3\rangle$. In particular, if $|\alpha|=|\beta|$, then the Lie algebras with parameters $(\alpha,\beta)$ and $(\beta,\alpha)$ are isomorphic relative to the Lie algebra isomorphism that interchanges $e_2$ with $e_3$ and leaves $e_1$ invariant. It will be relevant in what follows that the lift of this Lie algebra automorphism to $\Lambda^2\mathfrak{s}_3$ interchanges $x_1$ with $x_2$, $x_5$ with $x_6$, it maps $x_4$ to $-x_4$, and it leaves $x_3$ invariant. Due to the previous isomorphism, we restrict ourselves to the case $\alpha\geq \beta$ when $|\alpha|=|\beta|$.

Using the structure constants given in Table \ref{Tab:StruCons} and the induced Schouten brackets between the relevant to our purposes basis elements of $\mathfrak{s}_3$, $\Lambda^2\mathfrak{s}_3$, and $\Lambda^3\mathfrak{s}_3$ in Tables   \ref{Tab:Schoten_1_2}--\ref{Tab:Schoten_1_3},  one gets that $(\Lambda^2 \mathfrak{s}_3)^{\mathfrak{s}_3} \subset \{e_{12}, e_{13}, e_{23}\}$. More specifically, $e_{12} \in (\Lambda^2 \mathfrak{s}_3)^{\mathfrak{s}_3}$ if and only if $\alpha = -1$; while $e_{13} \in (\Lambda^2 \mathfrak{s}_3)^{\mathfrak{s}_3}$ if and only if $\beta = -1$. Finally, $e_{23} \in (\Lambda^2 \mathfrak{s}_3)^{\mathfrak{s}_3}$ just when $\alpha + \beta = 0$. Moreover, $(\Lambda^3\mathfrak{s}_3)^{\mathfrak{s}_3}=\langle e_{123}\rangle$ if $\alpha+\beta=-1$ and $(\Lambda^3\mathfrak{s}_3)^{\mathfrak{s}_3}=0$, otherwise.

Since the Lie algebra class $\mathfrak{s}_3$ contains so many subcases that relevant properties may change from one to another subcase with given parameters $(\alpha,\beta)$, e.g. the Lie algebra automorphism group, we will develop here a modification of our method consisting of applying the Darboux family method to the fundamental vector fields of the action of the Lie group ${\rm Aut}_{\rm all}(\mathfrak{s}_3)$ of all common Lie algebra automorphisms for all parameters $\alpha,\beta$ and analysing its relation to the Lie algebra automorphisms for each case, $\mathfrak{s}_{3}^{\alpha,\beta}$, to obtain our final classification. To simplify the notation, we skip the parameters $\alpha,\beta$ in $\mathfrak{s}_{3}^{\alpha,\beta}$ when they are not needed to understand what we are talking out, e.g. if we jest care about the dimension of the linear space $\Lambda^2\mathfrak{s}_3$.

By Remark \ref{Re:DerAlg}, one obtains the following common Lie algebra of derivations of $\mathfrak{s}_3$ for all the values of $\alpha$ and $\beta$:
$$
\mathfrak{der}_{\rm all}(\mathfrak{s}_3)=\left\{\left(
\begin{array}{cccc}
\mu_{11} & 0 & 0 & \mu_{14} \\
0 & \mu_{22} & 0 & \mu_{24} \\
0 & 0 & \mu_{33} & \mu_{34} \\
0 & 0 & 0 & 0
\end{array}
\right): \mu_{11}, \mu_{22}, \mu_{33}, \mu_{14}, \mu_{24}, \mu_{34} \in \mathbb{R}\right\},
$$
and, after lifting the elements of a basis of $\mathfrak{der}_{\rm all}(\mathfrak{s}_3)$ to $\Lambda^2\mathfrak{s}_3$, we obtain a basis of $V^{\rm all}_{\mathfrak{s}_3}$, namely the Lie algebra of fundamental vector fields of the action of the Lie algebra automorphism group ${\rm Aut}_{\rm all}(\mathfrak{s}_3)$ common to all  $\alpha,\beta$ acting on $\Lambda^2\mathfrak{s}_3$, of the form
\begin{equation*}
	\begin{gathered}
	X_1 = x_1 \partial_{x_1} + x_2 \partial_{x_2} + x_3 \partial_{x_3}, \quad
	X_2 = - x_5 \partial_{x_1} - x_6 \partial_{x_2}, \quad
	X_3 = x_1 \partial_{x_1} + x_4 \partial_{x_4} + x_5 \partial_{x_5}, \\
	X_4 = x_3 \partial_{x_1} - x_6 \partial_{x_4}, \quad
	X_5 = x_2 \partial_{x_2} + x_4 \partial_{x_4} + x_6 \partial_{x_6}, \quad
	X_6 = x_3 \partial_{x_2} + x_5 \partial_{x_4}.
	\end{gathered}
	\end{equation*}
It is worth stressing that if $\alpha=\beta$ and/or one on the coefficients $\alpha,\beta$ are equal to one, then the Lie algebra of derivations of the particular $\mathfrak{s}^{\alpha,\beta}_3$ is larger due to the existence of Lie algebra automorphisms of $\mathfrak{s}^{\alpha,\beta}_3$ leaving invariant the eigenspaces of ${\rm ad}_{e_4}$ acting on $\langle e_1,e_2,e_3\rangle$ (more detailed calculations can be seen at the end of this subsection). 
By dealing with $\mathfrak{der}_{\rm all}(\mathfrak{s}_3)$, we shall derive  the orbits of the connected part of the identity of the group of common Lie algebra automorphisms for all $\alpha,\beta$, i.e. ${\rm Aut}_{\rm all}(\mathfrak{s}_3)$, on each $\mathcal{Y}_{\mathfrak{s}^{\alpha,\beta}_3}$ via Darboux families. To obtain the equivalence of $r$-matrices up to the action of ${\rm Aut}(\mathfrak{s}^{\alpha,\beta}_3)$ for each pair $(\alpha,\beta)$, we will use the action of elements of ${\rm Aut}(\mathfrak{s}^{\alpha,\beta}_3)$ not contained in ${\rm Aut}_{\rm all,c}(\mathfrak{s}_3)$ for each particular pair of parameters $(\alpha,\beta)$. 

For an element $r \in \Lambda^2 \mathfrak{s}_3$, we get
$$
[r,r] = 2[(1 + \alpha)x_1 x_6 - (1 + \beta)x_2 x_5 + (\alpha + \beta)x_3 x_4]e_{123} + 2(\alpha - 1)x_3 x_5e_{124} + 2(\beta - 1)x_3 x_6e_{134} + 2(\beta - \alpha)x_5 x_6e_{234}.
$$

For $1 + \alpha + \beta \neq 0$, we have $(\Lambda^3 \mathfrak{s}^{\alpha,\beta}_3)^{\mathfrak{s}^{\alpha,\beta}_3} = 0$. Then, the mCYBE and the CYBE are the same in this case and they read
$$
(1 + \alpha)x_1 x_6 - (1 + \beta)x_2 x_5 + (\alpha + \beta)x_3 x_4 = 0, \quad (\alpha - 1)x_3 x_5 = 0, \quad (\beta - 1)x_3 x_6 = 0, \quad (\beta - \alpha)x_5 x_6 = 0.
$$

If $1 + \alpha + \beta = 0$, then $(\Lambda^3 \mathfrak{s}^{\alpha,-1-\alpha}_3)^{\mathfrak{s}^{\alpha,-1-\alpha}_3} = \langle e_{123}\rangle$ and the mCYBE reads
$$
(\alpha - 1)x_3 x_5 = 0, \quad (\beta - 1)x_3 x_6 = 0, \quad (\beta - \alpha)x_5 x_6 = 0.
$$

Since $x_3, x_5, x_6$ are   bricks of $\mathfrak{s}^{\alpha,\beta}_3$ for every pair $(\alpha,\beta)$, our Darboux tree for the Darboux families of  $V_{\mathfrak{s}_3}^{\rm all}$ starts with cases $x_i = 0$ and $x_i \neq 0$, $i \in \{3,5,6\}$. The full Darboux tree is presented below.

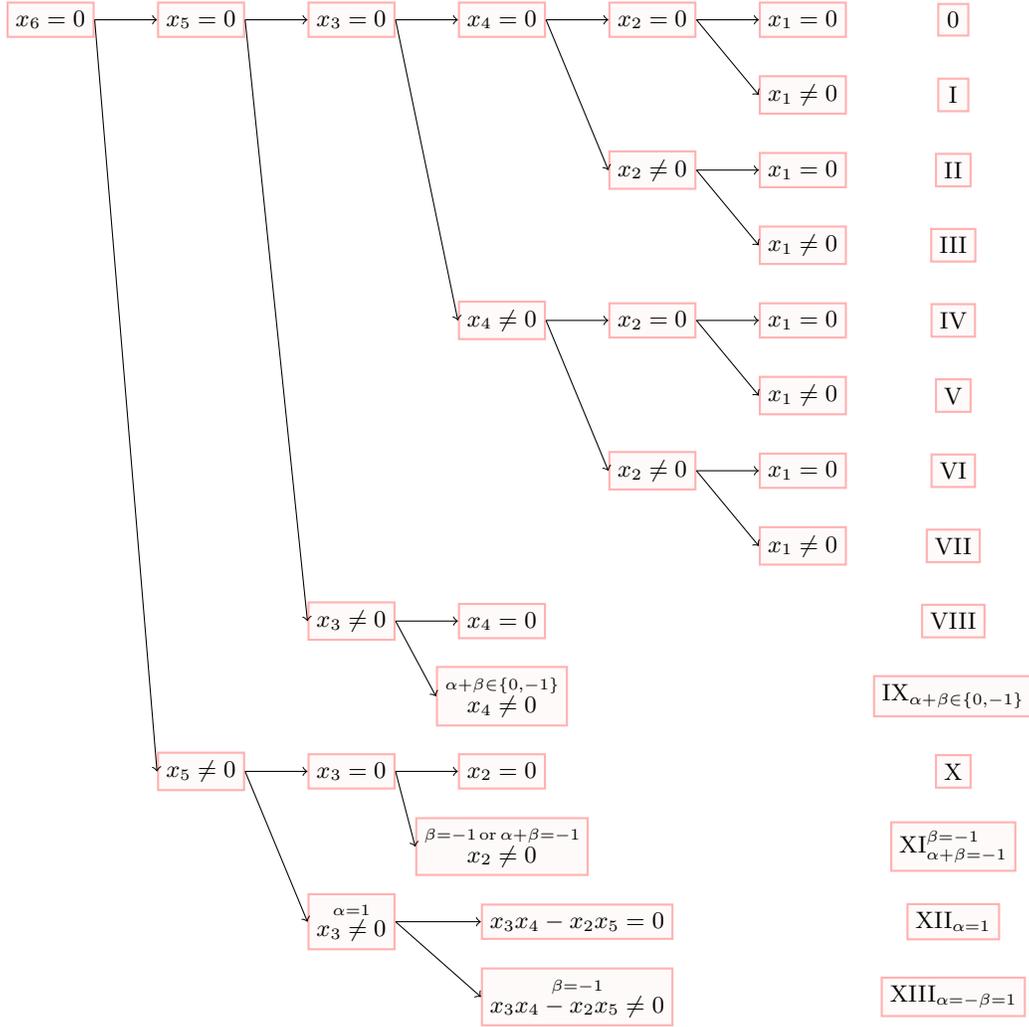
\begin{center}
{\small
\begin{tikzpicture}[
roundnode/.style={rounded rectangle, draw=green!40, fill=green!3, very thick, minimum size=2mm},
squarednode/.style={rectangle, draw=red!30, fill=red!2, thick, minimum size=4mm}
]
\node[squarednode] (brick) at (0,0) {$x_6 = 0$};

\node[squarednode] (u)  at (2,0) {$x_5=0$};
\node[squarednode] (d)  at (2,-10) {$x_5 \neq 0$};

\node[squarednode] (uu)  at (4,0) {$x_3=0$};
\node[squarednode] (ud)  at (4,-8) {$x_3 \neq 0$};
\node[squarednode] (du)  at (4,-10) {$x_3=0$};
\node[squarednode] (dd)  at (4,-12) {$\stackrel{\alpha = 1}{x_3 \neq 0}$};

\node[squarednode] (uuu)  at (6,0) {$x_4=0$};
\node[squarednode] (uud)  at (6,-4) {$x_4 \neq 0$};
\node[squarednode] (udu)  at (6,-8) {$x_4=0$};
\node[squarednode] (udd)  at (6,-9) {$\stackrel{\alpha + \beta \in \{0, -1\}}{x_4 \neq 0}$};
\node[squarednode] (duu)  at (6,-10) {$x_2=0$};
\node[squarednode] (dud)  at (6,-11) {$\stackrel{\beta = -1 \, {\rm or} \, \alpha + \beta = -1}{x_2 \neq 0}$};
\node[squarednode] (ddu)  at (7,-12) {$x_3 x_4 - x_2 x_5=0$};
\node[squarednode] (ddd)  at (7,-13) {$\stackrel{\beta  = -1}{x_3 x_4 - x_2 x_5 \neq 0}$};

\node[squarednode] (uuuu)  at (8,0) {$x_2=0$};
\node[squarednode] (uuud)  at (8,-2) {$x_2 \neq 0$};
\node[squarednode] (uudu)  at (8,-4) {$x_2=0$};
\node[squarednode] (uudd)  at (8,-6) {$x_2 \neq 0$};

\node[squarednode] (uuuuu)  at (10,0) {$x_1=0$};
\node[squarednode] (uuuud)  at (10,-1) {$x_1 \neq 0$};
\node[squarednode] (uuudu)  at (10,-2) {$x_1=0$};
\node[squarednode] (uuudd)  at (10,-3) {$x_1 \neq 0$};
\node[squarednode] (uuduu)  at (10,-4) {$x_1=0$};
\node[squarednode] (uudud)  at (10,-5) {$x_1 \neq 0$};
\node[squarednode] (uuddu)  at (10,-6) {$x_1=0$};
\node[squarednode] (uuddd)  at (10,-7) {$x_1 \neq 0$};

\node[squarednode] (0) at (12,0) {0};
\node[squarednode] (I) at (12,-1) {I};
\node[squarednode] (II) at (12,-2) {II};
\node[squarednode] (III) at (12,-3) {III};
\node[squarednode] (IV) at (12,-4) {IV};
\node[squarednode] (V) at (12,-5) {V};
\node[squarednode] (VI) at (12,-6) {VI};
\node[squarednode] (VII) at (12,-7) {VII};
\node[squarednode] (VIII) at (12,-8) {VIII};
\node[squarednode] (IX) at (12,-9) {${\rm IX}_{\alpha + \beta \in \{0, -1\}}$};
\node[squarednode] (X) at (12,-10) {X};
\node[squarednode] (XI) at (12,-11) {${\rm XI}^{\beta = -1}_{\alpha + \beta = -1}$};
\node[squarednode] (XII) at (12,-12) {${\rm XII}_{\alpha = 1}$};
\node[squarednode] (XIII) at (12,-13) {${\rm XIII}_{\alpha = -\beta = 1}$};

\draw[->] (brick.east) -- (u.west);
\draw[->] (brick.east) -- (d.west);

\draw[->] (u.east) -- (uu.west);
\draw[->] (u.east) -- (ud.west);
\draw[->] (d.east) -- (du.west);
\draw[->] (d.east) -- (dd.west);

\draw[->] (uu.east) -- (uuu.west);
\draw[->] (uu.east) -- (uud.west);
\draw[->] (ud.east) -- (udu.west);
\draw[->] (ud.east) -- (udd.west);
\draw[->] (du.east) -- (duu.west);
\draw[->] (du.east) -- (dud.west);
\draw[->] (dd.east) -- (ddu.west);
\draw[->] (dd.east) -- (ddd.west);

\draw[->] (uuu.east) -- (uuuu.west);
\draw[->] (uuu.east) -- (uuud.west);
\draw[->] (uud.east) -- (uudu.west);
\draw[->] (uud.east) -- (uudd.west);

\draw[->] (uuuu.east) -- (uuuuu.west);
\draw[->] (uuuu.east) -- (uuuud.west);
\draw[->] (uuud.east) -- (uuudu.west);
\draw[->] (uuud.east) -- (uuudd.west);
\draw[->] (uudu.east) -- (uuduu.west);
\draw[->] (uudu.east) -- (uudud.west);
\draw[->] (uudd.east) -- (uuddu.west);
\draw[->] (uudd.east) -- (uuddd.west);

\end{tikzpicture}
}
\end{center}

\begin{center}
{\small
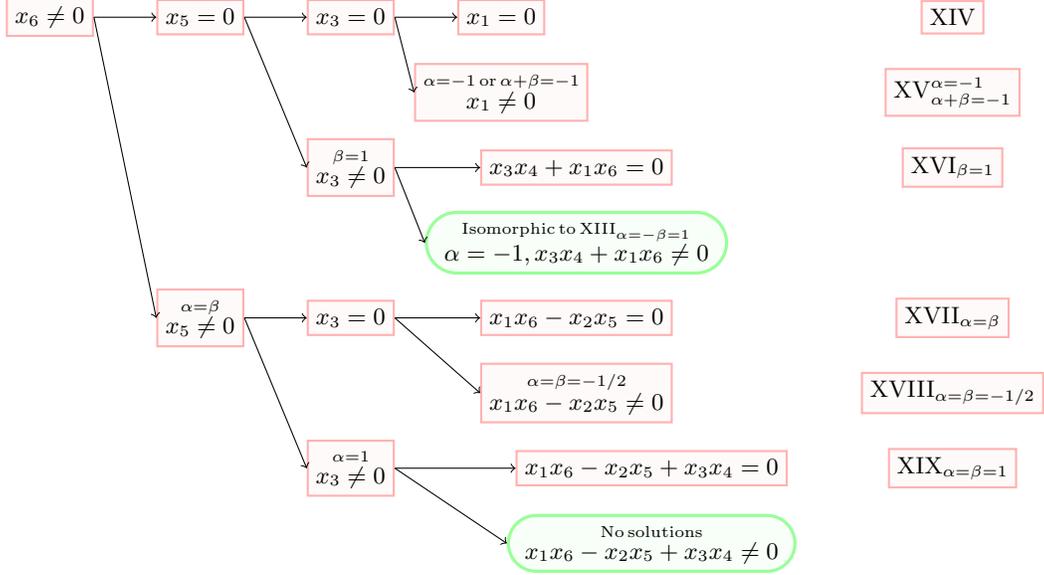
\begin{tikzpicture}[
roundnode/.style={rounded rectangle, draw=green!40, fill=green!3, very thick, minimum size=2mm},
squarednode/.style={rectangle, draw=red!30, fill=red!2, thick, minimum size=4mm}
]
\node[squarednode] (brick) at (0,0) {$x_6 \neq 0$};

\node[squarednode] (u)  at (2,0) {$x_5=0$};
\node[squarednode] (d)  at (2,-4) {$\stackrel{\alpha = \beta}{x_5 \neq 0}$};

\node[squarednode] (uu)  at (4,0) {$x_3=0$};
\node[squarednode] (ud)  at (4,-2) {$\stackrel{\beta = 1}{x_3 \neq 0}$};
\node[squarednode] (du)  at (4,-4) {$x_3=0$};
\node[squarednode] (dd)  at (4,-6) {$\stackrel{\alpha = 1}{x_3 \neq 0}$};

\node[squarednode] (uuu)  at (6,0) {$x_1=0$};
\node[squarednode] (uud)  at (6,-1) {$\stackrel{\alpha = -1 \, {\rm or} \, \alpha + \beta = -1}{x_1 \neq 0}$};
\node[squarednode] (udu)  at (7,-2) {$x_3 x_4 + x_1 x_6=0$};
\node[roundnode] (udd)  at (7,-3) {$\stackrel{{\rm Isomorphic\, to\, XIII}_{\alpha=-\beta=1}}{\alpha = -1,x_3 x_4 + x_1 x_6 \neq 0}$};
\node[squarednode] (duu)  at (7,-4) {$x_1 x_6 - x_2 x_5=0$};
\node[squarednode] (dud)  at (7,-5) {$\stackrel{\alpha=\beta=-1/2}{x_1 x_6 - x_2 x_5 \neq 0}$};
\node[squarednode] (ddu)  at (8,-6) {$x_1 x_6 - x_2 x_5 + x_3 x_4=0$};
\node[roundnode] (ddd)  at (8,-7) {$\stackrel{{\rm No \, solutions}}{x_1 x_6 - x_2 x_5 + x_3 x_4 \neq 0}$};

\node[squarednode] (0) at (12,0) {XIV};
\node[squarednode] (I) at (12,-1) {${\rm XV}^{\alpha = -1}_{\alpha + \beta = -1}$};
\node[squarednode] (II) at (12,-2) {${\rm XVI}_{\beta = 1}$};
\node[squarednode] (IV) at (12,-4) {${\rm XVII}_{\alpha = \beta}$};
\node[squarednode] (V) at (12,-5) {${\rm XVIII}_{\alpha =\beta=-1/2}$};
\node[squarednode] (VI) at (12,-6) {${\rm XIX}_{\alpha = \beta=1}$};

\draw[->] (brick.east) -- (u.west);
\draw[->] (brick.east) -- (d.west);

\draw[->] (u.east) -- (uu.west);
\draw[->] (u.east) -- (ud.west);
\draw[->] (d.east) -- (du.west);
\draw[->] (d.east) -- (dd.west);

\draw[->] (uu.east) -- (uuu.west);
\draw[->] (uu.east) -- (uud.west);
\draw[->] (ud.east) -- (udu.west);
\draw[->] (ud.east) -- (udd.west);
\draw[->] (du.east) -- (duu.west);
\draw[->] (du.east) -- (dud.west);
\draw[->] (dd.east) -- (ddu.west);
\draw[->] (dd.east) -- (ddd.west);

\end{tikzpicture}
}
\end{center}

\vspace{1cm}
We now study four subcases: a) $
\alpha=\beta=1$, b) $\alpha=1\neq \beta$, c) $\alpha=\beta\neq 1$, d) remaining non-isomorphic cases.

{\bf Case d)}: In this case, ${\rm ad}_{e_4}$ acts on $\langle e_1,e_2,e_3\rangle$ having three different eigenvalues and this leads to the fact that the only derivations are those common to all $\alpha,\beta$, i.e. $\mathfrak{der}(\mathfrak{s}^{\alpha,\beta}_3)=\mathfrak{der}_{\rm all}(\mathfrak{s}_3)$. The connected parts of the loci of the Darboux families of the above Darboux tree can be found in Table \ref{Tab:g_orb_2}. Such connected parts are the orbits of ${\rm Aut}_{c}(\mathfrak{s}^{\alpha,\beta}_3)$ in $\mathcal{Y}_{\mathfrak{s}^{\alpha,\beta}_3}$ for $\alpha,\beta,1$ taking different values. Let us obtain the classification of $r$-matrices up to elements of ${\rm Aut}(\mathfrak{s}^{\alpha,\beta}_3)$. On the one hand, the Lie algebra automorphism group reads
$$
{\rm Aut}(\mathfrak{s}^{\alpha,\beta}_3) = \left\{ \left(
\begin{array}{cccc}
T^1_1 & 0 & 0 & T^4_1 \\
0 & T^2_2 & 0 & T^4_2 \\
0 & 0 & T^3_3 & T^4_3 \\
0 & 0 & 0 & 1
\end{array}
\right):  T^1_1, T^2_2, T^3_3 \in \rz, T^4_1, T^4_2, T^4_3 \in \rr \right\}.
$$
As in previous sections, we only need for our purposes one element of each connected part of ${\rm Aut} (\mathfrak{s}^{\alpha,\beta}_3)$, which  has eight ones. An element of each connected component and its lift to $\Lambda^2\mathfrak{s} _3$ are given by
\begin{equation}\label{aut_s3_con}
T_{\lambda_1,\lambda_2,\lambda_3}:=\left(
\begin{array}{cccc}
\lambda_1 & 0 & 0 & 0 \\
0 & \lambda_2 & 0 & 0 \\
0 & 0 & \lambda_3 & 0 \\
0 & 0 & 0 & 1
\end{array}
\right) \Rightarrow \Lambda^2T_{\lambda_1,\lambda_2,\lambda_3}:={\small
\left(
\begin{array}{cccccc}
\lambda_1 \lambda_2 & 0 & 0 & 0 &0&0 \\
0 & \lambda_1 \lambda_3 & 0 & 0 &0&0 \\
0 & 0 & \lambda_1 & 0 &0&0 \\
0 & 0 & 0 & \lambda_2 \lambda_3&0&0 \\
0 & 0 & 0 & 0&\lambda_2&0 \\
0 & 0 & 0 & 0&0&\lambda_3 \\
\end{array}
\right)},\quad \lambda_1, \lambda_2, \lambda_3 \in \{\pm 1\}.
\end{equation}

By using (\ref{aut_s3_con}), we can verify whether some of the strata of $\mathscr{E}_{\mathfrak{s} _3}$ in $\mathcal{Y}_{\mathfrak{s}_3^{\alpha,\beta}}$ are still connected by a Lie algebra automorphism of $\mathfrak{s}^{\alpha,\beta}_3$. The results are summarised in Table \ref{Tab:g_orb_1}. To obtain the equivalence classes of coboundary coproducts for each $\alpha,\beta$, it is enough to identify the orbits in Table \ref{Tab:g_orb_1} whose elements are the same up to an element of $(\Lambda^2\mathfrak{s}^{\alpha,\beta}_3)^{\mathfrak{s}^{\alpha,\beta}_3}$.
	In particular, we have the following subcases:
\begin{itemize}
    \item Case $\alpha=-1, \beta \neq 1$. Hence, $(\Lambda^2\mathfrak{s}^{-1,\beta}_3)^{\mathfrak{s}^{-1,\beta}_3}=\langle e_{12}\rangle$. By analysing Table \ref{Tab:g_orb_1}, we obtain the coboundary coproduct classes
$$
a)\, {\rm I},\,b)\,{\rm II},{\rm III},\,c)\,{\rm IV}, {\rm V},\,d)\,{\rm VI}, {\rm VII}_\pm,\,e)\,{\rm VIII},\,f)\,{\rm X},\,g)\,{\rm XIV},{\rm XV}_{\alpha=-1}.
$$
\item Case $\alpha\neq -1$ and $\beta\neq -1$. Since $(\Lambda^2\mathfrak{s}^{\alpha,\beta}_3)^{\mathfrak{s}^{\alpha,\beta}_3}=0$, the classes of equivalent coboundary coproducts are given by the orbits of ${\rm Aut}(\mathfrak{s}^{\alpha,\beta}_3)$ within $\mathcal{Y}_{\mathfrak{s}^{\alpha,\beta}_3}$. In this case, we have the classes of equivalent coboundary coproducts given by 
$$
{\rm I}, {\rm II},  {\rm III}, {\rm IV}, {\rm V},  {\rm VI},  {\rm VII}_+, {\rm VII}_-,  {\rm VIII},  {\rm IX}_{\alpha+\beta\in\{0,-1\}},  {\rm X},  {\rm XI}_{\alpha+\beta=-1}, {\rm XIV}, {\rm XV}_{\alpha+\beta=-1}.
$$
\end{itemize}
{\bf Case c):} This time $\alpha=\beta\neq 1$. Then, we have
$$
{\rm Aut}(\mathfrak{s}^{\alpha,\alpha}_3) = \left\{ \left(
\begin{array}{cccc}
T^1_1 & 0 & 0 & T^4_1 \\
0 & T^2_2 & T_2^3 & T^4_2 \\
0 & T_3^2 & T^3_3 & T^4_3 \\
0 & 0 & 0 & 1
\end{array}
\right):  T^1_1, T^2_2T^3_3-T^2_3T^3_2 \in \rz, T_1^1,T^2_2,T^3_3,T_2^3,T_3^2,T^4_1, T^4_2, T^4_3 \in \rr \right\}.
$$
To obtain the orbits of ${\rm Aut}(\mathfrak{s}^{\alpha,\alpha}_3)$ on $\mathcal{Y}_{\mathfrak{s}^{\alpha,\alpha}_3}$ from the orbits of ${\rm Aut}_{{\rm all},c}(\mathfrak{s}_3)$, it is necessary to write ${\rm Aut}(\mathfrak{s}^{\alpha,\alpha}_3)$ as a composition of ${\rm Aut}_{{\rm all},c}(\mathfrak{s}_3)$ with certain Lie algebra automorphisms of $\mathfrak{s}^{\alpha,\alpha}_3$ so that their composition generates ${\rm Aut}(\mathfrak{s}^{\alpha,\alpha}_3)$. This can be done by using the Lie algebra automorphisms of $\mathfrak{s}_3^{\alpha,\alpha}$ of the form $
T_A:={\rm Id}\otimes A{\otimes} {\rm Id}
$, for $A\in GL(2,\mathbb{R})$. Then, $\Lambda^2T_A=A\otimes {\rm Id}\otimes (\det A) {\rm Id}\otimes A$. By taking the action of these $\Lambda^2T_A$  on the strata of the distribution spanned by $V^{\rm all}_{\mathfrak{s}_3}$ in $\mathcal{Y}_{\mathfrak{s}^{\alpha,\alpha}_3}$, we obtain the orbits of ${\rm Aut}(\mathfrak{s}^{\alpha,\alpha}_3)$ on $\mathcal{Y}_{\mathfrak{s}^{\alpha,\alpha}_3}$. Our results are summarised in Table \ref{Tab:g_orb_2}. 


To obtain the equivalence classes of coboundary coproducts for each $\alpha=\beta$, it is enough to identify the orbits in Table \ref{Tab:g_orb_1} whose elements are the same up to an element of $(\Lambda^2\mathfrak{s}^{\alpha,\alpha}_3)^{\mathfrak{s}^{\alpha,\alpha}_3}=0$. Therefore, the result is given by the classes of equivalent $r$-matrices detailed in Table \ref{Tab:g_orb_2}.

{\bf Case b):} This time $\alpha=1\neq \beta$. Therefore,
$$
{\rm Aut}(\mathfrak{s}^{1,\beta}_3) = \left\{ \left(
\begin{array}{cccc}
T^1_1 & T_1^2 & 0 & T^4_1 \\
T_2^1 & T^2_2 & 0 & T^4_2 \\
0 & 0 & T^3_3 & T^4_3 \\
0 & 0 & 0 & 1
\end{array}
\right):  T^3_3, T^1_1T^2_2-T^1_2T^2_1 \in \rz, T_1^1,T_2^2,T_1^2,T_2^1,T^4_1, T^4_2, T^4_3 \in \rr \right\}.
$$
To derive the orbits of ${\rm Aut}(\mathfrak{s}^{1,\beta}_3)$ on $\mathcal{Y}_{\mathfrak{s}^{1,\beta}_3}$ from the orbits of ${\rm Aut}_{{\rm all},c}(\mathfrak{s}_3)$, we again write ${\rm Aut}(\mathfrak{s}^{1,\beta}_3)$ as a composition of ${\rm Aut}_{{\rm all},c}(\mathfrak{s}_3)$ with certain Lie algebra automorphisms of $\mathfrak{s}^{1,\beta}_3$ so that their composition generates ${\rm Aut}(\mathfrak{s}^{1,\beta}_3)$. This can be done by employing the Lie algebra autmorphisms of $\mathfrak{s}^{1,\beta}_{3}$ given by $
T_A:=A{\otimes} {\rm Id}
\otimes {\rm Id}$, for $A\in GL(2,\mathbb{R})$. Then $\Lambda^2T_A=(\det A){\rm Id}\otimes (\tau_{43}\circ A\otimes A\circ \tau_{43})\otimes {\rm Id}$, where $\tau_{43}$ is the permutation of coordinates three and four in $\Lambda^2\mathfrak{s}^{1,\beta}_3$. By taking the action of these $\Lambda^2T_A$  on the strata of the distribution spanned by $V^{\rm all}_{\mathfrak{s}_3}$ in $\mathcal{Y}_{\mathfrak{s}^{1,\beta}_3}$, we obtain the orbits of ${\rm Aut}(\mathfrak{s}^{1,\beta}_3)$ on $\mathcal{Y}_{\mathfrak{s}^{1,\beta}_3}$. Our results are summarised in Table \ref{Tab:g_orb_2}. 


To obtain the equivalence classes of coboundary coproducts for each $\beta$, it is enough to identify the equivalence classes of $r$-matrices in Table \ref{Tab:g_orb_1} whose elements are the same up to an element of $(\Lambda^2\mathfrak{s}^{1,\beta}_3)^{\mathfrak{s}^{1,\beta}_3}$. If $\beta\neq -1$, these are just the induced by the orbits of ${\rm Aut}(\mathfrak{s}^{1,\beta}_3)$ on $\mathcal{Y}_{\mathfrak{s}^{1,\beta}_3}$. Otherwise, $(\Lambda^2\mathfrak{s}^{1,-1}_3)^{\mathfrak{s}^{1,-1}_3}=\langle e_{13},e_{23}\rangle$ and 

 $$
 a)\,\mathscr{I},\mathscr{III},\,b) \mathscr{II} \,{\rm(zero-class)},\,c)\,\mathscr{IV},\,d)\,\mathscr{V}_{\beta=-1}\,e)\,\mathscr{VI}.
 $$

{\bf Case a):} In this case, it is easier to derive the Darboux families for the Lie algebra of derivations.

$$
\mathfrak{der}_{\rm all}(\mathfrak{s}_3)=\left\{\left(
\begin{array}{cccc}
A &  v \\
0 & 0 
\end{array}
\right): A\in \mathfrak{gl}(3,\mathbb{R}),v\in \mathbb{R}^3\right\}.
$$
Due to the larger family of symmetries for $\alpha=\beta=1$, the Lie algebra $V_{\mathfrak{s}_3^{1,1}}$ is spanned by
\begin{equation*}
	\begin{gathered}
	X_1 = x_1 \partial_{x_1} + x_2 \partial_{x_2} + x_3 \partial_{x_3}, \quad
	X_2 = - x_5 \partial_{x_1} - x_6 \partial_{x_2}, \quad
	X_3 = x_1 \partial_{x_1} + x_4 \partial_{x_4} + x_5 \partial_{x_5}, \\
	X_4 = x_3 \partial_{x_1} - x_6 \partial_{x_4}, \quad
	X_5 = x_2 \partial_{x_2} + x_4 \partial_{x_4} + x_6 \partial_{x_6}, \quad
	X_6 = x_3 \partial_{x_2} + x_5 \partial_{x_4},\\
	X_7=x_4\partial_{x_2}+x_5\partial_{x_3},\quad X_8=-x_4\partial_{x_1}+x_6\partial_{x_3},\quad X_9=x_2\partial_{x_4}+x_3\partial_{x_5},\quad X_{10}=x_2\partial_{x_1}+x_6\partial_{x_5},\\
	X_{11}=-x_1\partial_{x_4}+x_3\partial_{x_6},\quad X_{12}=x_1\partial_{x_2}+x_5\partial_{x_6}.
	\end{gathered}
	\end{equation*}

The Darboux tree is very simple and becomes

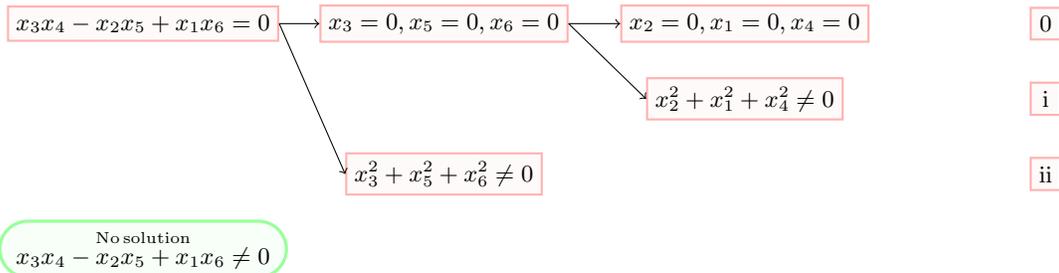
\begin{center}
{\small
\begin{tikzpicture}[
roundnode/.style={rounded rectangle, draw=green!40, fill=green!3, very thick, minimum size=2mm},
squarednode/.style={rectangle, draw=red!30, fill=red!2, thick, minimum size=4mm}
]
\node[squarednode] (brick) at (0,0) {$x_3x_4-x_2x_5+x_1x_6=0$};
\node[roundnode] (nosol) at (0,-3) {$\stackrel{{\rm No\, solution}}{x_3x_4-x_2x_5+x_1x_6\neq 0}$};
\node[squarednode] (u)  at (4,0) {$x_3=0,x_5=0,x_6=0$};
\node[squarednode] (d)  at (4,-2) {$x_3^2+x_5^2+x_6^2\neq 0$};

\node[squarednode] (uu)  at (8,0) {$x_2=0,x_1=0,x_4=0$};
\node[squarednode] (ud)  at (8,-1) {$x_2^2+x_1^2+x_4^2\neq 0$};


\node[squarednode] (0) at (12,0) {0};
\node[squarednode] (I) at (12,-1) {i};
\node[squarednode] (II) at (12,-2) {ii};

\draw[->] (brick.east) -- (u.west);
\draw[->] (brick.east) -- (d.west);

\draw[->] (u.east) -- (uu.west);
\draw[->] (u.east) -- (ud.west);


\end{tikzpicture}
}
\end{center}
It is immediate that the above gives rise to two non-zero coboundaries given in Table \ref{Tab:g_orb_2}. Since in this case $(\Lambda^2\mathfrak{s}^{1,1}_3)^{\mathfrak{s}^{1,1}_3}=0$, each class of equivalent $r$-matrices amounts to a class of equivalent coboundary coproducts.

\subsection{Lie algebra $\mathfrak{s}_{4}$}

Tables \ref{Tab:StruCons}, \ref{Tab:Schoten_1_2}--\ref{Tab:Schoten_1_3} contain the necessary information on the structure constants of $\mathfrak{s}_4$ and several Schouten brackets to prove our following results. Recall that $\alpha\in \mathbb{R}\backslash\{0\}$. First, $(\Lambda^3\mathfrak{s}_4)^{\mathfrak{s}_4}=\langle e_{123}\rangle$ for $2+\alpha=0$ and it is  zero for remaining allowed values of $\alpha$. Meanwhile, $(\Lambda^2\mathfrak{s}_4)^{\mathfrak{s}_4}=0$ for $\alpha\neq -1$ and $(\Lambda^2\mathfrak{s}_4)^{\mathfrak{s}_4}=\langle e_{13}\rangle$ for $\alpha=-1$.

As the Lie algebras of the class $\mathfrak{s}_4$ depend on a parameter $\alpha\in \mathbb{R}\backslash\{0\}$, the space of derivations depend on $\alpha$. It is indeed the same for all values of $\alpha\in \mathbb{R}\backslash\{0,1\}$, and it becomes larger for $\alpha=1$. Due to this, We shall proceed as in Subsection \ref{Subsec:3}. 
By Remark \ref{Re:DerAlg},  the space of derivations of $\mathfrak{s}_4$ for all possible values of $\alpha$ read
$$
\mathfrak{der}_{\rm all}(\mathfrak{s}_4):=\left\{\left(
\begin{array}{cccc}
\mu_{11} & \mu_{12} & 0 & \mu_{14} \\
0 & \mu_{11} & 0 & \mu_{24} \\
0 & 0 & \mu_{33} & \mu_{34} \\
0 & 0 & 0 & 0
\end{array}
\right):\mu_{11}, \mu_{12}, \mu_{33}, \mu_{14}, \mu_{24}, \mu_{34} \in \mathbb{R}\right\}.
$$
By extending the previous derivations to $\Lambda^2\mathfrak{s}_4$ and using Remark \ref{Re:DerAlg}, we obtain a basis of $V^{\rm all}_{\mathfrak{s}_4}$ of the form
\begin{equation*}
\begin{gathered}
X_1 = 2x_1 \partial_{x_1} + x_2 \partial_{x_2} + x_3 \partial_{x_3} + x_4 \partial_{x_4} + x_5 \partial_{x_5}, \quad
X_2 = x_4 \partial_{x_2} + x_5 \partial_{x_3}, \quad
X_3 = -x_5 \partial_{x_1} - x_6 \partial_{x_2}, \\
X_4 = x_3 \partial_{x_1} - x_6 \partial_{x_4}, \quad
X_5 = x_2 \partial_{x_2} + x_4 \partial_{x_4} + x_6 \partial_{x_6}, \quad
X_6 = x_3 \partial_{x_2} + x_5 \partial_{x_4}.
\end{gathered}
\end{equation*}
We recall that these vector fields span the Lie algebra of fundamental vector fields of the action on $\Lambda^2\mathfrak{s}_4$ of the Lie algebra automorphisms ${\rm Aut}_{\rm all}(\mathfrak{s}_4)$ that are common for all values of $\alpha\in \mathbb{R}\backslash\{0\}$.

For an element $r \in \Lambda^2 \mathfrak{s}_4$, we get
$$
[r,r] = 2[2x_1 x_6 - (1 + \alpha)x_2 x_5 + (1 + \alpha)x_3 x_4 - x_4 x_5]e_{123} - 2x_5^2e_{124} + 2[-(1 - \alpha)x_3 x_6 - x_5 x_6]e_{134} + 2(\alpha - 1)x_5 x_6e_{234}.
$$
If $\alpha \neq -2$, we have $(\Lambda^3 \mathfrak{s}_4)^{\mathfrak{s}_4} = 0$. Thus, the mCYBE and the CYBE are equal in this case and they read
$$
2x_1 x_6 + (1 + \alpha) x_3 x_4 = 0, \quad (1 - \alpha)x_3 x_6 = 0, \quad x_5 = 0.
$$
For the case $\alpha + 2 = 0$, we have $(\Lambda^3 \mathfrak{s}_4)^{\mathfrak{s}_4} = \langle e_{123} \rangle$. Thus, the mCYBE reads
$$
\quad (1 - \alpha)x_3 x_6 = 0, \quad x_5 = 0.
$$

Since $x_5, x_6$ are the bricks for $\mathfrak{s}_4$, our Darboux tree  for the Darboux families starts with the cases $x_i  = 0$ and $x_i \neq 0$, $i \in \{5,6\}$. The full Darboux tree is presented below.

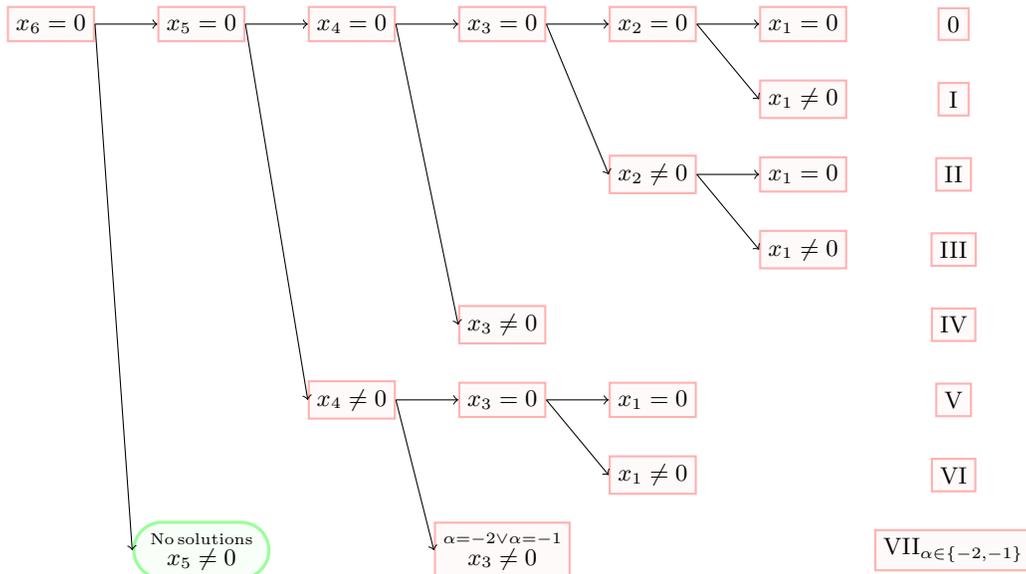
\begin{center}
{\small
\begin{tikzpicture}[
roundnode/.style={rounded rectangle, draw=green!40, fill=green!3, very thick, minimum size=2mm},
squarednode/.style={rectangle, draw=red!30, fill=red!2, thick, minimum size=4mm}
]
\node[squarednode] (brick) at (0,0) {$x_6=0$};

\node[squarednode] (u)  at (2,0) {$x_5=0$};
\node[roundnode] (d)  at (2,-7) {$\stackrel{{\rm No \, solutions}}{x_5 \neq 0}$};

\node[squarednode] (uu)  at (4,0) {$x_4=0$};
\node[squarednode] (ud)  at (4,-5) {$x_4 \neq 0$};

\node[squarednode] (uuu)  at (6,0) {$x_3=0$};
\node[squarednode] (uud)  at (6,-4) {$x_3 \neq 0$};
\node[squarednode] (udu)  at (6,-5) {$x_3=0$};
\node[squarednode] (udd)  at (6,-7) {$\stackrel{\alpha = -2 \lor \alpha = -1}{x_3 \neq 0}$};

\node[squarednode] (uuuu)  at (8,0) {$x_2=0$};
\node[squarednode] (uuud)  at (8,-2) {$x_2 \neq 0$};
\node[squarednode] (uduu)  at (8,-5) {$x_1=0$};
\node[squarednode] (udud)  at (8,-6) {$x_1 \neq 0$};

\node[squarednode] (uuuuu)  at (10,0) {$x_1=0$};
\node[squarednode] (uuuud)  at (10,-1) {$x_1 \neq 0$};
\node[squarednode] (uuudu)  at (10,-2) {$x_1=0$};
\node[squarednode] (uuudd)  at (10,-3) {$x_1 \neq 0$};

\node[squarednode] (0) at (12,0) {0};
\node[squarednode] (I) at (12,-1) {I};
\node[squarednode] (II) at (12,-2) {II};
\node[squarednode] (III) at (12,-3) {III};
\node[squarednode] (IV) at (12,-4) {IV};
\node[squarednode] (V) at (12,-5) {V};
\node[squarednode] (VI) at (12,-6) {VI};
\node[squarednode] (VII) at (12,-7) {${\rm VII}_{\alpha \in \{-2,-1\}}$};

\draw[->] (brick.east) -- (u.west);
\draw[->] (brick.east) -- (d.west);

\draw[->] (u.east) -- (uu.west);
\draw[->] (u.east) -- (ud.west);

\draw[->] (uu.east) -- (uuu.west);
\draw[->] (uu.east) -- (uud.west);
\draw[->] (ud.east) -- (udu.west);
\draw[->] (ud.east) -- (udd.west);

\draw[->] (uuu.east) -- (uuuu.west);
\draw[->] (uuu.east) -- (uuud.west);
\draw[->] (udu.east) -- (uduu.west);
\draw[->] (udu.east) -- (udud.west);

\draw[->] (uuuu.east) -- (uuuuu.west);
\draw[->] (uuuu.east) -- (uuuud.west);
\draw[->] (uuud.east) -- (uuudu.west);
\draw[->] (uuud.east) -- (uuudd.west);
\end{tikzpicture}

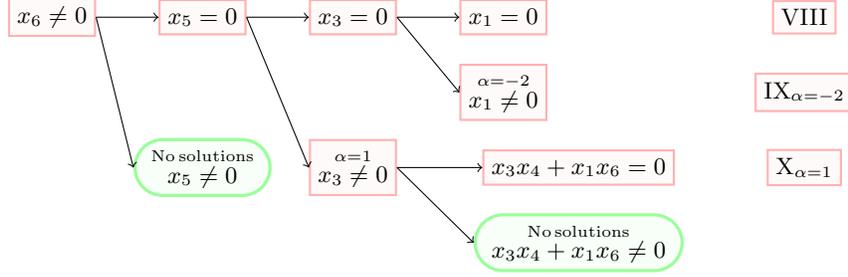
\vspace{1cm}

\begin{tikzpicture}[
roundnode/.style={rounded rectangle, draw=green!40, fill=green!3, very thick, minimum size=2mm},
squarednode/.style={rectangle, draw=red!30, fill=red!2, thick, minimum size=4mm}
]
\node[squarednode] (brick) at (0,0) {$x_6\neq 0$};

\node[squarednode] (u)  at (2,0) {$x_5=0$};
\node[roundnode] (d)  at (2,-2) {$\stackrel{{\rm No \, solutions}}{x_5 \neq 0}$};

\node[squarednode] (uu)  at (4,0) {$x_3=0$};
\node[squarednode] (ud)  at (4,-2) {$\stackrel{\alpha = 1}{x_3 \neq 0}$};

\node[squarednode] (uuu)  at (6,0) {$x_1=0$};
\node[squarednode] (uud)  at (6,-1) {$\stackrel{\alpha = -2}{x_1 \neq 0}$};
\node[squarednode] (udu)  at (7,-2) {$x_3 x_4 + x_1 x_6 = 0$};
\node[roundnode] (udd)  at (7,-3) {$\stackrel{{\rm No \, solutions}}{x_3 x_4 + x_1 x_6 \neq 0}$};

\node[squarednode] (VIII) at (10,0) {VIII};
\node[squarednode] (IX) at (10,-1) {${\rm IX}_{\alpha = -2}$};
\node[squarednode] (X) at (10,-2) {${\rm X}_{\alpha = 1}$};

\draw[->] (brick.east) -- (u.west);
\draw[->] (brick.east) -- (d.west);

\draw[->] (u.east) -- (uu.west);
\draw[->] (u.east) -- (ud.west);

\draw[->] (uu.east) -- (uuu.west);
\draw[->] (uu.east) -- (uud.west);
\draw[->] (ud.east) -- (udu.west);
\draw[->] (ud.east) -- (udd.west);

\end{tikzpicture}
}
\end{center}
The connected parts of the subspaces denoted in the above diagram are the orbits of ${\rm Aut}_{\rm all,c}(\mathfrak{s}_4)$ on $\mathcal{Y}_{\mathfrak{s}_4}$. To obtain the orbits of ${\rm Aut}_{\rm all}(\mathfrak{s}_4)$ on $\mathcal{Y}_{\mathfrak{s}_4}$, where 
$$
{\rm Aut}_{\rm all}(\mathfrak{s}_4) = \left\{
\left(
\begin{array}{cccc}
T^2_2 & T^2_1 & 0 & T^4_1 \\
0 & T^2_2 & 0 & T^4_2 \\
0 & 0 & T^3_3 & T^4_3 \\
0 & 0 & 0 & 1
\end{array}
\right): T^2_2, T^3_3 \in \rz, \, T^2_1, T_1^4, T_2^4, T_3^4 \in \rr
\right\},
$$
we verify whether some of the connected components of the orbits of ${\rm Aut}_{\rm all, c} (\mathfrak{s}_4)$ on $\mathcal{Y}_{\mathfrak{s}_4}$ are additionally connected by an element of ${\rm  Aut}_{\rm all}(\mathfrak{s}_4)$. To do so, 
we use, as previously, the lift to $\Lambda^2\mathfrak{s}_4$ of one element of each connected component of ${\rm Aut}_{\rm all}(\mathfrak{s}_4)$, for instance
\begin{equation}\label{aut_s4_con}
T_{\lambda_1,\lambda_2}:=\left(
\begin{array}{cccc}
\lambda_1 & 0 & 0 & 0 \\
0 & \lambda_1 & 0 & 0 \\
0 & 0 & \lambda_2 & 0 \\
0 & 0 & 0 & 1
\end{array}
\right) \Rightarrow \Lambda^2T_{\lambda_1,\lambda_2}={\small
\left(
\begin{array}{cccccc}
1 & 0 & 0 & 0 &0&0 \\
0 & \lambda_1 \lambda_2 & 0 & 0 &0&0 \\
0 & 0 & \lambda_1 & 0 &0&0 \\
0 & 0 & 0 & \lambda_1 \lambda_2&0&0 \\
0 & 0 & 0 & 0&\lambda_1&0 \\
0 & 0 & 0 & 0&0&\lambda_2 \\
\end{array}
\right)}, \quad \lambda_1, \lambda_2 \in \{\pm 1\}.
\end{equation}
For $\alpha\in \mathbb{R}\backslash\{1,0\}$, one has that ${\rm Aut}(\mathfrak{s}^\alpha_4)={\rm Aut}_{\rm all}(\mathfrak{s}_4)$, where $\mathfrak{s}^\alpha_4$ stands for the Lie algebra $\mathfrak{s}_4$ for a fixed value of $\alpha$. Then, the classes of equivalent $r$-matrices (up to Lie algebra automorphisms of $\mathfrak{s}_4^1$) on the Lie algebra $\mathfrak{s}^\alpha_4$ can  easily be obtained and they are summarised in Table \ref{Tab:g_orb_1}. Moreover, the classes of equivalent coboundary coproducts for $\alpha=-1$ are induced by the families
$$
a)\, {\rm II} ({\rm zero\, class}),\,b)\,{\rm I}_+,{\rm III}_+,\,c)\,{\rm I}_-,{\rm III}_-,\,d)\,{\rm IV},\,e)\,{\rm VII}_{\alpha=-1},\,f)\,{\rm V},\,g)\,{\rm VI}_+,\,h)\,{\rm VI}_-,\,i)\,{\rm VIII}.
$$
For those $\mathfrak{s}_{4}^\alpha$ with $|\alpha|\neq 1$, one has $(\Lambda^2\mathfrak{s}_4^{\alpha})^{\mathfrak{s}_4^{\alpha}}=0$ and the classes of equivalent coboundary coproducts are given by each one of the following subsets
$$
a)\,{\rm I}_-,\,b)\,{\rm I}_+,\,c)\,{\rm II},\,d)\,{\rm III}_+,\,e)\,{\rm III}_-,\,f)\,{\rm IV},\,g)\,{\rm V},\,h)\,{\rm VI}_+,\,i)\,{\rm VI}_-,\,j)\,{\rm VII}_{\alpha=-2},\,k)\,{\rm VIII},\,l)\,{\rm IX}_-^{\alpha=-2},\,m)\,{\rm IX}_+^{\alpha=-2}.
$$
Note that, for values $|\alpha|\neq 1$, not all the above classes may be simultaneously available as they arise for particular values of $\alpha$. 

Let us consider now the case of the Lie algebra $\mathfrak{s}^1_4$, i.e. the Lie algebra of the class $\mathfrak{s}_4$ for $\alpha=1$. In this case, the group of Lie algebra automorphisms is slightly larger than for remaining admissible values of $\alpha$. In particular,  
$$
{\rm Aut}(\mathfrak{s}^1_4) = \left\{
\left(
\begin{array}{cccc}
T^2_2 & T^2_1 & T^3_1 & T^4_1 \\
0 & T^2_2 & 0 & T^4_2 \\
0 & T_3^2 & T^3_3 & T^4_3 \\
0 & 0 & 0 & 1
\end{array}
\right): T^2_2, T^3_3 \in \rz, \, T_3^2,T^2_1,T^3_1, T_1^4, T_2^4, T_3^4 \in \rr
\right\}
$$
and ${\rm Aut}(\mathfrak{s}^1_4)$ is the group resulting of composing ${\rm Aut}_{\rm all}(\mathfrak{s}_4)$ with the Lie algebra automorphisms of $\mathfrak{s}_4^1$ of the form
$$
T(e_1)=e_1,\,T(e_2)=e_2+\lambda e_3,\,T(e_3)=e_3+\mu e_1,\,T(e_4)=e_4,\qquad \forall \lambda,\mu\in \mathbb{R}.
$$
Hence, to obtain the orbits of the action ${\rm Aut}(\mathfrak{s}_4)$ on $\mathcal{Y}_{\mathfrak{s}_4}$, it is enough to act $\Lambda^2T_A$ on the orbits of ${\rm Aut}_{\rm all}(\mathfrak{s}_4)$. Note that
$$
\begin{gathered}
\Lambda^2T(e_{12})=e_{12}+\lambda e_{13},\,\,\Lambda^2T(e_{13})=e_{13},\, \Lambda^2T(e_{14})=e_{14},\,\\
\Lambda^2T(e_{23})=e_{23}-e_{12}\mu-\lambda\mu e_{13},\, \Lambda^2T(e_{24})=e_{24}+\lambda e_{34},\, \Lambda^2T_{34}=e_{34}+\mu e_{14},
\end{gathered}$$
for every $\lambda,\mu\in \mathbb{R}$.
This will give the final orbits detailed in Table \ref{Tab:g_orb_2}. Since $(\Lambda^2\mathfrak{s}^1_4)^{\mathfrak{s}^1_4}=0$,  the families of equivalent coboundary coproducts are given by the ones induced by the families of equivalent $r$-matrices.

\subsection{Lie algebra $\mathfrak{s}_{5}$}\label{Sec:s5}

Structure constants for Lie algebra $\mathfrak{s}_5$ are given in Table \ref{Tab:StruCons}. Using this information, one can compute the Schouten brackets between the elements of $\mathfrak{s}_5$, $\Lambda^2 \mathfrak{s}_5$, and $\Lambda^3 \mathfrak{s}_5$ from the information contained in Tables  \ref{Tab:Schoten_1_2}--\ref{Tab:Schoten_1_3}. In particular, we obtain $(\Lambda^2 \mathfrak{s}_5)^{\mathfrak{s}_5} = \langle e_{23}\rangle$ for $\beta = 0$ and $(\Lambda^2 \mathfrak{s}_5)^{\mathfrak{s}_5} = \{0\}$ for $\beta \neq 0$. Moreover, $(\Lambda^3\mathfrak{s}_5)^{\mathfrak{s}_5}=0$ for $\alpha+2\beta\neq 0$ and $(\Lambda^3\mathfrak{s}_5)^{\mathfrak{s}_5}=\langle e_{123}\rangle$ otherwise.

By Remark \ref{Re:DerAlg}, we obtain that the derivations of $\mathfrak{s}_5$ read
$$
\mathfrak{der}(\mathfrak{s}_5):=\left\{\left(
\begin{array}{cccc}
 \mu_{11} & 0 & 0 & \mu_{14} \\
0 & \mu_{22} & \mu_{23} & \mu_{24} \\
0 & -\mu_{23} & \mu_{22} & \mu_{34} \\
0 & 0 & 0 & 0
\end{array}
\right): \mu_{11}, \mu_{22}, \mu_{23}, \mu_{14}, \mu_{24}, \mu_{34} \in \mathbb{R}\right\},
$$
which give rise to the basis of $V_{\mathfrak{s}_5}$ of the form
\begin{equation*}
\begin{gathered}
X_1 = x_1 \partial_{x_1} + x_2 \partial_{x_2} + x_3 \partial_{x_3}, \quad
X_2 = -x_5 \partial_{x_1} - x_6 \partial_{x_2}, \quad
X_3 = x_1 \partial_{x_1} + x_2 \partial_{x_2} + 2x_4 \partial_{x_4} + x_5 \partial_{x_5} + x_6 \partial_{x_6}, \\
X_4 = x_2 \partial_{x_1} - x_1 \partial_{x_2} + x_6 \partial_{x_5} - x_5 \partial_{x_6}, \quad
X_5 = x_3 \partial_{x_1} - x_6 \partial_{x_4}, \quad
X_6 = x_3 \partial_{x_2} + x_5 \partial_{x_4}.
\end{gathered}
\end{equation*}

For an element $r \in \Lambda^2 \mathfrak{s}_5$, we get
\begin{equation*}
\begin{split}
[r,r]&= 2[x_1 x_5 + (\alpha + \beta)x_1 x_6 - (\alpha + \beta)x_2 x_5 + x_2 x_6 + 2\beta x_3 x_4]e_{123} + 2[(\beta - \alpha)x_3 x_5 + x_3 x_6]e_{124} \\
&+ 2[-x_3 x_5 + (\beta - \alpha)x_3 x_6]e_{134} - 2(x_5^2+x_6^2)e_{234}.
\end{split}
\end{equation*}

If $\alpha + 2\beta \neq 0$, then $(\Lambda^3 \mathfrak{s}_5)^{\mathfrak{s}_5} = 0$. Thus, the mCYBE and the CYBE are equal in this case and they read
$$
\beta x_3 x_4 = 0, \quad x_5 = 0, \quad x_6 = 0.
$$

For $\alpha + 2\beta = 0$, we have $(\Lambda^3 \mathfrak{s}_5)^{\mathfrak{s}_5} = \langle e_{123} \rangle$ and the mCYBE reads
$$
 x_5 = 0, \quad x_6 = 0.
$$

Since $x_3$ is the only brick for $\mathfrak{s}_5$, we start our Darboux tree with the cases $x_3 = 0$ and $x_3 \neq 0$. The full Darboux tree is presented below.

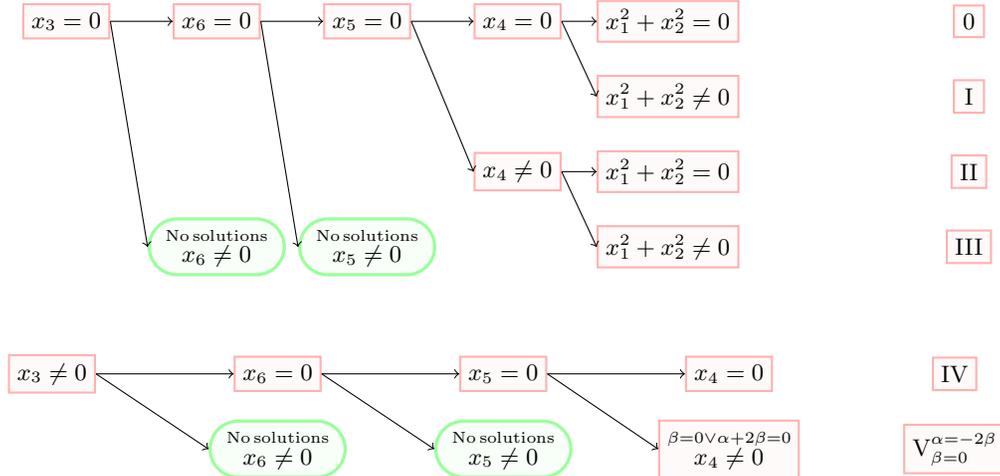
\begin{center}
{\small
\begin{tikzpicture}[
roundnode/.style={rounded rectangle, draw=green!40, fill=green!3, very thick, minimum size=2mm},
squarednode/.style={rectangle, draw=red!30, fill=red!2, thick, minimum size=4mm}
]
\node[squarednode] (brick) at (0,0) {$x_3=0$};

\node[squarednode] (u)  at (2,0) {$x_6=0$};
\node[roundnode] (d)  at (2,-3) {$\stackrel{{\rm No \, solutions}}{x_6 \neq 0}$};

\node[squarednode] (uu)  at (4,0) {$x_5=0$};
\node[roundnode] (ud)  at (4,-3) {$\stackrel{{\rm No\, solutions}}{x_5 \neq 0}$};

\node[squarednode] (uuu)  at (6,0) {$x_4=0$};
\node[squarednode] (uud)  at (6,-2) {$x_4 \neq 0$};

\node[squarednode] (uuuu)  at (8,0) {$x_1^2 + x_2^2=0$};
\node[squarednode] (uuud)  at (8,-1) {$x_1^2 + x_2^2 \neq 0$};
\node[squarednode] (uudu)  at (8,-2) {$x_1^2 + x_2^2=0$};
\node[squarednode] (uudd)  at (8,-3) {$x_1^2 + x_2^2 \neq 0$};

\node[squarednode] (0) at (12,0) {0};
\node[squarednode] (I) at (12,-1) {I};
\node[squarednode] (II) at (12,-2) {II};
\node[squarednode] (III) at (12,-3) {III};

\draw[->] (brick.east) -- (u.west);
\draw[->] (brick.east) -- (d.west);

\draw[->] (u.east) -- (uu.west);
\draw[->] (u.east) -- (ud.west);

\draw[->] (uu.east) -- (uuu.west);
\draw[->] (uu.east) -- (uud.west);

\draw[->] (uuu.east) -- (uuuu.west);
\draw[->] (uuu.east) -- (uuud.west);
\draw[->] (uud.east) -- (uudu.west);
\draw[->] (uud.east) -- (uudd.west);

\end{tikzpicture}

\vspace{1cm}

\begin{tikzpicture}[
roundnode/.style={rounded rectangle, draw=green!40, fill=green!3, very thick, minimum size=2mm},
squarednode/.style={rectangle, draw=red!30, fill=red!2, thick, minimum size=4mm}
]
\node[squarednode] (brick) at (0,0) {$x_3\neq 0$};

\node[squarednode] (u)  at (3,0) {$x_6=0$};
\node[roundnode] (d)  at (3,-1) {$\stackrel{{\rm No \, solutions}}{x_6 \neq 0}$};

\node[squarednode] (uu)  at (6,0) {$x_5=0$};
\node[roundnode] (ud)  at (6,-1) {$\stackrel{{\rm No \, solutions}}{x_5 \neq 0}$};

\node[squarednode] (uuu)  at (9,0) {$x_4=0$};
\node[squarednode] (uud)  at (9,-1) {$\stackrel{ \beta = 0 \lor \alpha + 2\beta = 0}{x_4 \neq 0}$};

\node[squarednode] (0) at (12,0) {IV};
\node[squarednode] (I) at (12,-1) {${\rm V}^{\alpha = -2\beta}_{\beta = 0}$};

\draw[->] (brick.east) -- (u.west);
\draw[->] (brick.east) -- (d.west);

\draw[->] (u.east) -- (uu.west);
\draw[->] (u.east) -- (ud.west);

\draw[->] (uu.east) -- (uuu.west);
\draw[->] (uu.east) -- (uud.west);

\end{tikzpicture}
}
\end{center}

Orbits of ${\rm Aut}_c(\mathfrak{s}_5)$ are given by the connected parts of the loci of the above Darboux tree. Solutions are described in Table \ref{Tab:g_orb_1}.

The automorphism group of $\mathfrak{s}_5$ reads
$$
{\rm Aut}(\mathfrak{s}_5) = \left\{
\left(
\begin{array}{cccc}
T^1_1 & 0 & 0 & T^4_1 \\
0 & T^2_2 & T^3_2 & T^4_2 \\
0 & -T^3_2 & T^2_2 & T^4_3 \\
0 & 0 & 0 & 1
\end{array}
\right): T^1_1\in \mathbb{R} , (T^2_2)^2 + (T^3_2)^2>0, \, T_2^2,T_2^3,T^4_1, T^4_2, T^4_3 \in \rr
\right\}.
$$

Since the subgroup of $GL(2,\mathbb{R})$ of matrices of the form
$$
\left(\begin{array}{cc}T^2_2&T^3_2\\-T^3_2&T^2_2\end{array}\right)\in GL(2,\mathbb{R}),\qquad (T^2_2)^2+(T^3_2)^2>0,$$
can be parametrised via $\phi\in [0,2\pi[$ and $\mu:=[(T^2_2)^2+(T^3_2)^2]^{1/2}\in \mathbb{R}_+$ by setting $T^2_2=\mu \cos\phi$ and $T^3_2=\mu \sin \phi$, there are two connected components of ${\rm Aut}(\mathfrak{s}_5)$. A representative of each connected part and its lift to $\Lambda^2\mathfrak{s}_5$ read
\begin{equation}\label{aut_s5_con}
T_\lambda:=\left(
\begin{array}{cccc}
\lambda & 0 & 0 & 0 \\
0 & 1 & 0 & 0 \\
0 & 0 & 1 & 0 \\
0 & 0 & 0 & 1
\end{array}
\right) \Longrightarrow \Lambda^2T_\lambda:=\left(
\begin{array}{cccccc}
\lambda  & 0 & 0 & 0 &0&0 \\
0 & \lambda  & 0 & 0 &0&0 \\
0 & 0 & \lambda  & 0 &0&0 \\
0 & 0 & 0 & 1 &0&0 \\
0 & 0 & 0 & 0& 1 &0 \\
0 & 0 & 0 & 0&0& 1 \\
\end{array}
\right), \quad \lambda  \in \{\pm 1\}.
\end{equation}

Using techniques from previous sections, we can easily  verify whether the orbits of ${\rm Aut}_c(\mathfrak{s}_5)$ are additionally connected by a Lie algebra automorphism of $\mathfrak{s}_5$ via $\Lambda^2T_\lambda$. Our results are summarised in Table \ref{Tab:g_orb_1}. Moreover, for $\beta \neq 0$, the classes of coboundary coproducts are induced from the following families of $r$-matrices: 
$$
{a)}\,{\rm I},\,{b)}\,{\rm II}_+,{ c)}\,{\rm II}_-,\,{d)}\,{\rm III}_+,\,{e)}\,{\rm III}_{-},\,{f)}\,{\rm IV},\,{g)}\,{\rm V}^{\alpha=-2\beta}_{+},\,{h)}\,{\rm V}^{\alpha=-2\beta}_{-}.
$$
Meanwhile, for $\beta=0$, the list of families of equivalents coboundary coproducts read
$$
{a)}\, {\rm II}_{\pm},\,{b)}\, {\rm I},{\rm III}_{\pm}\,, c)\, {\rm IV},{\rm V}_{\pm}^{\beta=0}.
$$

\subsection{Lie algebra $\mathfrak{s}_{6}$}\label{Se:s6}

Structure constants of Lie algebra $\mathfrak{s}_6$ are given in Table \ref{Tab:StruCons}. Using this information, we can compute the Schouten brackets of elements of $\mathfrak{s}_6$, $\Lambda^2 \mathfrak{s}_6$, and $\Lambda^3 \mathfrak{s}_6$ (see Tables  \ref{Tab:Schoten_1_2}--\ref{Tab:Schoten_1_3}). As in previous sections, Tables \ref{Tab:StruCons},\ref{Tab:Schoten_1_2}--\ref{Tab:Schoten_1_3} contain the necessary information to accomplish the following calculations. In particular, we have  $(\Lambda^3\mathfrak{s}_6)^{\mathfrak{s}_6}=\langle e_{123}\rangle$ while $(\Lambda^2\mathfrak{s}_6)^{\mathfrak{s}_6}=\{0\}$.

By Remark \ref{Re:DerAlg}, one obtains that
$$
\mathfrak{der}(\mathfrak{s}_6):=\left\{\left(
\begin{array}{cccc}
\mu_{11} & \mu_{12} & \mu_{13} & \mu_{14} \\
0 & \mu_{22} & 0 & -\mu_{13} \\
0 & 0 & \mu_{11} - \mu_{22} & -\mu_{12} \\
0 & 0 & 0 & 0
\end{array}
\right): \mu_{11}, \mu_{12}, \mu_{13}, \mu_{14}, \mu_{22} \in \mathbb{R}\right\},
$$
which gives rise to the basis of $V_{\mathfrak{s}_6}$ of the form
\begin{equation*}
\begin{gathered}
X_1 = x_1 \partial_{x_1} + 2x_2 \partial_{x_2} + x_3 \partial_{x_3} + x_4 \partial_{x_4} + x_6 \partial_{x_6}, \quad
X_2 = (-x_3 + x_4) \partial_{x_2} + x_5 \partial_{x_3} - x_5 \partial_{x_4}, \\
X_3 = (-x_3 - x_4) \partial_{x_1} + x_6 \partial_{x_3} + x_6 \partial_{x_4}, \quad
X_4 = -x_5 \partial_{x_1} - x_6 \partial_{x_2}, \quad
X_5 = x_1 \partial_{x_1} - x_2 \partial_{x_2} + x_5 \partial_{x_5} - x_6 \partial_{x_6}.
\end{gathered}
\end{equation*}

For an element $r \in \Lambda^2 \mathfrak{s}_6$, we get
$$
[r,r] = 2(x_1 x_6 + x_2 x_5 + x_4^2)e_{123} + 2(x_3 + x_4)x_5 e_{124} + 2(x_4-x_3)x_6 e_{134} - 4x_5 x_6 e_{234}.
$$
Since $(\Lambda^3 \mathfrak{s}_6)^{\mathfrak{s}_6} = \langle e_{123} \rangle$, the mCYBE reads
$$
(x_3 + x_4) x_5 = 0\quad (x_4 - x_3) x_6 = 0, \quad x_5 x_6 = 0,
$$
whereas the CYBE is
$$
x_1 x_6 + x_2 x_5 + x_4^2 = 0, \quad (x_3 + x_4) x_5 = 0, \quad (x_4 - x_3) x_6 = 0, \quad x_5 x_6 = 0.
$$
Note that the CYBE obtained above is exactly the result obtained in \cite[eq. (3.7)]{BH96} under the substitution $x_5=-\alpha_+$, $x_6=-\alpha_-$, $x_4=\xi$ and $x_3=\vartheta$. 

Since $x_5, x_6$ are the bricks of $\mathfrak{s}_6$, our Darboux tree starts with cases $x_i = 0$ and $x_i \neq 0$, for $i \in \{5,6\}$. The full Darboux tree is presented below.

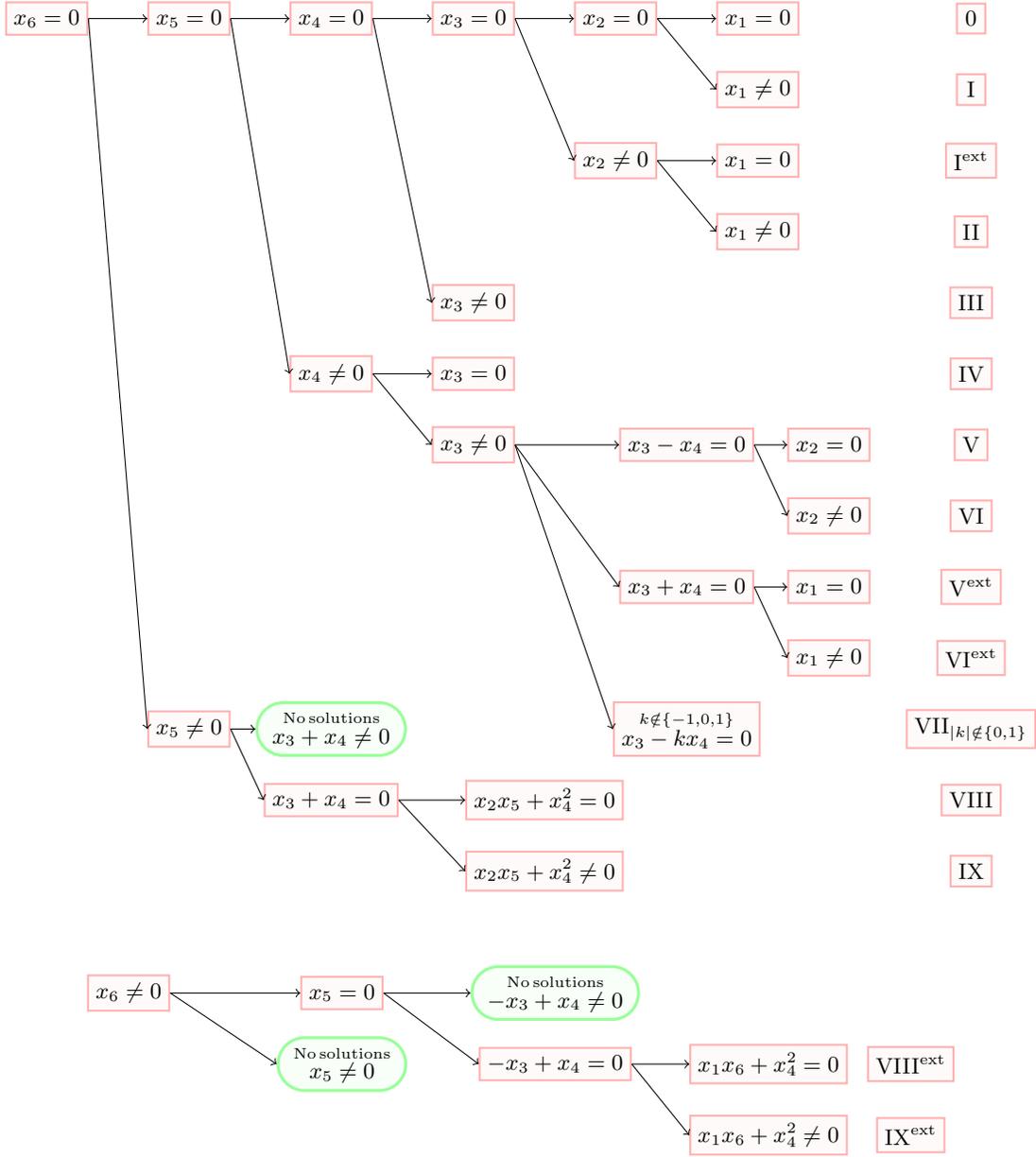
\begin{center}
{\small
\begin{tikzpicture}[
roundnode/.style={rounded rectangle, draw=green!40, fill=green!3, very thick, minimum size=2mm},
squarednode/.style={rectangle, draw=red!30, fill=red!2, thick, minimum size=4mm}
]
\node[squarednode] (brick) at (0,0) {$x_6=0$};

\node[squarednode] (u)  at (2,0) {$x_5=0$};
\node[squarednode] (d)  at (2,-10) {$x_5 \neq 0$};

\node[squarednode] (uu)  at (4,0) {$x_4=0$};
\node[squarednode] (ud)  at (4,-5) {$x_4 \neq 0$};
\node[roundnode] (du)  at (4,-10) {$\stackrel{\tiny{\rm No \, solutions}}{x_3 + x_4 \neq 0}$};
\node[squarednode] (dd)  at (4,-11) {$x_3 + x_4 = 0$};

\node[squarednode] (uuu)  at (6,0) {$x_3=0$};
\node[squarednode] (uud)  at (6,-4) {$x_3 \neq 0$};
\node[squarednode] (udu)  at (6,-5) {$x_3 = 0$};
\node[squarednode] (udd)  at (6,-6) {$x_3 \neq 0$};

\node[squarednode] (uuuu)  at (8,0) {$x_2=0$};
\node[squarednode] (uuud)  at (8,-2) {$x_2 \neq 0$};
\node[squarednode] (uddu)  at (9,-6) {$x_3 - x_4 = 0$};
\node[squarednode] (udddi)  at (9,-8) {$x_3 + x_4 = 0$};
\node[squarednode] (udddiid)  at (7,-11) {$x_2x_5+x_4^2=0$};
\node[squarednode] (udddiii)  at (7,-12) {$x_2x_5+x_4^2\neq 0$};
\node[squarednode] (udddii)  at (9,-10) {$\stackrel{k \notin \{-1,0,1\}}{x_3 - kx_4 = 0}$};

\node[squarednode] (uuuuu)  at (10,0) {$x_1 = 0$};
\node[squarednode] (uuuud)  at (10,-1) {$x_1 \neq 0$};
\node[squarednode] (uuudu)  at (10,-2) {$x_1 = 0$};
\node[squarednode] (uuudd)  at (10,-3) {$x_1 \neq 0$};
\node[squarednode] (udduu)  at (11,-6) {$x_2 = 0$};
\node[squarednode] (uddud)  at (11,-7) {$x_2 \neq 0$};
\node[squarednode] (udddiu)  at (11,-8) {$x_1 = 0$};
\node[squarednode] (udddid)  at (11,-9) {$x_1 \neq 0$};

\node[squarednode] (0) at (13,0) {0};
\node[squarednode] (I) at (13,-1) {I};
\node[squarednode] (II) at (13,-2) {I$^{\rm ext}$};
\node[squarednode] (III) at (13,-3) {II};
\node[squarednode] (IV) at (13,-4) {III};
\node[squarednode] (V) at (13,-5) {IV};
\node[squarednode] (VI) at (13,-6) {V};
\node[squarednode] (VII) at (13,-7) {VI};
\node[squarednode] (VII) at (13,-8) {${\rm V}^{{\rm ext}}$};
\node[squarednode] (IX) at (13,-9) {${\rm VI}^{{\rm ext}}$};
\node[squarednode] (X) at (13,-10) {VII$_{|k| \notin \{0, 1\}}$};
\node[squarednode] (XI) at (13,-11) {VIII};
\node[squarednode] (XII) at (13,-12) {IX};
\draw[->] (brick.east) -- (u.west);
\draw[->] (brick.east) -- (d.west);

\draw[->] (u.east) -- (uu.west);
\draw[->] (u.east) -- (ud.west);
\draw[->] (d.east) -- (du.west);
\draw[->] (d.east) -- (dd.west);

\draw[->] (uu.east) -- (uuu.west);
\draw[->] (uu.east) -- (uud.west);
\draw[->] (ud.east) -- (udu.west);
\draw[->] (ud.east) -- (udd.west);

\draw[->] (uuu.east) -- (uuuu.west);
\draw[->] (uuu.east) -- (uuud.west);
\draw[->] (udd.east) -- (uddu.west);
\draw[->] (udd.east) -- (udddi.west);
\draw[->] (udd.east) -- (udddii.west);

\draw[->] (uuuu.east) -- (uuuuu.west);
\draw[->] (uuuu.east) -- (uuuud.west);
\draw[->] (uuud.east) -- (uuudu.west);
\draw[->] (uuud.east) -- (uuudd.west);
\draw[->] (uddu.east) -- (udduu.west);
\draw[->] (uddu.east) -- (uddud.west);
\draw[->] (udddi.east) -- (udddiu.west);
\draw[->] (udddi.east) -- (udddid.west);

\draw[->] (dd.east) -- (udddiid.west);
\draw[->] (dd.east) -- (udddiii.west);

\end{tikzpicture}

\vspace{1cm}

\begin{tikzpicture}[
roundnode/.style={rounded rectangle, draw=green!40, fill=green!3, very thick, minimum size=2mm},
squarednode/.style={rectangle, draw=red!30, fill=red!2, thick, minimum size=4mm}
]
\node[squarednode] (brick) at (0,0) {$x_6 \neq 0$};

\node[squarednode] (u)  at (3,0) {$x_5 = 0$};
\node[roundnode] (d)  at (3,-1) {$\stackrel{\tiny{\rm No \, solutions}}{x_5 \neq 0}$};

\node[roundnode] (uu)  at (6,0) {$\stackrel{\tiny{\rm No \, solutions}}{- x_3 + x_4 \neq 0}$};
\node[squarednode] (ud)  at (6,-1) {$- x_3 + x_4 = 0$};

\node[squarednode] (udu)  at (9,-1) {$ x_1 x_6 + x_4^2 = 0$};
\node[squarednode] (udd)  at (9,-2) {$ x_1 x_6 + x_4^2 \neq 0$};

\node[squarednode] (XII) at (11,-1) {VIII$^{\rm ext}$};
\node[squarednode] (XIII) at (11,-2) {${\rm IX}^{{\rm ext}}$};

\draw[->] (brick.east) -- (u.west);
\draw[->] (brick.east) -- (d.west);

\draw[->] (u.east) -- (uu.west);
\draw[->] (u.east) -- (ud.west);

\draw[->] (ud.east) -- (udu.west);
\draw[->] (ud.east) -- (udd.west);

\end{tikzpicture}
}
\end{center}
The connected parts of the subsets denoted in the above Darboux tree are the orbits of ${\rm Aut}_c(\mathfrak{s}_6)$. 

The Lie algebra automorphism group of $\mathfrak{s}_6$ reads
$$
{\rm Aut}(\mathfrak{s}_6) = \left\{
\left(
\begin{array}{cccc}
-T^3_2 T^2_3 & T^4_2 T_3^2 & T^4_3 T^3_2 & T^4_1 \\
0 & 0 & T^3_2 & T^4_2 \\
0 & T^2_3 & 0 & T^4_3 \\
0 & 0 & 0 & -1
\end{array}
\right),\left(
\begin{array}{cccc}
T^2_2 T^3_3 & -T^2_2 T^4_3 & -T^4_2 T^3_3 & T^4_1 \\
0 & T^2_2 & 0 & T^4_2 \\
0 & 0 & T^3_3 & T^4_3 \\
0 & 0 & 0 & 1
\end{array}
\right):
\begin{array}{c}
 T^2_3,T^2_2 \in \rz, \\
 T^3_2,T^3_3 \in \rz, \\
 T^4_1, T^4_2, T^4_3 \in \rr
\end{array}
\right\}.
$$

One element of each connected component of ${\rm Aut}(\mathfrak{s}_6)$, which are eight, and their lifts to $\Lambda^2\mathfrak{s}_6$ are given by
\begin{equation}\label{aut_s6_con}
\left(
\begin{array}{cccc}
\mp\lambda_1 \lambda_2 & 0 & 0 & 0 \\
0 & \theta(\mp 1)\lambda_1 & \theta(\pm 1)\lambda_2 & 0 \\
0 & \theta(\pm 1)\lambda_1 & \theta(\mp 1)\lambda_2 & 0 \\
0 & 0 & 0 & \mp 1
\end{array}
\right) \Rightarrow {\small
\left(
\begin{array}{cccccc}
\theta(\mp 1)\lambda_2 & -\theta(\pm 1)\lambda_1  & 0 & 0 &0&0 \\
 -\theta(\pm 1) \lambda_2 & \theta(\mp 1)\lambda_1 & 0 & 0 &0&0 \\
0 & 0 & \lambda_1 \lambda_2 & 0 &0&0 \\
0 & 0 & 0 & \mp \lambda_1\lambda_2&0&0 \\
0 & 0 & 0 & 0&\theta(\mp 1)\lambda_1&-\theta(\pm 1)\lambda_2 \\
0 & 0 & 0 & 0&-\theta(\pm 1)\lambda_1 &\theta(\mp 1)\lambda_2 \\
\end{array}
\right)},
\end{equation}
where $\theta(x)$ stands for the Heaviside function and $\lambda_1,\lambda_2\in \{1,-1\}$. 
By using the previous information, we can verify whether some of the orbits of the action of ${\rm Aut}_c(\mathfrak{s}_6)$ on $\mathcal{Y}_{\mathfrak{s}_6}$ can be connected by an element of ${\rm Aut}(\mathfrak{s}_6)$, which gives rise to the classes of equivalent $r$-matrices up to the action of Lie algebra automorphisms of $\mathfrak{s}_6$. Since $(\Lambda^2\mathfrak{s}_6)^{\mathfrak{s}_6}=0$, each class of equivalent $r$-matrices gives rise to a separate class of coboundary Lie bialgebras on $\mathfrak{s}_6$. Hence, the classes of equivalent coboundary Lie bialgebras are given by
$$
{ a)}\,{\rm I},{\rm I}^{\rm ext}\,{  b)}\, {\rm II},\,{  c)},\,{\rm III},\,{  d)}\,{\rm IV},\,{  e)}\,{\rm V},{\rm V}^{\rm ext}\,{  f})\,{\rm VI,VI^{ext}},\,{  g)}\,{\rm VII}_{|k|\neq\{0,1\}},\,{  h)}\,{\rm VIII,VIII^{ext}},\, {  i)}\,{\rm IX,IX^{ext}}.
$$

\subsection{Lie algebra $\mathfrak{s}_{7}$}\label{Sec:s7}

The structure constants of Lie algebra $\mathfrak{s}_7$ are given in Table \ref{Tab:StruCons}. Using this information, we can compute some relevant Schouten brackets between the elements of the bases of $\mathfrak{s}_7$, $\Lambda^2 \mathfrak{s}_7$, and $\Lambda^3 \mathfrak{s}_7$ (see Tables  \ref{Tab:Schoten_1_2}--\ref{Tab:Schoten_1_3}). Note that, from these tables, we get $(\Lambda^2\mathfrak{s}_7)^{\mathfrak{s}_7}=0$ and $(\Lambda^3\mathfrak{s}_7)^{\mathfrak{s}_7}=\langle e_{123}\rangle$.

By Remark \ref{Re:DerAlg}, one obtains that
$$
\mathfrak{der}(\mathfrak{s}_7):=\left\{\left(
\begin{array}{cccc}
\mu_{11} & \mu_{12} & \mu_{13} & \mu_{14} \\
0 & \frac{1}{2}\mu_{11} & \mu_{23} & \mu_{12} \\
0 & -\mu_{23} & \frac{1}{2}\mu_{11} & \mu_{13} \\
0 & 0 & 0 & 0
\end{array}
\right):  \mu_{11}, \mu_{12}, \mu_{13}, \mu_{14}, \mu_{23} \in \mathbb{R}\right\},
$$
which give rise to the basis of $V_{\mathfrak{s}_7}$ of the form
\begin{equation*}
\begin{gathered}
X_1 = \frac{3}{2} x_1 \partial_{x_1} + \frac{3}{2} x_2 \partial_{x_2} + x_3 \partial_{x_3} + x_4 \partial_{x_4} + \frac{1}{2} x_5 \partial_{x_5} + \frac{1}{2} x_6 \partial_{x_6}, \quad
X_2 = x_3 \partial_{x_1} + x_4 \partial_{x_2} + x_5 \partial_{x_3} - x_6 \partial_{x_4}, \\
X_3 = - x_4 \partial_{x_1} + x_3 \partial_{x_2} + x_6 \partial_{x_3} + x_5 \partial_{x_4}, \quad
X_4 = -x_5 \partial_{x_1} - x_6 \partial_{x_2}, \quad
X_5 = x_2 \partial_{x_1} - x_1 \partial_{x_2} + x_6 \partial_{x_5} - x_5 \partial_{x_6}.
\end{gathered}
\end{equation*}

Then,
$$
[r,r] = 2 (x_4^2 + x_1 x_5 + x_2 x_6) e_{123} + 2 (x_4 x_5 + x_3 x_6) e_{124} + 2 (x_4 x_6 - x_3 x_5) e_{134} - 2 (x_5^2 + x_6^2) e_{234}.
$$
Since $(\Lambda^3\mathfrak{s}_7)^{\mathfrak{s}_7}=\langle e_{123}\rangle$, then the mCYBE reads
$$
x_5=x_6=0.
$$
Meanwhile, the CYBE takes the form 
$$
x_4=x_5=x_6=0.
$$
The Darboux tree is presented below.

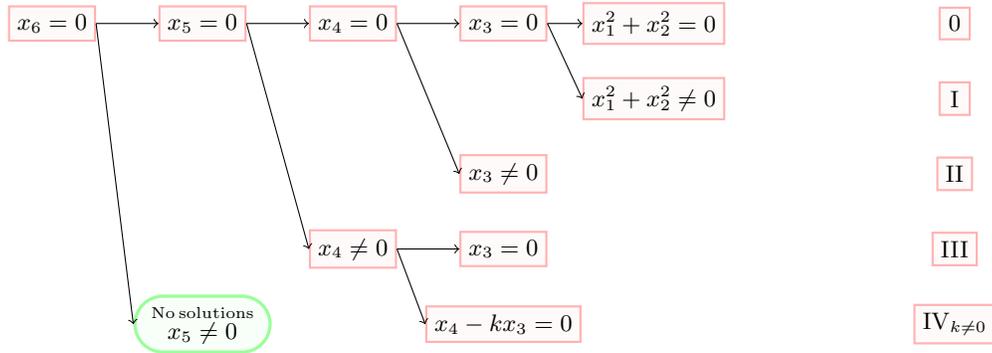
\begin{center}
{\small
\begin{tikzpicture}[
roundnode/.style={rounded rectangle, draw=green!40, fill=green!3, very thick, minimum size=2mm},
squarednode/.style={rectangle, draw=red!30, fill=red!2, thick, minimum size=4mm}
]
\node[squarednode] (brick) at (0,0) {$x_6=0$};

\node[squarednode] (u)  at (2,0) {$x_5=0$};
\node[roundnode] (d)  at (2,-4) {$\stackrel{\tiny{\rm No \, solutions}}{x_5 \neq 0}$};

\node[squarednode] (uu)  at (4,0) {$x_4=0$};
\node[squarednode] (ud)  at (4,-3) {$x_4 \neq 0$};

\node[squarednode] (uuu)  at (6,0) {$x_3=0$};
\node[squarednode] (uud)  at (6,-2) {$x_3 \neq 0$};
\node[squarednode] (udu)  at (6,-3) {$x_3 = 0$};
\node[squarednode] (udd)  at (6,-4) {$x_4 - k x_3 = 0$};

\node[squarednode] (uuuu)  at (8,0) {$x_1^2 + x_2^2 = 0$};
\node[squarednode] (uuud)  at (8,-1) {$x_1^2 + x_2^2 \neq 0$};

\node[squarednode] (0) at (12,0) {0};
\node[squarednode] (I) at (12,-1) {I};
\node[squarednode] (II) at (12,-2) {II};
\node[squarednode] (III) at (12,-3) {III};
\node[squarednode] (IV) at (12,-4) {${\rm IV}_{k\neq 0}$};

\draw[->] (brick.east) -- (u.west);
\draw[->] (brick.east) -- (d.west);

\draw[->] (u.east) -- (uu.west);
\draw[->] (u.east) -- (ud.west);

\draw[->] (uu.east) -- (uuu.west);
\draw[->] (uu.east) -- (uud.west);
\draw[->] (ud.east) -- (udu.west);
\draw[->] (ud.east) -- (udd.west);

\draw[->] (uuu.east) -- (uuuu.west);
\draw[->] (uuu.east) -- (uuud.west);

\end{tikzpicture}
}
\end{center}

The connected parts of the subsets denoted in the above Darboux tree are the orbits of ${\rm Aut}_c(\mathfrak{s}_{7})$ in $\mathcal{Y}_{\mathfrak{s}_7}$. By using the extension of the automorphism group ${\rm Aut}(\mathfrak{s}_{7})$ to $\Lambda^2 \mathfrak{s}_{7}$, we can verify whether some of these parts are additionally connected by a Lie algebra automorphism of $\mathfrak{s}_{7}$. The results are summarised in Table \ref{Tab:g_orb_2}. As before, each family of $r$-matrices induces a separate class of coboundary coproducts.

The Lie algebra automorphism group of $\mathfrak{s}_7$ reads
\begin{equation*}
\begin{split}
{\rm Aut}(\mathfrak{s}_7) &= \left\{
\left(
\begingroup
\setlength\arraycolsep{4.5pt}
\begin{array}{cccc}
\pm [(T^2_2)^2 + (T^2_3)^2] & \pm T^4_2 T^2_2 - T^4_3 T^2_3 & T^4_3 T^2_2 \pm  T^4_2 T^2_3 & T^4_1 \\
0 & T^2_2 & T^2_3 & T^4_2 \\
0 &\mp T^2_3 & \pm T^2_2 & T^4_3 \\
0 & 0 & 0 & \pm 1
\end{array}
\endgroup
\right):
\begingroup
\setlength\arraycolsep{4.5pt}
\begin{array}{c}
(T^2_2)^2+(T_3^2)^2 \in \mathbb{R} \backslash \{ 0\} \\
T^2_2,T^2_3,T^4_1, T^4_2, T^4_3 \in \rr
\end{array}
\endgroup
\right\}. 
\end{split}
\end{equation*}
Since the two subsets of $GL(2,\mathbb{R})$ of matrices of the form
$$
\left(\begin{array}{cc}T^2_2&T^2_3\\\mp T^2_3&\pm T^2_2\end{array}\right)\in GL(2,\mathbb{R}),\qquad (T^2_2)^2+(T^2_3)^2>0,$$
can be parametrised via $\phi\in [0,2\pi[$ and $\mu:=[(T^2_2)^2+(T^2_3)^2]^{1/2}\in \mathbb{R}_+$ by setting $T^2_2=\mu \cos\phi$ and $T^2_3=\mu \sin \phi$, one gets that ${\rm Aut}(\mathfrak{s}_7)$ has two connected components. As usual, we only need one element for every connected component of ${\rm Aut}(\mathfrak{s}_7)$ and their lifts to $\Lambda^2\mathfrak{s}_7$, namely
$$
T_{\pm}:=\left(\begin{array}{cccc}
\pm 1& 0 & 0&0 \\
0 & 1 & 0 & 0 \\
0 &0 & \pm 1 &0 \\
0 & 0 & 0 & \pm 1
\end{array}\right)\qquad \Longrightarrow \qquad \Lambda^2T_{\pm}=\left(
\begin{array}{cccccc}
\pm 1 & 0 & 0 & 0 &0&0 \\
0 & 1 & 0 & 0 &0&0 \\
0 & 0 & 1 & 0 &0&0 \\
0 & 0 & 0 & \pm1&0&0 \\
0 & 0 & 0 & 0&\pm 1&0 \\
0 & 0 & 0 & 0&0&1 \\
\end{array}
\right).
$$
	
	Taking this into account and since $(\Lambda^2\mathfrak{s}_7)^{\mathfrak{s}_7}=0$, we  obtain that the families or equivalent coboundary Lie bialgebras are given by
	$$
	{\rm a)}\, {\rm I},\,{\rm b)}\,{\rm II}_{+},\,{\rm c)}\,{\rm II}_-,\,{\rm 
	d)}\,{\rm III},\,{\rm e)}\,({\rm IV}_+)_{|k|\in \mathbb{R}_+}\,{\rm f)}\,({\rm IV}_-)_{|k|\in \mathbb{R}_+}.
	$$

\subsection{Lie algebra $\mathfrak{s}_{8}$}

The structure constants for the Lie algebra $\mathfrak{s}_{8}$ are given in Table \ref{Tab:StruCons}. Recall that $\alpha\in ]-1,1]\backslash\{0\}$. As in the previous cases, one can obtain some relevant Schouten brackets between the basis elements of $\mathfrak{s}_8$, $\Lambda^2\mathfrak{s}_8$, and $\Lambda^3\mathfrak{s}_8$ (see Tables  \ref{Tab:Schoten_1_2}--\ref{Tab:Schoten_1_3}). In view of  Tables \ref{Tab:StruCons},\ref{Tab:Schoten_1_2}--\ref{Tab:Schoten_1_3}, one has that  $(\Lambda^2 \mathfrak{s}_8)^{\mathfrak{s}_8} = \langle e_{13}\rangle$ for $\alpha = -\frac{1}{2}$ and $(\Lambda^2 \mathfrak{s}_8)^{\mathfrak{s}_8} = 0$,  otherwise. Moreover, $(\Lambda^3\mathfrak{s}_8)^{\mathfrak{s}_8}=0$.

By Remark \ref{Re:DerAlg}, one obtains that the derivations of $\mathfrak{s}^\alpha_8$, for a fixed value $\alpha\in ]-1,1[\backslash\{0\}$ read
\begin{equation}\label{Eq:Der8a}
\mathfrak{der}(\mathfrak{s}^\alpha_8)=\left\{\left(
\begin{array}{cccc}
\mu_{11} & \mu_{12} & \mu_{13} & \mu_{14} \\
0 & \mu_{22} & 0 & -\mu_{13} \\
0 & 0 & \mu_{11} - \mu_{22} & \alpha \mu_{12} \\
0 & 0 & 0 & 0
\end{array}
\right): \mu_{11}, \mu_{12}, \mu_{13}, \mu_{14}, \mu_{22} \in \mathbb{R}\right\}.
\end{equation}
In case $\alpha=1$, one gets
\begin{equation}\label{Eq:Der8a1}
\mathfrak{der}(\mathfrak{s}^1_8)=\left\{\left(
\begin{array}{cccc}
\mu_{11} & \mu_{12} & \mu_{13} & \mu_{14} \\
0 & \mu_{22} & \mu_{23} & -\mu_{13} \\
0 & \mu_{32} & \mu_{11} - \mu_{22} & \mu_{12} \\
0 & 0 & 0 & 0
\end{array}
\right): \mu_{11}, \mu_{12}, \mu_{13}, \mu_{14}, \mu_{22},\mu_{23},\mu_{32} \in \mathbb{R}\right\}.
\end{equation}
For any value of $\alpha\in ]-1,1]\backslash\{0\}$ and an element $r \in \Lambda^2\mathfrak{s} _{8}$, one obtains
$$
[r,r]=2[(2+ \alpha)x_1x_6 - (1 + 2\alpha)x_2x_5 + (1+\alpha)x_3x_4 + x_4^2]e_{123} + 2(x_4-\alpha x_3)x_5e_{124} + 2x_6(x_4-x_3)e_{134} + 2(\alpha - 1)x_5x_6e_{234}.
$$ 
Since $ (\Lambda^3\mathfrak{s}_{8})^{\mathfrak{s}_{8}}=0$, the mCYBE reads
$$
(2+ \alpha)x_1x_6 - (1 + 2\alpha)x_2x_5 + (1+\alpha)x_3x_4 + x_4^2= 0, \quad (x_4-\alpha x_3)x_5= 0, \quad x_6(x_4-x_3)= 0, \quad (\alpha - 1)x_5x_6 = 0.
$$
Let us now consider two cases given by $\alpha\in ]-1,1[\backslash\{0\}$ and $\alpha=1$. In the first case, the space of derivations gives rise to the basis of $V_{\mathfrak{s}^\alpha_8}$ of the form
	\begin{equation}\label{Eq:Vs8a}
\begin{gathered}
X_1=	x_1\partial_{x_1}+2x_2\partial_{x_2}+x_3\partial_{x_3}+x_4\partial_{x_4}+x_6\partial_{x_6},\qquad  X_2=(x_3\alpha+x_4)\partial_{x_2}+x_5\partial_{x_3}+\alpha x_5\partial_{x_4}, \\ 
X_3=-(x_3+x_4)\partial_{x_1}+x_6\partial_{x_3}+x_6\partial_{x_4},\qquad
X_4=-x_5\partial_{x_1}-x_6\partial_{x_2},\qquad X_5=x_1\partial_{x_1}-x_2\partial_{x_2}+x_5\partial_{x_5}-x_6\partial_{x_6}.
\end{gathered}
\end{equation}
	The following diagram depicts the Darboux families that give us the orbits of ${\rm Aut}_c(\mathfrak{s}^\alpha_{8})$ on $\mathcal{Y}_{\mathfrak{s}^\alpha_{8}}$. Since $x_5, x_6$ are the bricks for every $\mathfrak{s}^\alpha_8$, our Darboux tree starts with the cases $x_i = 0$ and $x_i \neq 0$, for $i \in \{5,6\}$. The full Darboux tree is presented below.
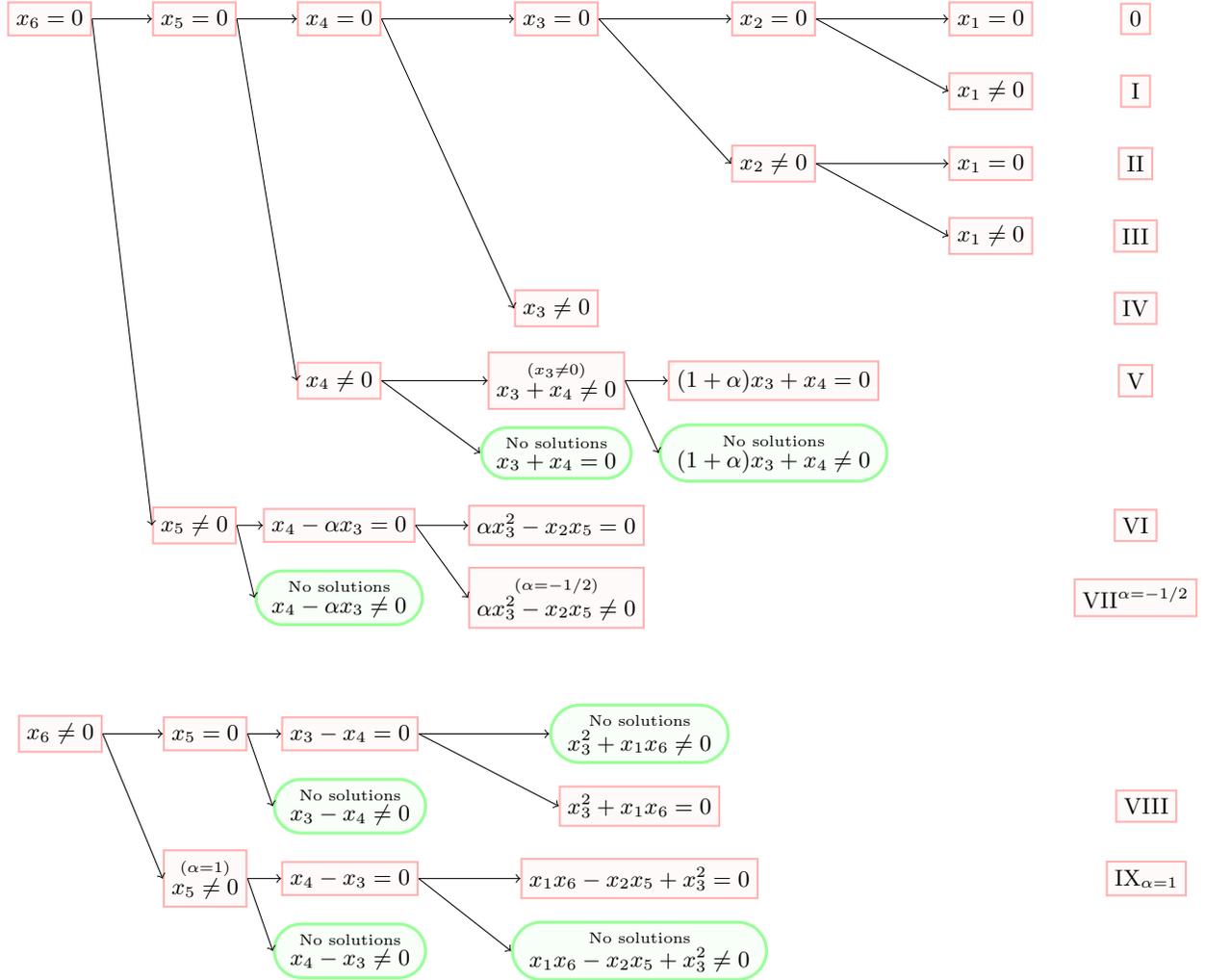
\begin{center}
{\small
	\begin{tikzpicture}[
roundnode/.style={rounded rectangle, draw=green!40, fill=green!3, very thick, minimum size=2mm},
squarednode/.style={rectangle, draw=red!30, fill=red!2, thick, minimum size=4mm}
]
	\node[squarednode] (brick) at (0,0) {$x_6=0$};

	\node[squarednode] (65)  at (2,0) {$x_5=0$};
	\node[squarednode] (65u)  at (2,-7) {$x_5\neq 0$};

	\node[squarednode]   (654)  at (4,0) {$x_4=0$};
	\node[squarednode]   (654u) at (4,-5)  {$x_4\neq 0$};
	\node[squarednode]   (65ue) at (4,-7)  {$x_4-\alpha x_3=0$};
	\node[roundnode]   (65ueu) at (4,-8)  {$\stackrel{\rm No\,\,solutions}{x_4-\alpha x_3\neq 0}$};

	\node[squarednode]   (6543) at (7,0)  {$x_3=0$};
	\node[squarednode]   (6543u) at (7,-4)  {$x_3\neq 0$};
	\node[squarednode]   (654ufu) at (7,-5)  {$\stackrel{(x_3\neq 0)}{x_3+x_4\neq 0}$};
	\node[roundnode]   (654uf) at (7,-6)  {$\stackrel{{\rm No\,\, solutions}}{x_3+x_4=0}$};
	\node[squarednode]   (65uey) at (7,-7)  {$\alpha x_3^2-x_2x_5=0$};
	\node[squarednode]   (65ueyu) at (7,-8)  {$\stackrel{(\alpha=-1/2)}{\alpha x_3^2-x_2x_5\neq 0}$};

	\node[squarednode]   (65432) at (10,0)  {$x_2=0$};
	\node[squarednode]   (65432u) at (10,-2)  {$x_2\neq 0$};
	\node[squarednode]   (654ufug) at (10,-5)  {$(1+\alpha)x_3+x_4=0$};
	\node[roundnode]   (654ufugu) at (10,-6)  {$\stackrel{{\rm No\,\,solutions}}{(1+\alpha)x_3+x_4\neq 0}$};

	\node[squarednode]   (654321) at (13,0)  {$x_1=0$};
	\node[squarednode]   (654321u) at (13,-1)  {$x_1\neq 0$};
	\node[squarednode]   (65432u1) at (13,-2)  {$x_1=0$};
	\node[squarednode]   (65432u1u) at (13,-3)  {$x_1\neq 0$};

	\node[squarednode] (0) at (15,0) {0};
	\node[squarednode] (I) at (15,-1) {I};
	\node[squarednode] (II) at (15,-2) {II};
	\node[squarednode] (III) at (15,-3) {III};
	\node[squarednode] (IV) at (15,-4) {IV};
	\node[squarednode] (V) at (15,-5) {V};
	\node[squarednode] (VI) at (15,-7) {VI};
	\node[squarednode] (VIa) at (15,-8) {VII$^{\alpha=-1/2}$};
	
	\draw[->] (brick.east) -- (65.west);
	\draw[->] (brick.east) -- (65u.west);
	\draw[->] (65.east) -- (654.west);
	\draw[->] (65.east) -- (654u.west);
	\draw[->] (65u.east) -- (65ue.west);
	\draw[->] (65u.east) -- (65ueu.west);
	\draw[->] (65ue.east) -- (65uey.west);
	\draw[->] (65ue.east) -- (65ueyu.west);
	\draw[->] (6543.east)--(65432.west);
	\draw[->] (654.east) -- (6543.west);
\draw[->] (654.east) -- (6543u.west);
	\draw[->] (654u.east) -- (654uf.west);
\draw[->] (654u.east) -- (654ufu.west);
	\draw[->] (654ufu.east) -- (654ufug.west);
	\draw[->] (654ufu.east) -- (654ufugu.west);
	\draw[->] (6543.east) -- (65432u.west);
\draw[->] (65432.east) -- (654321.west);
\draw[->] (65432.east) -- (654321u.west);
\draw[->] (65432u.east) -- (65432u1.west);
\draw[->] (65432u.east) -- (65432u1u.west);
\end{tikzpicture}
	
\vspace{1cm}

		\begin{tikzpicture}[
roundnode/.style={rounded rectangle, draw=green!40, fill=green!3, very thick, minimum size=2mm},
squarednode/.style={rectangle, draw=red!30, fill=red!2, thick, minimum size=4mm}
]
	\node[squarednode] (6u) at (0,0) {$x_6\neq 0$};

	\node[squarednode] (6u5)  at (2,0)  {$x_5=0$};
	\node[squarednode] (6u5u)  at (2,-2) {$\stackrel{(\alpha=1)}{x_5\neq 0}$};

	\node[squarednode]   (6u5e)  at (4,0) {$x_3-x_4=0$};
	\node[roundnode]   (6u5eu) at (4,-1)  {$\stackrel{{\rm No\,\, solutions}}{x_3-x_4\neq 0}$};
	\node[squarednode]   (6u5ue) at (4,-2)  {$x_4-x_3=0$};
	\node[roundnode]   (6u5ueu) at (4,-3)  {$\stackrel{\rm No\,\,solutions}{x_4-x_3\neq 0}$};

	\node[roundnode]   (6u5efu) at (8,0)  {$\stackrel{{\rm No\,\, solutions}}{x_3^2+x_1x_6\neq 0}$};
	\node[squarednode]   (6u5ef) at (8,-1)  {$x_3^2+x_1x_6=0$};
	\node[squarednode]   (6u5ueg) at (8,-2)  {$x_1x_6-x_2x_5+x_3^2=0$};
	\node[roundnode]   (6u5uegu) at (8,-3)  {$\stackrel{{\rm No\,\,solutions}}{x_1x_6-x_2x_5+x_3^2\neq0}$};

	\node[squarednode] (VII) at (15,-1) {VIII};
	\node[squarednode] (VIII) at (15,-2) {IX$_{\alpha=1}$};

	\draw[->] (6u.east) -- (6u5.west);
	\draw[->] (6u.east) -- (6u5u.west);
	\draw[->] (6u5.east) -- (6u5e.west);
	\draw[->] (6u5.east) -- (6u5eu.west);
	\draw[->] (6u5e.east) -- (6u5ef.west);
	\draw[->] (6u5e.east) -- (6u5efu.west);
	\draw[->] (6u5u.east)--(6u5ue.west);
	\draw[->] (6u5u.east) -- (6u5ueu.west);
	\draw[->] (6u5ue.east) -- (6u5ueg.west);
	\draw[->] (6u5ue.east) -- (6u5uegu.west);

	\end{tikzpicture}
	}
\end{center}

The Lie algebra automorphism group of $\mathfrak{s}^\alpha_8$ for each $\alpha \in ]-1,1[\backslash\{0\}$, reads
$$
{\rm Aut}(\mathfrak{s}^\alpha_8) = \left\{
\left(
\begin{array}{cccc}
T^2_2 T^3_3 & T^2_1 & T^3_1 & T^4_1 \\
0 & T^2_2 & 0 & -\frac{T^3_1}{T^3_3} \\
0 & 0 & T^3_3 & \frac{\alpha T^2_1}{T^2_2} \\
0 & 0 & 0 & 1
\end{array}
\right):
\begin{array}{c}
 T^2_2, \, T^3_3 \in \rz \\
T^2_1, T^3_1, T^4_1 \in \rr
\end{array}
\right\}.
$$

In reality, we are only concerned with obtaining one element of ${\rm Aut}(\mathfrak{s}^\alpha_8)$ and its lift to $\Lambda^2\mathfrak{s}_8$ for each one of its connected components. For instance, we can choose
\begin{equation}\label{aut_s8_con}
T_{\lambda_1, \lambda_2} = \left(
\begin{array}{cccc}
\lambda_1 \lambda_2 & 0 & 0 & 0 \\
0 & \lambda_1 & 0 & 0 \\
0 & 0 & \lambda_2 & 0 \\
0 & 0 & 0 & 1
\end{array}
\right)\quad \Longrightarrow \quad
\Lambda^2T_{\lambda_1,\lambda_2}=\left(
\begin{array}{cccccc}
\lambda_2 & 0 & 0 & 0 & 0 & 0 \\
0 & \lambda_1 & 0 & 0 & 0 & 0 \\
0 & 0 & \lambda_1 \lambda_2 & 0 & 0 & 0 \\
0 & 0 & 0 & \lambda_1 \lambda_2 & 0 & 0 \\
0 & 0 & 0 & 0 & \lambda_1 & 0 \\
0 & 0 & 0 & 0 & 0 & \lambda_2 
\end{array}
\right),  \quad \lambda_1, \lambda_2 \in \{\pm 1\}.
\end{equation}

As in previous sections, the maps $\Lambda^2T_{\lambda_1,\lambda_2}$ allow us to identify the orbits of the action of ${\rm Aut} (\mathfrak{s}^\alpha_8)$ on $\mathcal{Y}_{\mathfrak{s}^\alpha_8}$. Our results are presented in Table \ref{Tab:g_orb_2}. Moreover, for $\alpha \neq -\frac{1}{2}$, each family of equivalent $r$-matrices give rise to a separate class of equivalent coboundary Lie bialgebras. For $\alpha = -\frac{1}{2}$, we obtain eight families of equivalent coboundary  Lie bialgebras given by:
 $$
 a) \,{\rm  II}\,({\rm zero\,class}), \quad b)\, {\rm I, III}, \quad c)\, {\rm IV},\quad d)\, {\rm V}, \quad e)\, {\rm VI},  {\rm VII}^{\alpha =- \frac{1}{2}}_-, {\rm VII}^{\alpha = -\frac{1}{2}}_+, \quad f)\, {\rm VIII}.
 $$
Let us know study the case $\alpha=1$ using the previous results. Since (\ref{Eq:Der8a}) for $\alpha=1$ is a Lie subalgebra of the space of derivations for $\mathfrak{s}_8^\alpha$ given in (\ref{Eq:Der8a1}), the Lie algebra spanned by the vector fields of (\ref{Eq:Vs8a}) for $\alpha=1$ is a Lie subalgebra of the Lie algebra of fundamental vector fields of the action of ${\rm Aut}(\mathfrak{s}_8^1)$ on $\Lambda^2\mathfrak{s}_8$. Hence, the loci of the above Darboux families allow us to characterise the strata of the distribution spanned by the vector fields (\ref{Eq:Vs8a}) for $\alpha=1$, which in turn gives as the orbits of a Lie subgroup of ${\rm Aut}(\mathfrak{s}_8^1)$. To easily follow our discussion, we detail that 
$$
{\rm Aut}(\mathfrak{s}^\alpha_8) = \left\{
\left(
\begin{array}{cccc}
T^2_2T^3_3-T_2^3T^2_3 & -T_2^4T_3^2+T^2_2T_3^4 & -T_2^4T^3_3+T_2^3T_3^4 & T^4_1 \\
0 & T^2_2 & T_2^3 & T_2^4 \\
0 & T_3^2 & T^3_3 & T_3^4\\
0 & 0 & 0 & 1
\end{array}
\right):
\begin{array}{c}
 T^2_2T^3_3-T_3^2T^3_2\neq 0, \, T_1^4,
T_2^4, T_3^4\in \rr
\end{array}
\right\}.
$$

To obtain the orbits of ${\rm Aut}(\mathfrak{s}_8^1)$ on $\mathcal{Y}_{\mathfrak{s}_8^1}$, it is enough to write ${\rm Aut}(\mathfrak{s}_8^1)$ as a composition of the previous subgroup with certain Lie algebra automorphisms of ${\rm Aut}(\mathfrak{s}_8^1)$, e.g. the Lie algebra automorphisms $T_A=(\det A){\rm Id}\otimes A\otimes {\rm Id}$ for every $A\in GL(2,\mathbb{R})$. %
Hence, $\Lambda^2T_A=(\det A)A\otimes (\det A){\rm Id}\otimes (\det A){\rm Id}$.
The action of the  $\Lambda^2T_A$ on the loci of the Darboux families of (\ref{Eq:Vs8a}) assuming  $\alpha=1$ allows us to obtain the orbits of the action of ${\rm Aut}(\mathfrak{s}_8^1)$ on $\mathcal{Y}_{\mathfrak{s}_8^1}$. In particular, one obtains the subsets given in Table \ref{Tab:g_orb_2}. 
Since $(\Lambda^2\mathfrak{s}^1_8)^{\mathfrak{s}^1_8}=0$, the classes of equivalent coboundary coproducts for $\mathfrak{s}_8^1$ are given by
$$
a)\, \mathcal{I},\quad b)\,\mathcal{II},\quad c)\, \mathcal{III},\quad d)\,\mathcal{IV}.
$$

\subsection{Lie algebra $\mathfrak{s}_{9}$}

As previously, we use the structure constants for $\mathfrak{s}_9$ in Table \ref{Tab:StruCons} to compute the Schouten brackets given in Tables  \ref{Tab:Schoten_1_2}--\ref{Tab:Schoten_1_3}. Moreover, $(\Lambda^2 \mathfrak{s}_9)^{\mathfrak{s}_9} = 0$ and $(\Lambda^3\mathfrak{s}_9)^{\mathfrak{s}_9}=0$ for every $\alpha>0$. These calculations are enough to verify the remaining results of present subsection.

By Remark \ref{Re:DerAlg}, one obtains that
$$
\mathfrak{der}(\mathfrak{s}^\alpha_9):=\left\{\left(
\begin{array}{cccc}
2 \mu_{22} & \mu_{12} & \mu_{13} & \mu_{14} \\
0 & \mu_{22} & \mu_{23} & \mu_{12} - \alpha \mu_{13} \\
0 & -\mu_{23} & \mu_{22} & \alpha \mu_{12} + \mu_{13} \\
0 & 0 & 0 & 0
\end{array}
\right):  \mu_{12}, \mu_{13}, \mu_{14}, \mu_{22}, \mu_{23} \in \mathbb{R}\right\}.
$$
The obtained derivations give rise to the basis of $V_{\mathfrak{s}^\alpha_9}$ of the form
\begin{equation*}
\begin{gathered}
X_1 = x_3 \partial_{x_1} + (\alpha x_3 + x_4) \partial_{x_2} + x_5 \partial_{x_3} + (\alpha x_5 - x_6) \partial_{x_4}, \quad
X_2 = (- \alpha x_3 - x_4) \partial_{x_1} + x_3 \partial_{x_2} + x_6 \partial_{x_3} + (x_5 + \alpha x_6) \partial_{x_4}, \\
X_3 = -x_5 \partial_{x_1} - x_6 \partial_{x_2}, \quad
X_4 = 3x_1 \partial_{x_1} + 3x_2 \partial_{x_2} + 2x_3 \partial_{x_3} + 2x_4 \partial_{x_4} + x_5 \partial_{x_5} + x_6 \partial_{x_6}, \\
X_5 = x_2 \partial_{x_1} - x_1 \partial_{x_2} + x_6 \partial_{x_5} - x_5 \partial_{x_6}.
\end{gathered}
\end{equation*}

For an element $r \in \Lambda^2 \mathfrak{s}_9$, we get
\begin{equation*}
\begin{split}
[r,r] &= 2(x_1 x_5 + 3\alpha x_1 x_6 - 3\alpha x_2 x_5 + x_2 x_6 + 2\alpha x_3 x_4 + x_4^2)e_{123} + 2(-\alpha x_3 x_5 + x_3 x_6 + x_4 x_5)e_{124} \\
&+ 2(x_4 x_6-x_3 x_5 - \alpha x_3 x_6 )e_{134} - 2(x_5^2 + x_6^2)e_{234}.
\end{split}
\end{equation*}

And since $(\Lambda^3 \mathfrak{s}_9)^{\mathfrak{s}_9} = 0$ for every value of $\alpha$, the mCYBE and the CYBE are equal and read
$$
(2\alpha x_3 + x_4) x_4 = 0, \quad x_5 = 0, \quad x_6 = 0.
$$

The Darboux tree for the class $\mathfrak{s}_9$ is presented below.

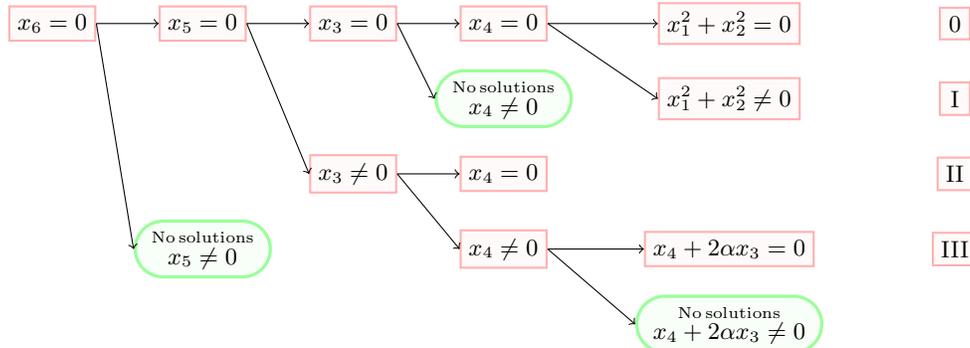
\begin{center}
{\small
\begin{tikzpicture}[
roundnode/.style={rounded rectangle, draw=green!40, fill=green!3, very thick, minimum size=2mm},
squarednode/.style={rectangle, draw=red!30, fill=red!2, thick, minimum size=4mm}
]
	\node[squarednode] (brick) at (0,0) {$x_6 = 0$};

	\node[squarednode] (u)  at (2,0) {$x_5=0$};
	\node[roundnode] (d)  at (2,-3) {$\stackrel{{\rm No \, solutions}}{x_5 \neq 0}$};

	\node[squarednode] (uu)  at (4,0) {$x_3=0$};
	\node[squarednode] (ud)  at (4,-2) {$x_3 \neq 0$};

	\node[squarednode] (uuu)  at (6,0) {$x_4=0$};
	\node[roundnode] (uud)  at (6,-1) {$\stackrel{{\rm No \, solutions}}{x_4 \neq 0}$};
	\node[squarednode] (udu)  at (6,-2) {$x_4=0$};
	\node[squarednode] (udd)  at (6,-3) {$x_4 \neq 0$};

	\node[squarednode] (uuuu)  at (9,0) {$x_1^2 + x_2^2=0$};
	\node[squarednode] (uuud)  at (9,-1) {$x_1^2 + x_2^2 \neq 0$};
	\node[squarednode] (uddu)  at (9,-3) {$x_4 + 2\alpha x_3  = 0$};
    \node[roundnode] (uddd)  at (9,-4) {$\stackrel{{\rm No \, solutions}}{x_4 + 2\alpha x_3  \neq 0}$};

	\node[squarednode] (0) at (12,0) {0};
	\node[squarednode] (I) at (12,-1) {I};
	\node[squarednode] (II) at (12,-2) {II};
	\node[squarednode] (I) at (12,-3) {III};
	
	\draw[->] (brick.east) -- (u.west);
	\draw[->] (brick.east) -- (d.west);

	\draw[->] (u.east) -- (uu.west);
	\draw[->] (u.east) -- (ud.west);

	\draw[->] (uu.east) -- (uuu.west);
	\draw[->] (uu.east) -- (uud.west);
	\draw[->] (ud.east) -- (udu.west);
	\draw[->] (ud.east) -- (udd.west);

	\draw[->] (uuu.east) -- (uuuu.west);
	\draw[->] (uuu.east) -- (uuud.west);
	\draw[->] (udd.east) -- (uddu.west);
	\draw[->] (udd.east) -- (uddd.west);

\end{tikzpicture}
}
\end{center}

If we define $\Delta:=(T^2_2)^2 + (T^2_3)^2$, the Lie algebra automorphisms group of each $\mathfrak{s}^\alpha_9$ reads
$$
{\rm Aut}(\mathfrak{s}^\alpha_9) = \left\{
\left(
\begin{array}{cccc}
\Delta & \frac{T^4_2( T^2_2 + \alpha T^2_3) + T^4_3(\alpha T^2_2 -  T^2_3)}{1+\alpha^2} & \frac{(-\alpha T^2_2 + T^2_3) T^4_2 + T^4_3 (T^2_2 + \alpha  T^2_3)}{1+\alpha^2} & T^4_1 \\
0 & T^2_2 & T^2_3 & T^4_2 \\
0 & -T^2_3 & T^2_2 & T^4_3 \\
0 & 0 & 0 & 1
\end{array}
\right): 
\begin{array}{c}
\Delta\in \mathbb{R}_+, \\
T^2_2,T^2_3,T^4_1, T^4_2, T^4_3 \in \rr
\end{array}
\right\}.
$$

Using ideas from Section \ref{Sec:s5}, we obtain that each ${\rm Aut}(\mathfrak{s}^\alpha_9)$ has one connected component.  Consequently, the orbits of ${\rm Aut}(\mathfrak{s}^\alpha_9)$ on $\Lambda^2\mathfrak{s}^\alpha_{9}$ are the strata of $\mathscr{E}_{\mathfrak{s}^\alpha_9}$. Since $(\Lambda^2\mathfrak{s}_9)^{\mathfrak{s}_9}=0$ for every $\alpha>0$, the strata of  $\mathscr{E}_{\mathfrak{s}^\alpha_9}$ within $\mathcal{Y}_{\mathfrak{s}^\alpha_9}$ amount for the families of equivalent coboundary Lie bialgebras on $\mathfrak{s}^\alpha_9$. Our final results are summarised in Table \ref{Tab:g_orb_2}.

\subsection{Lie algebra $\mathfrak{s}_{10}$}

As in the previous cases, we hereafter use the structure constants for Lie algebra $\mathfrak{s}_{10}$ in Table \ref{Tab:StruCons} and the induced Schouten brackets between elements of $\mathfrak{s}_{10}$, $\Lambda^2 \mathfrak{s}_{10}$, and $\Lambda^3\mathfrak{s}_{10}$ depicted in  Tables  \ref{Tab:Schoten_1_2}--\ref{Tab:Schoten_1_3}. It is remarkable that $(\Lambda^3\mathfrak{s}_{10})^{\mathfrak{s}_{10}}=0$ and $(\Lambda^2\mathfrak{s}_{10})^{\mathfrak{s}_{10}}=0$.

By Remark \ref{Re:DerAlg}, the derivations of $\mathfrak{s}_{10}$ read
$$
\mathfrak{der}(\mathfrak{s}_{10})=\left\{\left(
\begin{array}{cccc}
\mu_{11} & \mu_{12} & \mu_{13} & \mu_{14} \\
0 & \frac{1}{2}\mu_{11} & \mu_{23} & \mu_{12} - \mu_{13} \\
0 & 0 & \frac{1}{2}\mu_{11} & \mu_{12} \\
0 & 0 & 0 & 0
\end{array}
\right): \mu_{11}, \mu_{12}, \mu_{13}, \mu_{14}, \mu_{23} \in \mathbb{R}\right\},
$$
which give rise to the basis of $V_{\mathfrak{s}_{10}}$ of the form
\begin{equation*}
\begin{gathered}
X_1 = \frac{3}{2} x_1 \partial_{x_1} + \frac{3}{2} x_2 \partial_{x_2} + x_3 \partial_{x_3} + x_4 \partial_{x_4} + \frac{1}{2} x_5 \partial_{x_5} + \frac{1}{2} x_6 \partial_{x_6}, \quad
X_2 = x_3 \partial_{x_1} + (x_3 + x_4) \partial_{x_2} + x_5 \partial_{x_3} + (x_5 - x_6) \partial_{x_4}, \\
X_3 = (-x_3-x_4) \partial_{x_1} + x_6 \partial_{x_3} + x_6 \partial_{x_4}, \quad
X_4 = -x_5 \partial_{x_1} - x_6 \partial_{x_2}, \quad
X_5 = x_2 \partial_{x_1} + x_6 \partial_{x_5}.
\end{gathered}
\end{equation*}

For an element $r \in \Lambda^2 \mathfrak{s}_{10}$, we get
$$
[r,r] = 2(3x_1 x_6 - 3x_2 x_5 + x_2 x_6 + 2x_3 x_4 + x_4^2)e_{123} + 2(-x_3 x_5 + x_3 x_6 + x_4 x_5)e_{124} + 2(x_4-x_3) x_6e_{134} -2x_6^2e_{234}.
$$

Since $(\Lambda^3 \mathfrak{s}_{10})^{\mathfrak{s}_{10}} = 0$, the mCYBE and the CYBE are equal and they read
$$
- 3x_2 x_5 + 2x_3 x_4 + x_4^2 = 0, \quad (x_4 - x_3) x_5 = 0, \quad x_6 = 0.
$$

Since $x_6$ is the only brick for $\mathfrak{s}_{10}$, our Darboux tree starts with the cases $x_6 = 0$ and $x_6 \neq 0$. The full Darboux tree is presented below.

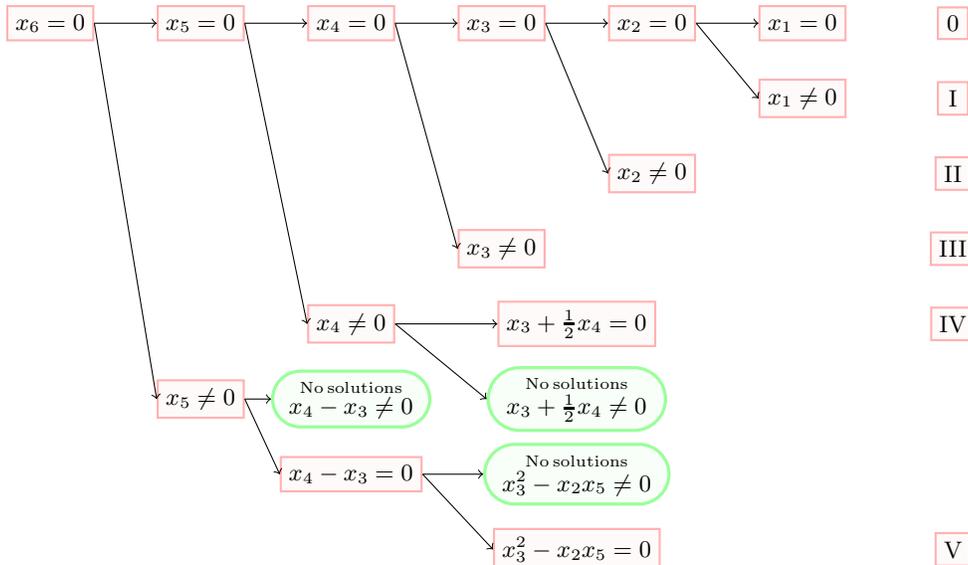
\begin{center}
{\small
\begin{tikzpicture}[
roundnode/.style={rounded rectangle, draw=green!40, fill=green!3, very thick, minimum size=2mm},
squarednode/.style={rectangle, draw=red!30, fill=red!2, thick, minimum size=4mm}
]
	\node[squarednode] (brick) at (0,0) {$x_6=0$};

	\node[squarednode] (u)  at (2,0) {$x_5=0$};
	\node[squarednode] (d)  at (2,-5) {$x_5 \neq 0$};

	\node[squarednode] (uu)  at (4,0) {$x_4=0$};
	\node[squarednode] (ud)  at (4,-4) {$x_4 \neq 0$};
	\node[roundnode] (du)  at (4,-5) {$\stackrel{{\rm No \, solutions}}{x_4 - x_3 \neq 0}$};
	\node[squarednode] (dd)  at (4,-6) {$x_4 -x_3 = 0$};

	\node[squarednode] (uuu)  at (6,0) {$x_3=0$};
	\node[squarednode] (uud)  at (6,-3) {$x_3 \neq 0$};
	\node[squarednode] (udu)  at (7,-4) {$x_3 + \frac{1}{2} x_4 = 0$};
	\node[roundnode] (udd)  at (7,-5) {$\stackrel{{\rm No \, solutions}}{x_3 + \frac{1}{2} x_4 \neq 0}$};
	\node[roundnode] (ddu)  at (7,-6) {$\stackrel{{\rm No \, solutions}}{x_3^2 - x_2 x_5  \neq 0}$};
	\node[squarednode] (ddd)  at (7,-7) {$x_3^2 - x_2 x_5 = 0$};
	
	\node[squarednode] (uuuu)  at (8,0) {$x_2=0$};
	\node[squarednode] (uuud)  at (8,-2) {$x_2 \neq 0$};

	\node[squarednode] (uuuuu)  at (10,0) {$x_1 = 0$};
	\node[squarednode] (uuuud)  at (10,-1) {$x_1 \neq 0$};
	
	\node[squarednode] (0) at (12,0) {0};
	\node[squarednode] (I) at (12,-1) {I};
	\node[squarednode] (II) at (12,-2) {II};
	\node[squarednode] (III) at (12,-3) {III};
	\node[squarednode] (IV) at (12,-4) {IV};
	\node[squarednode] (V) at (12,-7) {V};
	
	\draw[->] (brick.east) -- (u.west);
	\draw[->] (brick.east) -- (d.west);

	\draw[->] (u.east) -- (uu.west);
	\draw[->] (u.east) -- (ud.west);
	\draw[->] (d.east) -- (du.west);
	\draw[->] (d.east) -- (dd.west);

	\draw[->] (uu.east) -- (uuu.west);
	\draw[->] (uu.east) -- (uud.west);
	\draw[->] (ud.east) -- (udu.west);
	\draw[->] (ud.east) -- (udd.west);
	\draw[->] (dd.east) -- (ddu.west);
	\draw[->] (dd.east) -- (ddd.west);

	\draw[->] (uuu.east) -- (uuuu.west);
	\draw[->] (uuu.east) -- (uuud.west);
	
	\draw[->] (uuuu.east) -- (uuuuu.west);
	\draw[->] (uuuu.east) -- (uuuud.west);

\end{tikzpicture}
}
\end{center}
The connected parts of the subspaces denoted in the Darboux tree are the orbits of ${\rm Aut}_c(\mathfrak{s}_{10})$ within $\mathcal{Y}_{\mathfrak{s}_{10}}$. To obtain the orbits of ${\rm Aut}(\mathfrak{s}_{10})$ within $\mathcal{Y}_{\mathfrak{s}_{10}}$, we proceed as in previous sections and we derive the automorphism group of $\mathfrak{s}_{10}$, which takes the form
$$
{\rm Aut}(\mathfrak{s}_{10}) = \left\{
\left(
\begin{array}{cccc}
(T^2_2)^2 & T^4_3 T^2_2 & T^4_3 T^2_2 + T^4_3 T^3_2 - T^4_2 T^2_2 & T^4_1 \\
0 & T^2_2 & T^3_2 & T^4_2 \\
0 & 0 & T^2_2 & T^4_3 \\
0 & 0 & 0 & 1
\end{array}
\right):
\begin{array}{c}
 T^2_2 \in \rz, \\
T^3_2, T^4_1, T^4_2, T^4_3 \in \rr
\end{array}
\right\}.
$$
In reality, we do not need the full group. It is enough to note that there are two connected components of ${\rm Aut}(\mathfrak{s}_{10})$, represented by two automorphisms that read, along with their extensions to $\Lambda^2\mathfrak{s}_{10}$, as follows
\begin{equation}\label{aut_s10_con}
T_\lambda:=\left(
\begin{array}{cccc}
1 & 0 & 0 & 0 \\
0 & \lambda & 0 & 0 \\
0 & 0 & \lambda & 0 \\
0 & 0 & 0 & 1
\end{array}
\right)\qquad \Longrightarrow \qquad \Lambda^2T_\lambda:={\small
\left(
\begin{array}{cccccc}
\lambda  & 0 & 0 & 0 &0&0 \\
0 & \lambda  & 0 & 0 &0&0 \\
0 & 0 & 1 & 0 &0&0 \\
0 & 0 & 0 & 1&0&0 \\
0 & 0 & 0 & 0 &\lambda&0 \\
0 & 0 & 0 & 0&0&\lambda \\
\end{array}
\right)},\quad \lambda \in \{\pm 1\}.
\end{equation}
By applying $\Lambda^2T_\lambda$ to the orbits of ${\rm Aut}_c(\mathfrak{s}_{10})$ in $\mathcal{Y}_{\mathfrak{s}_{10}}$, depicted in  Table \ref{Tab:g_orb_2}, we obtain the orbits of Aut($\mathfrak{s}_{10})$ in $\mathcal{Y}_{\mathfrak{s}_{10}}$. Since $(\Lambda^2\mathfrak{s}_{10})^{\mathfrak{s}_{10}}=0$, such orbits amount to  the classification of equivalent coboundary Lie bialgebras up to Lie algebra automorphisms of $\mathfrak{s}_{10}$.

\subsection{Lie algebra $\mathfrak{s}_{11}$}

As previously, the information in Tables \ref{Tab:StruCons},  \ref{Tab:Schoten_1_2}--\ref{Tab:Schoten_1_3} allows us to accomplish the calculations of this section. In particular, one gets $(\Lambda^2\mathfrak{s}_{11})^{\mathfrak{s}_{11}}=0$ and $(\Lambda^3\mathfrak{s}_{11})^{\mathfrak{s}_{11}}=0$.

By Remark \ref{Re:DerAlg}, one obtains 
$$
\mathfrak{der}(\mathfrak{s}_{11})=\left\{\left(
\begin{array}{cccc}
\mu_{11} & \mu_{12} & \mu_{13} & \mu_{14} \\
0 & \mu_{22} & 0 & -\mu_{13} \\
0 & 0 & \mu_{11} - \mu_{22} & 0 \\
0 & 0 & 0 & 0
\end{array}
\right): \mu_{11}, \mu_{12}, \mu_{13}, \mu_{14}, \mu_{22} \in \mathbb{R}\right\}.
$$
Such derivations give rise to the basis of $V_{\mathfrak{s}_{11}}$ of the form
\begin{equation*}
\begin{gathered}
X_1 = x_1 \partial_{x_1} + 2x_2 \partial_{x_2} + x_3 \partial_{x_3} + x_4 \partial_{x_4}+ x_6 \partial_{x_6}, \quad
X_2 = x_4 \partial_{x_2} + x_5 \partial_{x_3}, \quad
X_3 = (-x_3-x_4) \partial_{x_1} + x_6 \partial_{x_3} + x_6 \partial_{x_4}, \\
X_4 = -x_5 \partial_{x_1} - x_6 \partial_{x_2}, \quad
X_5 = x_1 \partial_{x_1} - x_2 \partial_{x_2} + x_5 \partial_{x_5} - x_6 \partial_{x_6},
\end{gathered}
\end{equation*}

For an element $r \in \Lambda^2 \mathfrak{s}_{11}$, we get
$$
[r,r] = 2(2x_1 x_6 - x_2 x_5 + x_3 x_4 + x_4^2)e_{123} + 2x_4 x_5e_{124} + 2(x_4-x_3) x_6 e_{134} -2x_6x_5e_{234}.
$$
Since $(\Lambda^3 \mathfrak{s}_{11})^{\mathfrak{s}_{11}} = \{0\}$, the mCYBE and the CYBE are equal and read
$$
2x_1 x_6 - x_2 x_5 + x_3 x_4 + x_4^2 = 0, \quad x_4 x_5 = 0, \quad (x_4 - x_3) x_6 = 0, \quad x_5 x_6 = 0.
$$

Since $x_5, x_6$ are the bricks for $\mathfrak{s}_{11}$, our Darboux tree starts with the cases $x_i = 0$ and $x_i \neq 0$, for $i \in \{5,6\}$. The Darboux tree is presented below.

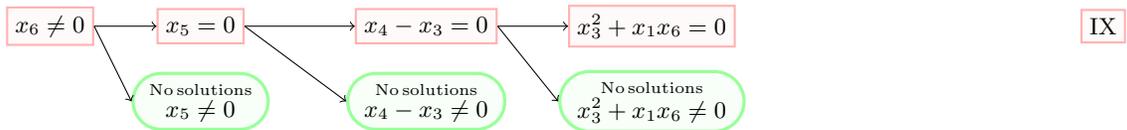
\begin{center}
{\small 
\begin{tikzpicture}[
roundnode/.style={rounded rectangle, draw=green!40, fill=green!3, very thick, minimum size=2mm},
squarednode/.style={rectangle, draw=red!30, fill=red!2, thick, minimum size=4mm}
]
	\node[squarednode] (brick) at (0,0) {$x_6 \neq 0$};

	\node[squarednode] (u)  at (2,0) {$x_5=0$};
	\node[roundnode] (d)  at (2,-1) {$\stackrel{{\rm No \, solutions}}{x_5 \neq 0}$};

	\node[squarednode] (uu)  at (5,0) {$x_4 - x_3 = 0$};
	\node[roundnode] (ud)  at (5,-1) {$\stackrel{{\rm No \, solutions}}{x_4 - x_3 \neq 0}$};

	\node[squarednode] (uuu)  at (8,0) {$x_3^2 + x_1 x_6 = 0$};
	\node[roundnode] (uud)  at (8,-1) {$\stackrel{{\rm No \, solutions}}{x_3^2 + x_1 x_6 \neq 0}$};
	
	\node[squarednode] (IX) at (14,0) {IX};
	
	\draw[->] (brick.east) -- (u.west);
	\draw[->] (brick.east) -- (d.west);

	\draw[->] (u.east) -- (uu.west);
	\draw[->] (u.east) -- (ud.west);

	\draw[->] (uu.east) -- (uuu.west);
	\draw[->] (uu.east) -- (uud.west);

\end{tikzpicture}}
\end{center}
\begin{center}
{\small
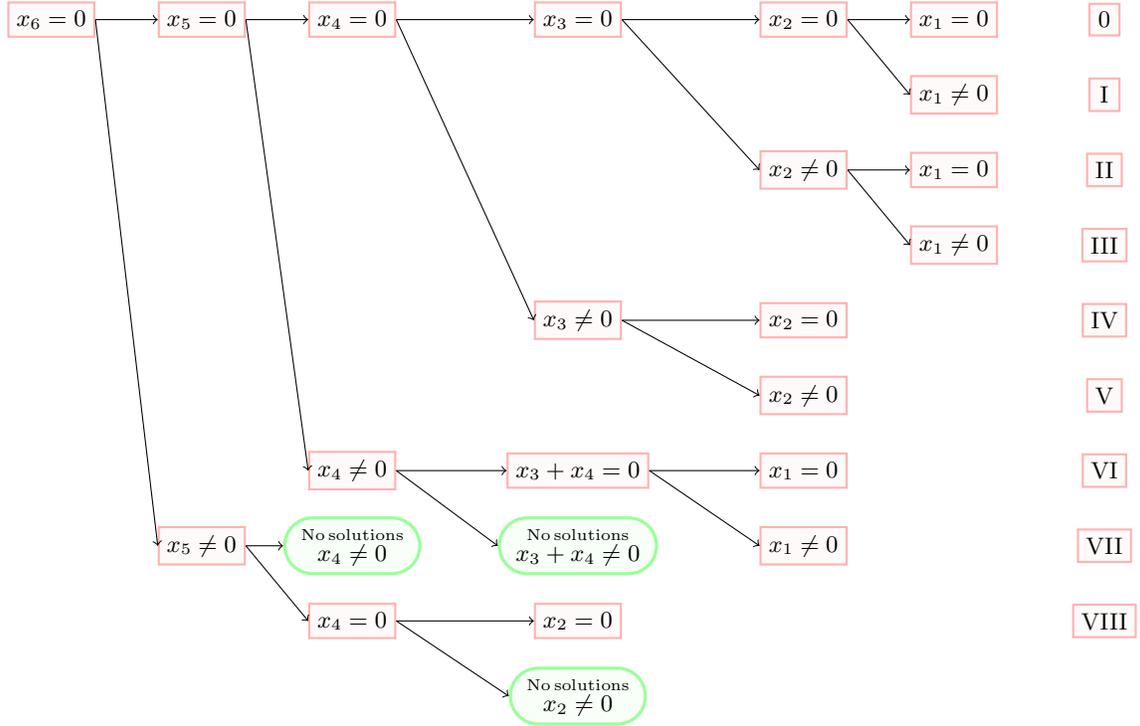
\begin{tikzpicture}[
roundnode/.style={rounded rectangle, draw=green!40, fill=green!3, very thick, minimum size=2mm},
squarednode/.style={rectangle, draw=red!30, fill=red!2, thick, minimum size=4mm}
]
	\node[squarednode] (brick) at (0,0) {$x_6=0$};

	\node[squarednode] (u)  at (2,0) {$x_5=0$};
	\node[squarednode] (d)  at (2,-7) {$x_5 \neq 0$};

	\node[squarednode] (uu)  at (4,0) {$x_4=0$};
	\node[squarednode] (ud)  at (4,-6) {$x_4 \neq 0$};
	\node[roundnode] (du)  at (4,-7) {$\stackrel{{\rm No \, solutions}}{x_4 \neq 0}$};
	\node[squarednode] (dd)  at (4,-8) {$x_4 = 0$};

	\node[squarednode] (uuu)  at (7,0) {$x_3=0$};
	\node[squarednode] (uud)  at (7,-4) {$x_3 \neq 0$};
	\node[squarednode] (udu)  at (7,-6) {$x_3 + x_4 = 0$};
	\node[roundnode] (udd)  at (7,-7) {$\stackrel{{\rm No \, solutions}}{x_3 + x_4 \neq 0}$};
	\node[squarednode] (ddu)  at (7,-8) {$x_2 = 0$};
	\node[roundnode] (ddd)  at (7,-9) {$\stackrel{{\rm No \, solutions}}{x_2 \neq 0}$};

	\node[squarednode] (uuuu)  at (10,0) {$x_2=0$};
	\node[squarednode] (uuud)  at (10,-2) {$x_2 \neq 0$};
	\node[squarednode] (uudu)  at (10,-4) {$x_2 = 0$};
	\node[squarednode] (uudd)  at (10,-5) {$x_2 \neq 0$};
	\node[squarednode] (uduu)  at (10,-6) {$x_1 = 0$};
	\node[squarednode] (udud)  at (10,-7) {$x_1 \neq 0$};

	\node[squarednode] (uuuuu)  at (12,0) {$x_1=0$};
	\node[squarednode] (uuuud)  at (12,-1) {$x_1 \neq 0$};
	\node[squarednode] (uuudu)  at (12,-2) {$x_1 = 0$};
	\node[squarednode] (uuudd)  at (12,-3) {$x_1 \neq 0$};
	
	\node[squarednode] (0) at (14,0) {0};
	\node[squarednode] (I) at (14,-1) {I};
	\node[squarednode] (II) at (14,-2) {II};
	\node[squarednode] (III) at (14,-3) {III};
	\node[squarednode] (IV) at (14,-4) {IV};
	\node[squarednode] (V) at (14,-5) {V};
	\node[squarednode] (VI) at (14,-6) {VI};
	\node[squarednode] (VII) at (14,-7) {VII};
	\node[squarednode] (VIII) at (14,-8) {VIII};
	
	\draw[->] (brick.east) -- (u.west);
	\draw[->] (brick.east) -- (d.west);

	\draw[->] (u.east) -- (uu.west);
	\draw[->] (u.east) -- (ud.west);
	\draw[->] (d.east) -- (du.west);
	\draw[->] (d.east) -- (dd.west);

	\draw[->] (uu.east) -- (uuu.west);
	\draw[->] (uu.east) -- (uud.west);
	\draw[->] (ud.east) -- (udu.west);
	\draw[->] (ud.east) -- (udd.west);
	\draw[->] (dd.east) -- (ddu.west);
	\draw[->] (dd.east) -- (ddd.west);

	\draw[->] (uuu.east) -- (uuuu.west);
	\draw[->] (uuu.east) -- (uuud.west);
	\draw[->] (uud.east) -- (uudu.west);
	\draw[->] (uud.east) -- (uudd.west);
	\draw[->] (udu.east) -- (uduu.west);
	\draw[->] (udu.east) -- (udud.west);

	\draw[->] (uuuu.east) -- (uuuuu.west);
	\draw[->] (uuuu.east) -- (uuuud.west);
	\draw[->] (uuud.east) -- (uuudu.west);
	\draw[->] (uuud.east) -- (uuudd.west);

\end{tikzpicture}

\vspace{0.5cm}

}
\end{center}

The above Darboux tree gives the orbits of ${\rm Aut}_c(\mathfrak{s}_{11})$ in $\mathcal{Y}_{\mathfrak{s}_{11}}$. To obtain the equivalence classes of $r$-matrices up to Lie algebra automorphisms of ${\rm Aut}(\mathfrak{s}_{11})$, we obtain that
$$
{\rm Aut}(\mathfrak{s}_{11}) = \left\{
\left(
\begin{array}{cccc}
T^2_2 T^3_3 & T^2_1 & -T^3_3 T^4_2 & T^4_1 \\
0 & T^2_2 & 0 & T^4_2 \\
0 & 0 & T^3_3 & 0 \\
0 & 0 & 0 & 1
\end{array}
\right): 
\begin{array}{c}
T^2_2, T^3_3 \in \rz \\
T^2_1, T^4_1, T^4_2 \in \rr
\end{array}
\right\}.
$$
In fact, only one representative for each connected component is needed for our purposes. There are four connected components of ${\rm Aut}(\mathfrak{s}_{11})$, each one represented by one of the following automorphisms, which are given together with their lifts to $\Lambda^2\mathfrak{s}_{11}$, namely
\begin{equation}\label{aut_s11_con}
T_{\lambda_1,\lambda_2}:=\left(
\begin{array}{cccc}
\lambda_1 \lambda_2 & 0 & 0 & 0 \\
0 & \lambda_1 & 0 & 0 \\
0 & 0 & \lambda_2 & 0 \\
0 & 0 & 0 & 1
\end{array}
\right) \quad \Longrightarrow \quad \Lambda^2T_{\lambda_1,\lambda_2}= {\small
\left(
\begin{array}{cccccc}
  \lambda_2 & 0 & 0 & 0 &0&0 \\
0 & \lambda_1   & 0 & 0 &0&0 \\
0 & 0 & \lambda_1 \lambda_2 & 0 &0&0 \\
0 & 0 & 0 & \lambda_1 \lambda_2&0&0 \\
0 & 0 & 0 & 0 &\lambda_1&0 \\
0 & 0 & 0 & 0&0&\lambda_2 \\
\end{array}
\right),\qquad \lambda_1, \lambda_2 \in \{ \pm 1\}}.
\end{equation}

The $\Lambda^2T_{\lambda_1,\lambda_2}$ allow us to check whether the orbits of ${\rm Aut}_c(\mathfrak{s}_{11})$ are connected by an element of ${\rm Aut}(\mathfrak{s}_{11})$ acting on $\Lambda^2\mathfrak{s}_{11}$. The equivalence classes of such $r$-matrices (up to the action of elements of ${\rm Aut}(\mathfrak{s}_{11})$, as standard) are given in Table \ref{Tab:g_orb_2}. Since $(\Lambda^2\mathfrak{s}_{11})^{\mathfrak{s}_{11}}=0$, each family of $r$-matrices gives rise to a separate class of equivalent coboundary Lie bialgebras on $\mathfrak{s}_{11}$.

\subsection{Lie algebra $\mathfrak{s}_{12}$}

The structure constants for Lie algebra $\mathfrak{s}_{12}$ are given in Table \ref{Tab:StruCons}. This along with the selected Schouten brackets between basis elements of $\mathfrak{s}_{12}$, $\Lambda^2 \mathfrak{s}_{12}$, and $\Lambda^3\mathfrak{s}_{12}$ described in Tables \ref{Tab:Schoten_1_2}--\ref{Tab:Schoten_1_3}, allow us to accomplish the calculations of this section. In particular, we obtain that  $(\Lambda^2\mathfrak{s}_{12})^{\mathfrak{s}_{12}}=0$ and $(\Lambda^3\mathfrak{s}_{12})^{\mathfrak{s}_{12}}=0$.

By Remark \ref{Re:DerAlg}, one obtains that
$$
\mathfrak{der}(\mathfrak{s}_{12})=\left\{\left(
\begin{array}{cccc}
\mu_{11} & \mu_{12} & \mu_{13} & \mu_{14} \\
-\mu_{12} & \mu_{11} & \mu_{14} & -\mu_{13} \\
0 & 0 & 0 & 0 \\
0 & 0 & 0 & 0
\end{array}
\right): \mu_{11}, \mu_{12}, \mu_{13}, \mu_{14} \in \mathbb{R}\right\}.
$$
The previous derivations give rise to the basis of $V_{\mathfrak{s}_{12}}$ given by
\begin{equation*}
\begin{gathered}
X_1 = 2x_1 \partial_{x_1} + x_2 \partial_{x_2} + x_3 \partial_{x_3} + x_4 \partial_{x_4}+ x_5 \partial_{x_5}, \quad
X_2 = x_4 \partial_{x_2} + x_5 \partial_{x_3} -  x_2 \partial_{x_4} -  x_3 \partial_{x_5}, \\
X_3 = (-x_3-x_4) \partial_{x_1} + x_6 \partial_{x_3} + x_6 \partial_{x_4}, \quad
X_4 = (x_2 - x_5) \partial_{x_1} - x_6 \partial_{x_2} + x_6 \partial_{x_5}.
\end{gathered}
\end{equation*}

For an element $r \in \Lambda^2 \mathfrak{s}_{12}$, we get
$$
[r,r] = -2(x_2 x_3 + x_4 x_5)e_{123} + 2(-2x_1 x_6 + x_2 x_5 -x_3^2 - x_3 x_4 - x_5^2)e_{124} - 2(x_2 + x_5)x_6e_{134} + 2(x_3 - x_4)x_6e_{234}.
$$
Since $(\Lambda^3 \mathfrak{s}_{12})^{\mathfrak{s}_{12}} = 0$, the mCYBE and CYBE are equal and read
$$
x_2 x_3 + x_4 x_5 = 0, \quad 2x_1 x_6 - x_2 x_5 + x_3^2 + x_3 x_4 + x_5^2 = 0, \quad (x_2 + x_5) x_6 = 0, \quad (x_3 - x_4) x_6 = 0.
$$

Since $x_6$ is the only brick for $\mathfrak{s}_{12}$, our Darboux tree starts with the cases $x_6 = 0$ and $x_6 \neq 0$. The full Darboux tree is given next.

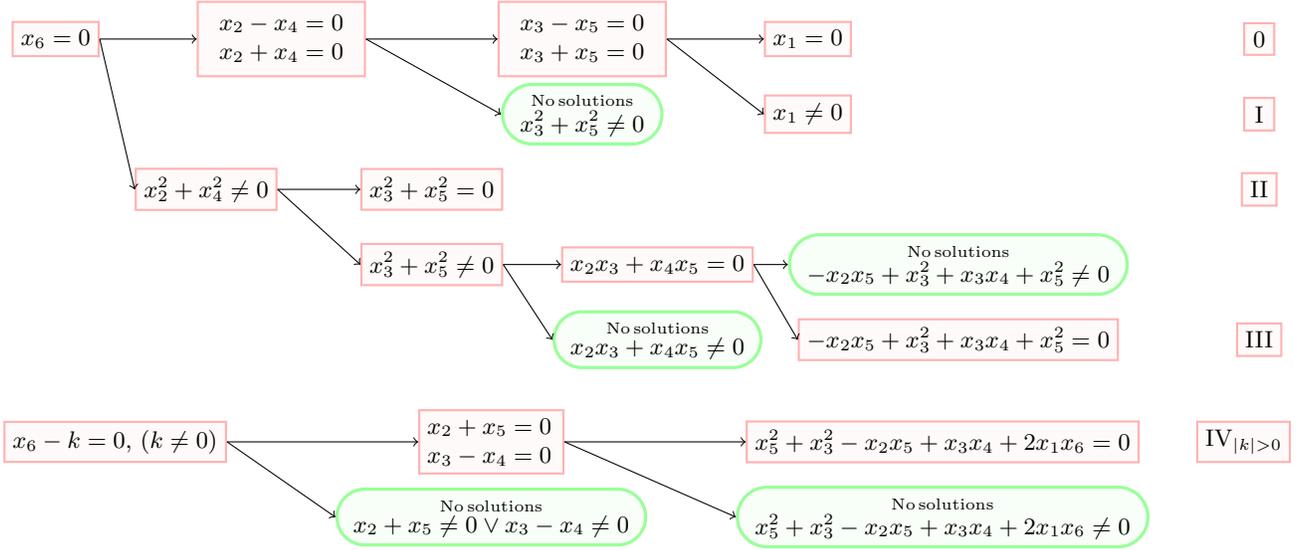
\begin{center}
{\small
\begin{tikzpicture}[
roundnode/.style={rounded rectangle, draw=green!40, fill=green!3, very thick, minimum size=2mm},
squarednode/.style={rectangle, draw=red!30, fill=red!2, thick, minimum size=4mm}
]
	\node[squarednode] (brick) at (0,0) {$x_6=0$};

	\node[squarednode] (u)  at (3,0) {$\begin{array}{c}
 x_2-x_4=0\\
x_2+x_4=0 
\end{array} $};
	\node[squarednode] (d)  at (2,-2) {$x_2^2 + x_4^2 \neq 0$};

	\node[squarednode] (uu)  at (7,0) {$\begin{array}{c}
 x_3-x_5=0\\
x_3+x_5=0 
\end{array} $};
	\node[roundnode] (ud)  at (7,-1) {$\stackrel{{\rm No \, solutions}}{x_3^2 + x_5^2 \neq 0}$};
	\node[squarednode] (du)  at (5,-2) {$x_3^2 + x_5^2 = 0$};
	\node[squarednode] (dd)  at (5,-3) {$x_3^2 + x_5^2 \neq 0$};

	\node[squarednode] (uuu)  at (10,0) {$x_1=0$};
	\node[squarednode] (uud)  at (10,-1) {$x_1 \neq 0$};
	\node[squarednode] (ddu)  at (8,-3) {$x_2 x_3 + x_4 x_5 = 0$};
	\node[roundnode] (ddd)  at (8,-4) {$\stackrel{{\rm No \, solutions}}{x_2 x_3 + x_4 x_5 \neq 0}$};

	\node[roundnode] (dddu)  at (12,-3) {$\stackrel{{\rm No \, solutions}}{-x_2 x_5 + x_3^2 + x_3 x_4 + x_5^2 \neq 0}$};
	\node[squarednode] (dddd)  at (12,-4) {$-x_2 x_5 + x_3^2 + x_3 x_4 + x_5^2 = 0$};
	
	\node[squarednode] (0) at (16,0) {0};
	\node[squarednode] (I) at (16,-1) {I};
	\node[squarednode] (II) at (16,-2) {II};
	\node[squarednode] (III) at (16,-4) {III};
	
	\draw[->] (brick.east) -- (u.west);
	\draw[->] (brick.east) -- (d.west);

	\draw[->] (u.east) -- (uu.west);
	\draw[->] (u.east) -- (ud.west);
	\draw[->] (d.east) -- (du.west);
	\draw[->] (d.east) -- (dd.west);

	\draw[->] (uu.east) -- (uuu.west);
	\draw[->] (uu.east) -- (uud.west);
	\draw[->] (dd.east) -- (ddu.west);
	\draw[->] (dd.east) -- (ddd.west);

	\draw[->] (ddu.east) -- (dddu.west);
	\draw[->] (ddu.east) -- (dddd.west);

\end{tikzpicture}

\vspace{0.5cm}

\begin{tikzpicture}[
roundnode/.style={rounded rectangle, draw=green!40, fill=green!3, very thick, minimum size=2mm},
squarednode/.style={rectangle, draw=red!30, fill=red!2, thick, minimum size=4mm}
]
	\node[squarednode] (brick) at (0,0) {$x_6-k=0$, ($k\neq 0$)};

	\node[squarednode] (u)  at (5,0) {\parbox{1.7cm}{$x_2 + x_5=0 \\ x_3 - x_4 = 0$}};
	\node[roundnode] (d)  at (5,-1) {$\stackrel{{\rm No \, solutions}}{x_2 + x_5 \neq 0 \lor x_3 - x_4 \neq 0}$};

	\node[squarednode] (uu)  at (11,0) {$x_5^2 + x_3^2 - x_2 x_5 + x_3 x_4 + 2x_1 x_6 = 0$};
	\node[roundnode] (ud)  at (11,-1) {$\stackrel{{\rm No \, solutions}}{x_5^2 + x_3^2 - x_2 x_5 + x_3 x_4 + 2x_1 x_6 \neq 0}$};
	
	\node[squarednode] (IV) at (15,0) {IV$_{|k|>0}$};
	
	\draw[->] (brick.east) -- (u.west);
	\draw[->] (brick.east) -- (d.west);

	\draw[->] (u.east) -- (uu.west);
	\draw[->] (u.east) -- (ud.west);

\end{tikzpicture}
}
\end{center}
The connected parts of the loci of the Darboux families of the Darboux tree are the orbits of ${\rm Aut}_c(\mathfrak{s}_{12})$ within $\mathcal{Y}_{\mathfrak{s}_{12}}$. These are given by the connected parts of the subsets given in Table \ref{Tab:g_orb_2}. Let us obtain the orbits of ${\rm Aut}(\mathfrak{s}_{12})$ on $\mathcal{Y}_{\mathfrak{sl}_{12}}$.

The Lie algebra automorphism group of $\mathfrak{s}_{12}$ reads
\begin{equation*}
{\rm Aut}(\mathfrak{s}_{12}) = \left\{
\left(
\begingroup
\setlength\arraycolsep{4pt}
\begin{array}{cccc}
T^1_1 & T^2_1 & T^3_1 & \pm T^3_2 \\
\mp T^2_1 & \pm T^1_1 & T^3_2 & \mp T^3_1 \\
0 & 0 & 1 & 0 \\
0 & 0 & 0 & \pm 1 
\end{array}
\endgroup
\right):
\begingroup
\setlength\arraycolsep{4pt}
\begin{array}{c}
 (T_1^1)^2 + (T^2_1)^2 \in \rz \\
T^1_1,T^2_1,T^3_1, T^3_2 \in \rr
\end{array}
\endgroup
\right\}.
\end{equation*}
There are two connected components of ${\rm Aut}(\mathfrak{s}_{12})$. One such element for each connected component of ${\rm Aut}(\mathfrak{s}_{12})$ and their extensions to $\Lambda^2\mathfrak{s}_{12}$ are given by
$$
T_\pm:=\left(
\begin{array}{cccc}
1 & 0 & 0 & 0 \\
0 & \pm 1 & 0 & 0 \\
0 & 0 & 1 & 0 \\
0 & 0 & 0 & \pm 1 
\end{array}
\right)\quad \Longrightarrow \quad \Lambda^2T_\pm =\left(
\begin{array}{cccccc}
\pm 1 & 0 & 0 & 0&0&0 \\
0 & 1 & 0 & 0 &0&0\\
0 & 0 & \pm 1 & 0 &0&0\\
0 & 0 & 0 & \pm 1&0&0\\
0 & 0 & 0 & 0 &1 &0\\
0 & 0 & 0 & 0 &0&\pm 1\\
\end{array}
\right).
$$
 By using $\Lambda^2T_{\pm}$ on the loci of the above Darboux tree, we obtain the orbits of ${\rm Aut}(\mathfrak{s}_{12})$ that are given by the classes of equivalent $r$-matrices detailed in Table \ref{Tab:g_orb_2}.  Since $(\Lambda^2\mathfrak{s}_{12})^{\mathfrak{s}_{12}}=0$, each family of equivalent $r$-matrices gives rise to a separate class of equivalent coboundary Lie bialgebras.

\subsection{Lie algebra $\mathfrak{n}_1$}
The structure constants for the Lie algebra $\mathfrak{n}_1$ are given in Table \ref{Tab:StruCons}, while relevant Schouten brackets between the elements of bases of $\mathfrak{n}_1$, $\Lambda^2\mathfrak{n}_1$, and $\Lambda^3\mathfrak{n}_1$ to be used hereafter are displayed in Tables  \ref{Tab:Schoten_1_2}--\ref{Tab:Schoten_1_3}. From these calculations, one obtains that  $(\Lambda^2\mathfrak{n}_1)^{\mathfrak{n}_1}= \langle e_{12}\rangle$ and $(\Lambda^3\mathfrak{n}_1)^{\mathfrak{n}_1}=\langle e_{123},e_{124}\rangle$.
By Remark \ref{Re:DerAlg}, one has that the derivations\footnote{Note that the dimension of the space of derivation matches the result given in the W\v{S} classification} of $\mathfrak{n}_1$ take the form
$$
\mathfrak{der}(\mathfrak{n}_1)=\left\{\left(
\begin{array}{cccc}
\mu_{11} & \mu_{12} & \mu_{13} & \mu_{14} \\
0 & \mu_{22} & \mu_{12} & \mu_{24} \\
0 & 0 & 2\mu_{22} - \mu_{11} & \mu_{34} \\
0 & 0 & 0 & \mu_{11} - \mu_{22}
\end{array}
\right): \mu_{11}, \mu_{12}, \mu_{13}, \mu_{14}, \mu_{22}, \mu_{24}, \mu_{34} \in \mathbb{R}\right\},
$$
which, by lifting them to $\Lambda^2\mathfrak{n}_1$, give rise to the basis of $V_{\mathfrak{n}_1}$ of the form
\begin{equation*}
\begin{gathered}
X_1 = x_1 \partial_{x_1} + 2x_3 \partial_{x_3} - x_4 \partial_{x_4} + x_5 \partial_{x_5}, \quad
X_2 = x_2 \partial_{x_1} + x_4 \partial_{x_2} + x_5 \partial_{x_3} + x_6 \partial_{x_5}, \quad
X_3 = -x_4 \partial_{x_1} + x_6 \partial_{x_3}, \\
X_4 = -x_5 \partial_{x_1} - x_6 \partial_{x_2}, \quad
X_5 = x_1 \partial_{x_1} + 2x_2 \partial_{x_2} - x_3 \partial_{x_3} + 3x_4 \partial_{x_4} + x_6 \partial_{x_6}, \\
X_6 = x_3 \partial_{x_1} - x_6 \partial_{x_4}, \quad
X_7 = x_3 \partial_{x_2} + x_5 \partial_{x_4}.
\end{gathered}
\end{equation*}
For an element $r \in \Lambda^2 \mathfrak{n}_1$, we get
$$
[r, r] = 2(x_4 x_5 - x_2 x_6) e_{123} + 2(x_5^2 - x_3 x_6) e_{124} + 2 x_5 x_6 e_{134} + 2 x_6^2 e_{234}.
$$
Thus, CYBE reads
$$
x_5=0,\quad x_6 = 0.
$$
Since $(\Lambda^3 \mathfrak{n}_1)^{\mathfrak{n}_1} = \langle e_{123}, e_{124}\rangle$, the mCYBE differs from the CYBE and reads
$$
x_6 = 0.
$$
The Darboux tree for this Lie algebra reads
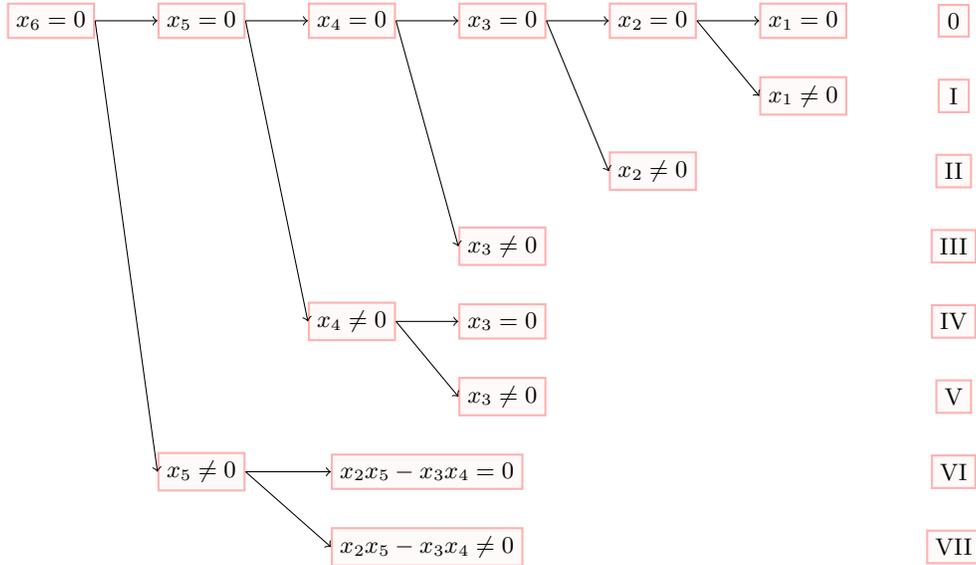
\begin{center}
{\small
\begin{tikzpicture}[
roundnode/.style={rounded rectangle, draw=green!40, fill=green!3, very thick, minimum size=2mm},
squarednode/.style={rectangle, draw=red!30, fill=red!2, thick, minimum size=4mm}
]
\node[squarednode] (brick) at (0,0) {$x_6=0$};

\node[squarednode] (u)  at (2,0) {$x_5=0$};
\node[squarednode] (d)  at (2,-6) {$x_5 \neq 0$};

\node[squarednode] (uu)  at (4,0) {$x_4 = 0$};
\node[squarednode] (ud)  at (4,-4) {$x_4 \neq 0$};
\node[squarednode] (du)  at (5,-6) {$x_2 x_5 - x_3 x_4 = 0$};
\node[squarednode] (dd)  at (5,-7) {$x_2 x_5 - x_3 x_4 \neq 0$};

\node[squarednode] (uuu)  at (6,0) {$x_3=0$};
\node[squarednode] (uud)  at (6,-3) {$x_3 \neq 0$};
\node[squarednode] (udu)  at (6,-4) {$x_3 = 0$};
\node[squarednode] (udd)  at (6,-5) {$x_3 \neq 0$};

\node[squarednode] (uuuu)  at (8,0) {$x_2=0$};
\node[squarednode] (uuud)  at (8,-2) {$x_2 \neq 0$};

\node[squarednode] (uuuuu)  at (10,0) {$x_1 = 0$};
\node[squarednode] (uuuud)  at (10,-1) {$x_1 \neq 0$};

\node[squarednode] (0) at (12,0) {0};
\node[squarednode] (I) at (12,-1) {I};
\node[squarednode] (II) at (12,-2) {II};
\node[squarednode] (III) at (12,-3) {III};
\node[squarednode] (IV) at (12,-4) {IV};
\node[squarednode] (V) at (12,-5) {V};
\node[squarednode] (VI) at (12,-6) {VI};
\node[squarednode] (VII) at (12,-7) {VII};

\draw[->] (brick.east) -- (u.west);
\draw[->] (brick.east) -- (d.west);

\draw[->] (u.east) -- (uu.west);
\draw[->] (u.east) -- (ud.west);
\draw[->] (d.east) -- (du.west);
\draw[->] (d.east) -- (dd.west);

\draw[->] (uu.east) -- (uuu.west);
\draw[->] (uu.east) -- (uud.west);
\draw[->] (ud.east) -- (udu.west);
\draw[->] (ud.east) -- (udd.west);

\draw[->] (uuu.east) -- (uuuu.west);
\draw[->] (uuu.east) -- (uuud.west);

\draw[->] (uuuu.east) -- (uuuuu.west);
\draw[->] (uuuu.east) -- (uuuud.west);

\end{tikzpicture}
}
\end{center}
As in all previous sections, the orbits of {\rm Aut}$_c(\mathfrak{n}_1$) in $\mathcal{Y}_{\mathfrak{n}_1}$ are given by the connected components of the loci of the Darboux families of the above Darboux tree. Results can be found in Table \ref{Tab:g_orb_2}.

Let us again obtain the orbits of {\rm Aut}$(\mathfrak{n}_1)$ in $\mathcal{Y}_{\mathfrak{n}_1}$. The automorphism group of $\mathfrak{n}_1$ takes the form
$$
{\rm Aut}(\mathfrak{n}_1) = \left\{\left(
\begin{array}{cccc}
T^3_3 (T^4_4)^2 & T^3_2 T^4_4 & T^3_1 & T^4_1 \\
0 & T^3_3 T^4_4 & T^3_2 & T^4_2 \\
0 & 0 & T^3_3 & T^4_3 \\
0 & 0 & 0 & T^4_4
\end{array}
\right): \quad T^3_3 , T^4_4 \in \rz, T^3_1, T^3_2, T^4_1, T^4_2, T^4_3 \in \rr\right\}.
$$
It has four connected components represented by the following Lie algebra automorphisms, which are also lifted to $\Lambda^2\mathfrak{n}_1$:
$$
T_{\lambda_1,\lambda_2}:=\left(
\begin{array}{cccc}
\lambda_1 & 0 & 0 & 0 \\
0 & \lambda_1\lambda_2 & 0 & 0 \\
0 & 0 & \lambda_1 & 0 \\
0 & 0 & 0 & \lambda_2
\end{array}
\right)\quad \Longrightarrow \quad \Lambda^2T_{\lambda_1,\lambda_2}=\left(
\begin{array}{cccccc}
\lambda_2 & 0 & 0 & 0&0&0 \\
0 & 1 & 0 & 0&0&0 \\
0 & 0 & \lambda_1\lambda_2 & 0&0&0 \\
0 & 0 & 0 & \lambda_2&0&0\\
0 & 0 & 0 & 0&\lambda_1&0\\
0 & 0 & 0 & 0&0&\lambda_1\lambda_2\\
\end{array}
\right),\quad \lambda_1,\lambda_2\in \{\pm 1\}.
$$
The action of the $\Lambda^2T_{\lambda_1,\lambda_2}$ on the loci of the above Darboux tree gives the families of equivalent $r$-matrices given in Table \ref{Tab:g_orb_2}. Since $(\Lambda^2\mathfrak{n}_1)^{\mathfrak{n}_1}=\langle e_{12}\rangle$ and using the information in Table \ref{Tab:g_orb_2}, we see that the families of equivalent coboundary Lie bialgebras on $\mathfrak{n}_1$ are given by orbits of ${\rm Aut}(\mathfrak{n}_1)$ with the exception that the class I and the zero class (given by the zero $r$-matrix) give the Lie bialgebra in $\mathfrak{n}_1$ with a zero coproduct. 

\section{Conclusions and outlook}

This work has devised a generalisation of  the theory of Darboux polynomials to determine and to classify up to Lie algebra automorphisms the coboundary real Lie bialgebras over indecomposable real four-dimensional Lie algebras in a geometric manner. As a byproduct, a technique for the matrix representation of Lie algebras with non-trivial kernel has been developed. Such matrix representations are frequently useful in calculations.

The procedures devised in this work are good enough to make affordable the classification, up to Lie algebra automorphisms, of coboundary Lie bialgebras on real and complex, of at least dimension five, indecomposable Lie algebras via the \v{S}nobl and Winternitz classification \cite{SW14}.  The case of two- and three-dimensional Lie bialgebras can also be obtained. Moreover, a brief analysis of the classification of coboundary coproducts on real four-dimensional decomposable Lie algebras through our methods shows that  their classification relies partially on the classification of coboundary coproducts on three- and two-dimensional Lie algebras, while the complexity of the procedure is significantly  easier than in the present work \cite{Go00,LW20}. We expect to tackle this task in the future and to compare our results with previous works on the topic \cite{BHP99,LW20,Ku94}. Our techniques are also expected to be applied to other higher-dimensional Lie algebras of particular types, e.g. five-dimensional nilpotent or semi-simple Lie algebras. In all previously mentioned cases, we aim to inspect new properties of the Darboux families used to study these problems.

The determination of coboundary Lie bialgebras is relevant to the study of Poisson Lie groups. In particular, it represents an initial step in the study of Poisson Lie groups of dimension four, which would give rise to an extension of the paper \cite{BBM12}.

Our work has focused on the classification of coboundary Lie bialgebras on indecomposable Lie algebras. Probably, one laborious task in the application of our method is the determination of a representative of each connected part of the Lie group of Lie algebra automorphisms of a Lie algebra. In the case of semi-simple Lie algebras, this is much easier, as algebraic techniques based on Dynkin diagrams and other results can be applied \cite{Bo05,Mu52}. It is left for further works to study this problem and to search for coboundary Lie bialgebras in the Lie algebra $\mathfrak{so}(3,2)$, which has applications in the  analysis of quantum gravity in 3+1 dimensions (see \cite{BHM14} and references therein).

\section{Acknowledgements}
 D. Wysocki acknowledges support from a grant financed by the University of Warsaw (UW) and the Kartezjusz program of the Jagiellonian University and UW. J. de Lucas acknowledges partial financial support from the NCN grant  HARMONIA 2016/22/M/ST1/00542.

{ 
	\begin{table}[h]
\centering
\resizebox{\columnwidth}{!}{
	\begin{tabular}{|c|c|c|c|c|c|c|c|c|c|}
\hline
\theadstart
$\mathfrak{g}$ & \textrm{Orbit} & {\rm Dim} & $x_1$ & $x_2$ & $x_3$ & $x_4$ & $x_5$ & $x_6$ & {\rm Repr. element } \\
\hhline{|=|=|=|=|=|=|=|=|=|=|}
\multirow{8}{*}{$\mathfrak{s}_1$} & ${\rm I}_\pm$ & 1  & $\rpm$ & 0 & 0 & 0 & 0 & 0 & $\pm e_{12}$ \\\hhline{|~|-|-|-|-|-|-|-|-|-|}
& {\rm II} & 1  & 0 & $\rz$ & 0 & 0 & 0 & 0 & $e_{13}$ \\\hhline{|~|-|-|-|-|-|-|-|-|-|}
& ${\rm III}_\pm$ & 2  & $\rpm$ & $\rz$ & 0 & 0 & 0 & 0 & $\pm e_{12} + e_{13}$ \\\hhline{|~|-|-|-|-|-|-|-|-|-|}
& {\rm IV} & 2  & 0 & $\rr$ & 0 & $\rz$ & 0 & 0 & $e_{23}$ \\\hhline{|~|-|-|-|-|-|-|-|-|-|}
& ${\rm V}_\pm$ & 3  & $\rpm$ & $\rr$ & 0 & $\rz$ & 0 & 0 & $\pm e_{12} + e_{23}$ \\\hhline{|~|-|-|-|-|-|-|-|-|-|}
& {\rm VI} & 3  & $\rr$ & $\rr$ & $\rz$ & 0 & 0 & 0 & $e_{14}$ \\\hhline{|~|-|-|-|-|-|-|-|-|-|}
& {\rm VII} & 3  & 0 & $\rr$ & 0 & $\rr$ & 0 & $\rz$ & $e_{34}$ \\\hhline{|~|-|-|-|-|-|-|-|-|-|}
& ${\rm VIII}_\pm$ & 4 & $\rpm$ & $\rr$ & 0 & $\rr$ & 0 & $\rz$ & $\pm e_{12} + e_{34}$  \\
\hline
\multirow{4}{*}{$\mathfrak{s}_2$} & ${\rm I}_\pm$ &1 & $\rpm$ & 0 & 0 & 0 & 0 & 0 & $\pm e_{12}$\\\hhline{|~|-|-|-|-|-|-|-|-|-|}
& ${\rm II}_\pm$ &2  & $\rr$ & $\rpm$ & 0 & 0 & 0 & 0 & $\pm e_{13}$\\\hhline{|~|-|-|-|-|-|-|-|-|-|}
& ${\rm III}$ &3 & $\rr$ & $\rr$ & $\rz$ & 0 & 0 & 0 & $ e_{14}$\\\hhline{|~|-|-|-|-|-|-|-|-|-|}
& ${\rm IV}_\pm$ &3 & $\rr$ & $\rr$ & 0 & $\rpm$ & 0 & 0 & $\pm e_{23}$\\
\hline
\multirow{13}*{}{} & {\rm I} & 1  & $\rz$ & 0 & 0 & 0 & 0 & 0 & $e_{12}$ \\\hhline{|~|-|-|-|-|-|-|-|-|-|}
& {\rm II} & 1  & 0 & $\rz$ & 0 & 0 & 0 & 0 & $e_{13}$ \\\hhline{|~|-|-|-|-|-|-|-|-|-|}
& {\rm III} & 2  & $\rz$ & $\rz$ & 0 & 0 & 0 & 0 & $e_{12} + e_{13}$ \\\hhline{|~|-|-|-|-|-|-|-|-|-|}
& {\rm IV} & 1 & 0 & 0 & 0 & $\rz$ & 0 & 0 & $e_{23}$ \\\hhline{|~|-|-|-|-|-|-|-|-|-|}
& {\rm V} & 2 & $\rz$ & 0 & 0 & $\rz$ & 0 & 0 & $e_{12} + e_{23}$ \\\hhline{|~|-|-|-|-|-|-|-|-|-|}
$\mathfrak{s}^{\alpha,\beta}_3$& {\rm VI} & 2  & 0 & $\rz$ & 0 & $\rz$ & 0 & 0 & $e_{13} + e_{23}$ \\\hhline{|~|-|-|-|-|-|-|-|-|-|}
${\rm differ}$&  {${\rm VII}_{\pm}$} &  {3}  &  {$\rz$} &  {$\rz$} & {0} & $\pm x_1x_2x_4>0$  &  {0} &  {0} &  {$\pm(e_{12} + e_{13} + e_{23})$} \\
\hhline{|~|-|-|-|-|-|-|-|-|-|}
$\alpha,\beta,1$& {\rm VIII} & 3  & $\rr$ & $\rr$ & $\rz$ & 0 & 0 & 0 & $e_{14}$ \\\hhline{|~|-|-|-|-|-|-|-|-|-|}
& ${\rm IX}_{\alpha + \beta \in \{0, -1*\}}$ & 4  & $\rr$ & $\rr$ & $\rz$ & $\rz$ & 0 & 0 & $e_{14} + e_{23}$ \\\hhline{|~|-|-|-|-|-|-|-|-|-|}
& {\rm X} & 3 & $\rr$ & 0 & 0 & $\rr$ & $\rz$ & 0 & $e_{24}$ \\\hhline{|~|-|-|-|-|-|-|-|-|-|}
& ${\rm XI}_{\alpha+\beta=-1}$ & 4  & $\rr$ & $\rz$ & 0 & $\rr$ & $\rz$ & 0 & $e_{13} + e_{24}$ \\
\hhline{|~|-|-|-|-|-|-|-|-|-|}
& {\rm XIV} & 3  & 0 & $\rr$ & 0 & $\rr$ & 0 & $\rz$ & $e_{34}$ \\
\hhline{|~|-|-|-|-|-|-|-|-|-|}
& ${\rm XV}^{\alpha=-1}_{\alpha+\beta=-1,*}$ & 4  & $\rz$ & $\rr$ & 0 & $\rr$ & 0 & $\rz$ & $e_{12} + e_{34}$ \\
\hline
\hhline{|~|-|-|-|-|-|-|-|-|-|}
& $\mathcal{I}$& 2  & $\rr$ & $x_1^2+x_2^2\neq 0$ & $0$ & $0$ & $0$ & 0 & $e_{12}$ \\
\hhline{|~|-|-|-|-|-|-|-|-|-|}
& $\mathcal{II}$&1  & $0$ & $0$ & $0$ & $\rz$ & $0$ & 0 & $e_{23}$ \\
\hhline{|~|-|-|-|-|-|-|-|-|-|}
$\mathfrak{s}^{\alpha,\alpha}_3$& $\mathcal{III}$& 3  & $\rr$ & $x_1^2+x_2^2\neq 0$ & $0$ & $\rz$ & $0$ & 0 & $e_{12}+e_{23}$ \\
\hhline{|~|-|-|-|-|-|-|-|-|-|}
& $\mathcal{IV}$& 3  & $\rr$ & $\rr$ & $\rz$ & $0$ & $0$ & 0 & $e_{14}$ \\
\hhline{|~|-|-|-|-|-|-|-|-|-|}
$_{\tiny \alpha\neq 1}$& $\mathcal{V}^*_{\alpha=-1/2}$& 4  & $\rr$ & $\rr$ & $\rz$ & $\rz$ & $0$ & 0 & $e_{14}+e_{23}$ \\
\hhline{|~|-|-|-|-|-|-|-|-|-|}
& $\mathcal{VI}^*_{\alpha=-1/2}$& 5  & \multicolumn{2}{c|}{$x_2x_5-x_1x_6\neq 0$}& $0$ & $\rr$ & $\rr$& $x_5^2+x_6^2\neq 0$ & $e_{12}+e_{34}$ \\
\hhline{|~|-|-|-|-|-|-|-|-|-|}
& $\mathcal{VII}$& 4  & \multicolumn{2}{c|}{$x_2x_5-x_1x_6=0$}& $0$ & $\rr$ & $\rr$ & $x_5^2+x_6^2\neq 0$ & $e_{34}$ \\
\hline 
\hhline{|~|-|-|-|-|-|-|-|-|-|}
& $\mathscr{I}$ & 1  & $\rz$ & $0$ & $0$ & $0$ & $0$ & $0$ & $e_{12}$ \\
\hhline{|~|-|-|-|-|-|-|-|-|-|}
& $\mathscr{II}$ & 2  & $0$ & $\rr$ & $0$ & $x_2^2+x_4^2\neq 0$ & $0$ & $0$ & $e_{13}$ \\
\hhline{|~|-|-|-|-|-|-|-|-|-|}
$\mathfrak{s}^{1,\beta}_3$&  $\mathscr{III}$ & 3  & $\rz$ & $\rr$ & $0$ & $x_2^2+x_4^2\neq 0$ & $0$ & $0$ & $e_{12}+e_{13}$ \\
\hhline{|~|-|-|-|-|-|-|-|-|-|}
& $\mathscr{IV}$ & 4  & $\rr$ & $\rr$ & $x_3^2+x_5^2\neq 0$ & $\rr$ & $x_2x_5-x_3x_4=0$ & $0$ & $e_{13}$ \\
\hhline{|~|-|-|-|-|-|-|-|-|-|}
$_{\beta\neq 1}$& $\mathscr{V}_{\beta=-1}$ & 5  & $\rr$ & $\rr$ & $x_3^2+x_5^2\neq 0$ & $\rr$ & $x_2x_5-x_3x_4\neq 0$ & $0$ & $e_{14}+e_{23}$ \\
\hhline{|~|-|-|-|-|-|-|-|-|-|}
& $\mathscr{VI}$ & 4  & $0$ & $\rr$ & $0$ & $\rr$ & $0$ & $\rz$ & $e_{34}$ \\
\hline
\hhline{|~|-|-|-|-|-|-|-|-|-|}
& ${\rm i}$ & 3  & $\rr$ & $\rr$ & $0$ & $x_1^2+x_2^2+x_4^2\neq 0$ & $0$& $0$& $e_{12}$ \\
 \hhline{|~|-|-|-|-|-|-|-|-|-|}
   \multirow{5}*{}{$\mathfrak{s}^{1,1}_3$}& ${\rm ii}$ & 5  & \multicolumn{2}{c|}{$x_3x_4+\!x_1x_6=\!x_2x_5$}&   $\rr$ & $x_3x_4+x_1x_6=x_2x_5$ & $\rr$& $x_3^2+x_5^2+x_6^2\neq 0$ &$e_{13}+e_{34}$ \\
\hline
\multirow{8}{*}{$\mathfrak{s}^\alpha_4$} & {\rm I}$_\pm$ & 1  & $\rpm$ & 0 & 0 & 0 & 0 & 0 & $\pm e_{12}$ \\\hhline{|~|-|-|-|-|-|-|-|-|-|}
& {\rm II} & 1  & 0 & $\rz$ & 0 & 0 & 0 & 0 & $e_{13}$ \\\hhline{|~|-|-|-|-|-|-|-|-|-|}
\multirow{8}{*}{{\footnotesize $\alpha\notin \{0,1\}$}}& {\rm III}$_\pm$ & 2  &$\rpm$ & $\rz$ & 0 & 0 & 0 & 0 & $\pm e_{12} + e_{13}$ \\\hhline{|~|-|-|-|-|-|-|-|-|-|}
& {\rm IV} & 3  & $\rr$ & $\rr$ & $\rz$ & 0 & 0 & 0 & $e_{14}$ \\\hhline{|~|-|-|-|-|-|-|-|-|-|}
& {\rm V} & 2  & 0 & $\rr$ & 0 & $\rz$ & 0 & 0 & $e_{23}$ \\\hhline{|~|-|-|-|-|-|-|-|-|-|}
& {\rm VI}$_\pm$ & 3  & $\rpm$ & $\rr$ & 0 & $\rz$ & 0 & 0 & $\pm e_{12} + e_{23}$ \\\hhline{|~|-|-|-|-|-|-|-|-|-|}
& ${\rm VII}_{\alpha \in \{-2^*,-1\}}$ & 4  & $\rr$ & $\rr$ & $\rz$ & $\rz$ & 0 & 0 & $e_{14} + e_{23}$ \\\hhline{|~|-|-|-|-|-|-|-|-|-|}
& {\rm VIII} & 3  & 0 & $\rr$ & 0 & $\rr$ & 0 & $\rz$ & $e_{34}$ \\\hhline{|~|-|-|-|-|-|-|-|-|-|}
& ${\rm IX}^{\alpha = -2,*}_\pm$ & 4  & $\rpm$ & $\rr$ & 0 & $\rr$ & 0 & $\rz$ & $\pm e_{12} + e_{34}$ \\\hhline{|~|-|-|-|-|-|-|-|-|-|}
& ${\rm X}_{\alpha = 1}$ & 4  & $\rr$ & $\rr$ & $\rz$ & $-\frac{x_1x_6}{x_3}$ & 0 & $\rz$ & $e_{14} + e_{34}$ \\
\hline
\multirow{7}{*}{$\mathfrak{s}^1_4$} & $\mathcal{I}_\pm$ & 2  & $\rpm$ & $
\rr$ & 0 & 0 & 0 & 0 & $\pm e_{12}$ \\\hhline{|~|-|-|-|-|-|-|-|-|-|}
& $\mathcal{II}$& 1  & 0 & $\rz$ & 0 & 0 & 0 & 0 & $e_{13}$ \\\hhline{|~|-|-|-|-|-|-|-|-|-|}
& $\mathcal{III}$ & 3  & $\rr$ & $\rr$ & $\rz$ & 0 & 0 & 0 & $e_{14}$ \\\hhline{|~|-|-|-|-|-|-|-|-|-|}
& $\mathcal{IV}$ & 3  & $\rr$ & $\rr$ & 0 & $\rz$ & 0 & 0 & $  e_{23}$ \\
\hhline{|~|-|-|-|-|-|-|-|-|-|}
& $\mathcal{V}$ & 4  & $\rr$ & $\rr$ & \multicolumn{2}{c|}{$x_4x_3-x_1x_6=0$} & 0 & $\rz$ & $e_{14} + e_{34}$ \\
\hline
\multirow{5}{*}{$\mathfrak{s}_5$} & {\rm I} & 1  & $\rr$ & $x_1^2 + x_2^2 \neq 0$ & 0 & 0 & 0 & 0 & $e_{12}$ \\\hhline{|~|-|-|-|-|-|-|-|-|-|}
& ${\rm II}_{\pm}$ & 1  & 0 & 0 & 0 & $\rpm$ & 0 & 0 & $\pm e_{23}$ \\\hhline{|~|-|-|-|-|-|-|-|-|-|}
& ${\rm III}_{\pm}$ & 2 & $\rr$ & $x_1^2 + x_2^2 \neq 0$ & 0 & $\rpm$ & 0 & 0 & $e_{12} \pm e_{23}$ \\\hhline{|~|-|-|-|-|-|-|-|-|-|}
& {\rm IV} & 3  & $\rr$ & $\rr$ & $\rz$ & 0 & 0 & 0 & $e_{14}$ \\\hhline{|~|-|-|-|-|-|-|-|-|-|}
& $({\rm V}_\pm)^{\beta = 0}_{{\alpha=-2\beta}*}$ & 4  & $\rr$ & $\rr$ & $\rz$ & $\rpm$ & 0 & 0 & $e_{14} \pm e_{23}$ \\
\hline
 \multirow{7}{*}{$\mathfrak{s}_6$}& {\rm I} & 1  & $\rz$ & 0 & 0 & 0 & 0 & 0 & $e_{12}$ \\\hhline{|~|-|-|-|-|-|-|-|-|-|}
& {\rm I}$^{\rm ext}$ & 1  & 0 & $\rz$ & 0 & 0 & 0 & 0 & \\\hhline{|~|-|-|-|-|-|-|-|-|-|}
& {\rm II} & 2  & $\rz$ & $\rz$ & 0 & 0 & 0 & 0 & $e_{12} + e_{13}$ \\\hhline{|~|-|-|-|-|-|-|-|-|-|}
& {\rm III} & 3  & $\rr$ & $\rr$ & $\rz$ & 0 & 0 & 0 & $e_{14}$ \\\hhline{|~|-|-|-|-|-|-|-|-|-|}
& {\rm IV}$^*$ & 3  & $\rr$ & $\rr$ & 0 & $\rz$ & 0 & 0 & $e_{23}$ \\  \hhline{|~|-|-|-|-|-|-|-|-|-|}
& {\rm V}$^*$ & 2  & $\rr$ & 0 & $x_4$ & $\rz$ & 0 & 0 & $e_{14} + e_{23}$ \\
\hhline{|~|-|-|-|-|-|-|-|-|-|}
& ${\rm V}^{{\rm ext}*}$ & 2  & 0 & $\rr$ & $-x_4$ & $\rz$ & 0 & 0 &  \\

\hline
	\end{tabular}
	}
	\end{table}
}

{\scriptsize
\begin{table}[h]
\centering
\resizebox{\columnwidth}{!}{
\begin{tabular}{|c|c|c|c|c|c|c|c|c|c|}
\hline
\theadstart
\thead $\mathfrak{g}$ & \textrm{Orbit} & {\rm Dim.}  & $x_1$ & $x_2$ & $x_3$ & $x_4$ & $x_5$ & $x_6$ & {\rm Repr. element} \\\hline
\hhline{|~|-|-|-|-|-|-|-|-|-|}
$\multirow{8}{*}{$\mathfrak{s}_6$}$& {\rm VI}$^*$ & 3  & $\rr$ & $\rz$ & $x_4$ & $\rz$ & 0 & 0 & $e_{13} + e_{14} + e_{23}$ \\\hhline{|~|-|-|-|-|-|-|-|-|-|}
& ${\rm VI}^{{\rm ext}*}$ & 3  & $\rz$ & $\rr$ & $-x_4$ & $\rz$ & 0 & 0 &  \\\hhline{|~|-|-|-|-|-|-|-|-|-|}	
& ${\rm VII}^*_{|k| \notin \{0,1\}}$ & 3 & $\rr$ & $\rr$ & $kx_4$ & $\rz$ & 0 & 0 & $k e_{14} + e_{23}$ \\
\hhline{|~|-|-|-|-|-|-|-|-|-|}
& {\rm VIII} & 4 & $\rr$ & $-x_4^2/x_5$ & $\rr$ & $-x_3$  & $\rz$ & 0 & $e_{24}$ \\
\hhline{|~|-|-|-|-|-|-|-|-|-|}
& {\rm VIII}$^{\rm ext}$ & 3 & $\frac{-x_4^2}{x_6}$ & $\rr$ & $\rr$ & $x_3$  & 0 & $\rz$ &   \\\hhline{|~|-|-|-|-|-|-|-|-|-|}
& {\rm IX}$^*$ & 4 & $\rr$ & $x_2+x_4^2/x_5\neq 0$ & $\rr$ & $-x_3$  & $\rz$ & 0 & $e_{24}+e_{13}$ \\
\hhline{|~|-|-|-|-|-|-|-|-|-|}
& ${\rm IX}^{{\rm ext}*}$ & 4 & $x _1 + \frac{x_4^2}{x_6} \neq 0$ & $\rr$ & $\rr$ & $x_3$  & 0 & $\rz$ &  \\\hhline{|-|-|-|-|-|-|-|-|-|-|}
\multirow{4}{*}{$\mathfrak{s}_7$} & {\rm I} & 2 & $\rr$ & $x_1^2 + x_2^2 \neq 0$ & 0 & 0 & 0 & 0 & $e_{12}$ \\\hhline{|~|-|-|-|-|-|-|-|-|-|}
& ${\rm II}_{\pm}$ & 3  & $\rr$ & $\rr$ & $\rpm$ & 0 & 0 & 0 & $\pm e_{14}$ \\\hhline{|~|-|-|-|-|-|-|-|-|-|}
& {\rm III}$^*$ & 3  & $\rr$ & $\rr$ & 0 & $\rz$ & 0 & 0 & $e_{23}$ \\\hhline{|~|-|-|-|-|-|-|-|-|-|}
& $({\rm IV}_{\pm})^*_{|k|\in \mathbb{R}_+}$ & 3  & $\rr$ & $\rr$ & $\rpm$ & $k x_3$ & 0 & 0& $\pm e_{14} \pm k e_{23}$ 
\\
\hline
\multirow{11}{*}{$\mathfrak{s}_8$} & {\rm I} & 1  & $\rz$ & 0 & 0 & 0 & 0 & 0 & $e_{12}$ \\\hhline{|~|-|-|-|-|-|-|-|-|-|}
& {\rm II} & 1  & 0 & $\rz$ & 0 & 0 & 0 & 0 & $e_{13}$ \\\hhline{|~|-|-|-|-|-|-|-|-|-|}
& {\rm III} & 2  & $\rz$ & $\rz$ & 0 & 0 & 0 & 0 & $e_{12} + e_{13}$ \\\hhline{|~|-|-|-|-|-|-|-|-|-|}
& {\rm IV} & 3  & $\rr$ & $\rr$ & $\rz$ & 0 & 0 & 0 & $e_{14}$ \\\hhline{|~|-|-|-|-|-|-|-|-|-|}
& {\rm V} & 3  & $\rr$ & $\rr$ & $\rz$ & $-(1+\alpha)x_3$ & 0 & 0 & $e_{14} - (1 + \alpha) e_{23}$ \\\hhline{|~|-|-|-|-|-|-|-|-|-|}
& {\rm VI} & 3  & $\rr$ &  $\frac{\alpha x_3^2}{x_5}$ & $\rr$ & $\alpha x_3$ & $\rz$ & 0 & $e_{24}$ \\\hhline{|~|-|-|-|-|-|-|-|-|-|}
&  {${\rm VII}^{\alpha=-1/2}_\pm$} &  {4}  &  {$\rr$} &  {$x_2x_5-\alpha x_3^2\in \rpm $} &  {$\rr$} &  {$\alpha x_3$} &  {$\rz$} &  {0} & $e_{13} \pm  e_{24}$ 
\\\hhline{|~|-|-|-|-|-|-|-|-|-|}
& {\rm VIII} & 3  & $-\frac{x_3^2}{x_6}$ & $\rr$ & $\rr$ & $x_3$ & 0 & $\rz$ & $e_{34}$ \\
\hline
\multirow{4}{*}{$\mathfrak{s}^1_8$} & $\mathcal{I}$ & 2  & $\rr$ & $x_1^2+x_2^2\neq 0$ & 0 & 0 & 0 & 0 & $e_{12}$ \\\hhline{|~|-|-|-|-|-|-|-|-|-|}
& $\mathcal{II}$ & 3  & $\rr$ & $\rr$ & $\rz$ & 0 & 0 & 0 & $e_{14}$ \\\hhline{|~|-|-|-|-|-|-|-|-|-|}
& $\mathcal{III}$  & 3  & $\rr$ & $\rr$ & $\rz$ & $-2x_3$ & 0 & 0 & $e_{14} -2 e_{23}$ \\\hhline{|~|-|-|-|-|-|-|-|-|-|}
&$\mathcal{IV}$ & 4  & \multicolumn{3}{c|}{$x_1x_6-x_2x_5+x_3^2=0$}  & $ x_3$ & $\rr$ & $x_5^2+x_6^2\neq 0$ & $e_{24}$ 
\\
\hline
\multirow{3}{*}{$\mathfrak{s}_9$} & {\rm I} & 2 & $\rr$ & $x_1^2 + x_2^2 \neq 0$ & 0 & 0 & 0 & 0 & $e_{12}$ \\\hhline{|~|-|-|-|-|-|-|-|-|-|}
& {\rm II}$_\pm$ & 2  & $\rr$ & $\rr$ & $\rpm$ & 0 & 0 & 0 & $\pm e_{14}$ \\\hhline{|~|-|-|-|-|-|-|-|-|-|}
& {\rm III}$_\pm$ & 3  & $\rr$ & $\rr$ & $\rpm$ & $-2\alpha x_3$ & 0 & 0 & $\pm 2\alpha e_{23} \mp e_{14}$ \\
\hline
\multirow{5}{*}{$\mathfrak{s}_{10}$} & {\rm I} & 1  & $\rz$ & 0 & 0 & 0 & 0 & 0 & $e_{12}$ \\\hhline{|~|-|-|-|-|-|-|-|-|-|}
& {\rm II} & 2  & $\rr$ & $\rz$ & 0 & 0 & 0 & 0 & $e_{13}$ \\\hhline{|~|-|-|-|-|-|-|-|-|-|}
& {\rm III}$_\pm$ & 3  & $\rr$ & $\rr$ & $\rpm$ & 0 & 0 & 0 & $\pm e_{14}$ \\\hhline{|~|-|-|-|-|-|-|-|-|-|}
& {\rm IV}$_\pm$ & 3  & $\rr$ & $\rr$ & $\rpm$ & $-2x_3$ & 0 & 0 & $\pm e_{14} \mp 2 e_{23}$ \\\hhline{|~|-|-|-|-|-|-|-|-|-|}
& {\rm V} & 3  & $\rr$ & $\frac{x_3^2}{x_5}$ & $\rr$ & $x_3$ & $\rz$ & 0 & $e_{24}$ \\
\hline
\multirow{9}{*}{$\mathfrak{s}_{11}$} & {\rm I} & 1  & $\rz$ & 0 & 0 & 0 & 0 & 0 & $e_{12}$ \\\hhline{|~|-|-|-|-|-|-|-|-|-|}
& {\rm II} & 1  & 0 & $\rz$ & 0 & 0 & 0 & 0 & $e_{13}$ \\\hhline{|~|-|-|-|-|-|-|-|-|-|}
& {\rm III} & 2  & $\rz$ & $\rz$ & 0 & 0 & 0 & 0 & $e_{12} + e_{13}$ \\\hhline{|~|-|-|-|-|-|-|-|-|-|}
& {\rm IV} & 2  & $\rr$ & 0 & $\rz$ & 0 & 0 & 0 & $e_{14}$ \\\hhline{|~|-|-|-|-|-|-|-|-|-|}
& {\rm V} & 3  & $\rr$ & $\rz$ & $\rz$ & 0 & 0 & 0 & $e_{13} + e_{14}$ \\\hhline{|~|-|-|-|-|-|-|-|-|-|}
& {\rm VI} & 2  & 0 & $\rr$ & $\rz$ & $-x_3$ & 0 & 0 & $e_{14} - e_{23}$ \\\hhline{|~|-|-|-|-|-|-|-|-|-|}
& {\rm VII} & 3  & $\rz$ & $\rr$ & $\rz$ & $-x_3$ & 0 & 0 & $e_{12} + e_{14} - e_{23}$ \\\hhline{|~|-|-|-|-|-|-|-|-|-|}
& {\rm VIII} & 3 & $\rr$ & 0 & $\rr$ & 0 & $\rz$ & 0 & $e_{24}$ \\\hhline{|~|-|-|-|-|-|-|-|-|-|}
& {\rm IX} & 3  & $-\frac{x_3^2}{x_6}$ & $\rr$ & $\rr$ & $x_3$ & 0 & $\rz$ & $e_{34}$ \\
\hline
\multirow{4}{*}{$\mathfrak{s}_{12}$} & {\rm I} & 1  & $\rz$ & 0 & 0 & 0 & 0 & 0 & $e_{12}$ \\\hhline{|~|-|-|-|-|-|-|-|-|-|}
& {\rm II} & 3 & $\rr$ & $\rr$ & 0 & $x_2^2 + x_4^2 \neq 0$ & 0 & 0 & $e_{13}$ \\\hhline{|~|-|-|-|-|-|-|-|-|-|}
& {\rm III} & 3  & $\rr$ & $x_5$ & $\rr$ & $-x_3$ & $x_3^2 + x_5^2 \neq 0$ & 0 & $e_{13} + e_{24}$ \\\hhline{|~|-|-|-|-|-|-|-|-|-|}
& {\rm IV}$_{|k|>0}$ & 2  & $-\frac{x_2^2 + x_3^2}{x_6}$ & $\rr$ & $\rr$ & $x_3$ & $-x_2$ & $k$ & $ke_{34}$ \\
\hline
\multirow{7}{*}{$\mathfrak{n}_{1}$} & {\rm I} & 1 & $\rz$ & 0 & 0 & 0 & 0 & 0 & $e_{12}$ \\\hhline{|~|-|-|-|-|-|-|-|-|-|}
& {\rm II}$_\pm$ & 2 & $\rr$ & $\rpm$ & 0 & 0 & 0 & 0 & $\pm e_{13}$ \\\hhline{|~|-|-|-|-|-|-|-|-|-|}
& {\rm III} & 3 & $\rr$ & $\rr$ & $\rz$ & 0 & 0 & 0 & $e_{14}$ \\\hhline{|~|-|-|-|-|-|-|-|-|-|}
& {\rm IV} & 3 & $\rr$ & $\rr$ & 0 & $\rz$ & 0 & 0 & $e_{23}$ \\\hhline{|~|-|-|-|-|-|-|-|-|-|}
& {\rm V} & 4 & $\rr$ & $\rr$ & $\rz$ & $\rz$ & 0 & 0 & $e_{14} + e_{23}$ \\\hhline{|~|-|-|-|-|-|-|-|-|-|}
& {\rm VI}$^*$ & 4 & $\rr$ & $ \frac{x_3 x_4}{x_5}$ & $\rr$ & $\rr$ & $\rz$ & 0 & $e_{24}$ \\\hhline{|~|-|-|-|-|-|-|-|-|-|}
& {\rm VII}$_{\pm}^*$ & 5 & $\rr$& $x_2- \frac{x_3 x_4}{x_5}\in \rpm$&$\rr$ &$\rr$ & $\rz$ & 0 & $e_{24}\pm e_{13}$\\
\hline
\end{tabular}
}
\caption{Orbits of the action of ${\rm Aut} (\mathfrak{g})$ on $\mathcal{Y}_{\mathfrak{g}}$, their dimensions, elements, and representatives, for real four-dimensional indecomposable Lie algebras $\mathfrak{g}$. Each Lie algebra is divided into several subsets enumerated by roman numbers, which classify the orbits of ${\rm Aut}(\mathfrak{g})$ in $\mathcal{Y}_{\mathfrak{g}}$. The orbits of ${\rm Aut}_c(\mathfrak{g})$ are given by the (topologically) connected components of each orbit of ${\rm Aut}(\mathfrak{g})$. The trivial orbits given by $0\in\Lambda^2\mathfrak{g}$ are not considered. Symbol $A_\pm$ stands for two orbits, one with $+$ and other with $-$. In these cases, $\mathbb{R}_{\pm}$ stands for $\mathbb{R}_+$ for the first orbit and $\mathbb{R}_-$ for the second. If an orbit of Aut$(\mathfrak{g})$ in $\Lambda^2\mathfrak{g}$ belongs to $\mathcal{Y}_{\mathfrak{g}}$ only for a certain set of parameters, each family of possible values of the parameters is indicated in a subindex, first, or as a superindex,  if a second family of parameters is available. A star $(*)$ is used to denote $r$-matrices that are not solutions to the CYBE.}\label{Tab:g_orb_1}
\label{Tab:g_orb_2}
\end{table}
}


{\scriptsize
\begin{table}[h]
\centering
\begin{tabular}{|c|c||c|c|c|c|c|c|}
\hline
& & $e_{12}$ & $e_{13}$ & $e_{14}$ & $e_{23}$ & $e_{24}$ & $e_{34}$ \\
\hhline{|=|=#=|=|=|=|=|=|}
\multirow{4}{*}{$\mathfrak{s}_1$} & $e_1$ & 0 & 0 & 0 & 0 & 0 & 0 \\
& $e_2$ & 0 & 0 & 0 & 0 & $e_{12}$ & $e_{13}$ \\
& $e_3$ & 0 & 0 & $-e_{13}$ & 0 & $-e_{23}$ & 0 \\
& $e_4$ & 0 & $e_{13}$ & 0 & $e_{13} + e_{23}$ & $e_{14}$ & $e_{34}$ \\
\hline
\multirow{4}{*}{$\mathfrak{s}_2$} & $e_1$ & 0 & 0 & 0 & 0 & $e_{12}$ & $e_{13}$ \\
& $e_2$ & 0 & 0 & $-e_{12}$ & 0 & $e_{12}$ & $e_{13} + e_{23}$ \\
& $e_3$ & 0 & 0 & $-e_{12} - e_{13}$ & 0 & $-e_{23}$ & $e_{23}$ \\
& $e_4$ & $2e_{12}$ & $e_{12} + 2e_{13}$ & $e_{14}$ & $e_{13} + 2e_{23}$ & $e_{14} + e_{24}$ & $e_{24} + e_{34}$ \\
\hline
\multirow{4}{*}{$\mathfrak{s}_3$} & $e_1$ & 0 & 0 & 0 & 0 & $e_{12}$ & $e_{13}$ \\
& $e_2$ & 0 & 0 & $-\alpha e_{12}$ & 0 & 0 & $\alpha e_{23}$ \\
& $e_3$ & 0 & 0 & $-\beta e_{13}$ & 0 & $-\beta e_{23}$ & 0 \\
& $e_4$ & $(1+ \alpha)e_{12}$ & $(1 + \beta)e_{13}$ & $e_{14}$ & $(\alpha + \beta)e_{23}$ & $\alpha e_{24}$ & $\beta e_{34}$ \\
\hline
\multirow{4}{*}{$\mathfrak{s}_4$} & $e_1$ & 0 & 0 & 0 & 0 & $e_{12}$ & $e_{13}$ \\
& $e_2$ & 0 & 0 & $-e_{12}$ & 0 & $e_{12}$ & $e_{13} + e_{23}$ \\
& $e_3$ & 0 & 0 & $-\alpha e_{13}$ & 0 & $-\alpha e_{23}$ & 0 \\
& $e_4$ & $2e_{12}$ & $(1 + \alpha)e_{13}$ & $e_{14}$ & $e_{13} + (1 + \alpha)e_{23}$ & $e_{14} + e_{24}$ & $\alpha e_{34}$ \\
\hline
\multirow{4}{*}{$\mathfrak{s}_5$} & $e_1$ & 0 & 0 & 0 & 0 & $\alpha e_{12}$ & $\alpha e_{13}$ \\
& $e_2$ & 0 & 0 & $e_{13} - \beta e_{12}$ & 0 & $e_{23}$ & $\beta e_{23}$ \\
& $e_3$ & 0 & 0 & $-e_{12} - \beta e_{13}$ & 0 & $-\beta e_{23}$ & $e_{23}$ \\
& $e_4$ & $(\alpha + \beta) e_{12} - e_{13}$ & $e_{12} + (\alpha + \beta) e_{13}$ & $\alpha e_{14}$ & $2\beta e_{23}$ & $\beta e_{24} - e_{34}$ & $e_{24} + \beta e_{34}$ \\
\hline
\multirow{4}{*}{$\mathfrak{s}_6$} & $e_1$ & 0 & 0 & 0 & 0 & 0 & 0 \\
& $e_2$ & 0 & 0 & $-e_{12}$ & $-e_{12}$ & 0 & $e_{14} + e_{23}$ \\
& $e_3$ & 0 & 0 & $e_{13}$ & $-e_{13}$ & $-e_{14} + e_{23}$ & 0 \\
& $e_4$ & $e_{12}$ & $-e_{13}$ & 0 & 0 & $e_{24}$ & $-e_{34}$ \\
\hline
\multirow{4}{*}{$\mathfrak{s}_7$} & $e_1$ & 0 & 0 & 0 & 0 & 0 & 0 \\
& $e_2$ & 0 & 0 & $e_{13}$ & $-e_{12}$ & $e_{23}$ & $e_{14}$ \\
& $e_3$ & 0 & 0 & $-e_{12}$ & $-e_{13}$ & $-e_{14}$ & $e_{23}$ \\
& $e_4$ & $-e_{13}$ & $e_{12}$ & 0 & 0 & $-e_{34}$ & $e_{24}$ \\
\hline
\multirow{4}{*}{$\mathfrak{s}_8$} & $e_1$ & 0 & 0 & 0 & 0 & $(1 + \alpha)e_{12}$ & $(1 + \alpha)e_{13}$ \\
& $e_2$ & 0 & 0 & $-e_{12}$ & $-e_{12}$ & 0 & $e_{14} + e_{23}$ \\
& $e_3$ & 0 & 0 & $-\alpha e_{13}$ & $-e_{13}$ & $-e_{14} - \alpha e_{23}$ & 0 \\
& $e_4$ & $(2 + \alpha)e_{12}$ & $(1 + 2\alpha)e_{13}$ & $(1 + \alpha)e_{14}$ & $(1 + \alpha)e_{23}$ & $e_{24}$ & $\alpha e_{34}$ \\
\hline
\multirow{4}{*}{$\mathfrak{s}_9$} & $e_1$ & 0 & 0 & 0 & 0 & $2\alpha e_{12}$ & $2\alpha e_{13}$ \\
& $e_2$ & 0 & 0 & $e_{13} - \alpha e_{12}$ & $-e_{12}$ & $e_{23}$ & $e_{14} + \alpha e_{23}$ \\
& $e_3$ & 0 & 0 & $-e_{12} - \alpha e_{13}$ & $-e_{13}$ & $-e_{14} - \alpha e_{23}$ & $e_{23}$ \\
& $e_4$ & $3\alpha e_{12} - e_{13}$ & $3\alpha e_{13} + e_{12}$ & $2\alpha e_{14}$ & $2\alpha e_{23}$ & $\alpha e_{24} - e_{34}$ & $e_{24} + \alpha e_{34}$ \\
\hline
\multirow{4}{*}{$\mathfrak{s}_{10}$} & $e_1$ & 0 & 0 & 0 & 0 & $2e_{12}$ & $2e_{13}$ \\
& $e_2$ & 0 & 0 & $-e_{12}$ & $-e_{12}$ & 0 & $e_{14} + e_{23}$ \\
& $e_3$ & 0 & 0 & $-e_{12} - e_{13}$ & $-e_{13}$ & $-e_{14} - e_{23}$ & $e_{23}$ \\
& $e_4$ & $3e_{12}$ & $e_{12} + 3e_{13}$ & $2e_{14}$ & $2e_{23}$ & $e_{24}$ & $e_{24} + e_{34}$ \\
\hline
\multirow{4}{*}{$\mathfrak{s}_{11}$} & $e_1$ & 0 & 0 & 0 & 0 & $e_{12}$ & $e_{13}$ \\
& $e_2$ & 0 & 0 & $-e_{12}$ & $-e_{12}$ & 0 & $e_{23} + e_{14}$ \\
& $e_3$ & 0 & 0 & 0 & $-e_{13}$ & $-e_{14}$ & 0 \\
& $e_4$ & $2e_{12}$ & $e_{13}$ & $e_{14}$ & $e_{23}$ & $e_{24}$ & 0 \\
\hline
\multirow{4}{*}{$\mathfrak{s}_{12}$} & $e_1$ & 0 & 0 & $e_{12}$ & $e_{12}$ & 0 & $-e_{14} - e_{23}$ \\
& $e_2$ & 0 & $-e_{12}$ & 0 & 0 & $e_{12}$ & $-e_{24} + e_{13}$ \\
& $e_3$ & $2e_{12}$ & $e_{13}$ & $e_{14}$ & $e_{23}$ & $e_{24}$ & 0 \\
& $e_4$ & 0 & $-e_{23}$ & $-e_{24}$ & $e_{13}$ & $e_{14}$ & 0 \\
\hline
\multirow{4}{*}{$\mathfrak{n}_{1}$} & $e_1$ & 0 & 0 & 0 & 0 & 0 & 0 \\
& $e_2$ & 0 & 0 & 0 & 0 & $-e_{12}$ & $-e_{13}$ \\
& $e_3$ & 0 & 0 & $e_{12}$ & 0 & 0 & $-e_{23}$ \\
& $e_4$ & 0 & $-e_{12}$ & 0 & $-e_{13}$ & $-e_{14}$ & $-e_{24}$ \\
\hline
\end{tabular}
\caption{Schouten brackets between basis elements of $\mathfrak{g}$ and $\Lambda^2\mathfrak{g}$ for real  four-dimensional indecomposable Lie algebras.} \label{Tab:Schoten_1_2}
\end{table}
}

\begin{landscape}
{\scriptsize
\begin{table}[h]
\centering
\begingroup
\renewcommand{\arraystretch}{1.0749}
\begin{tabular}[t]{|c|c||c|c|c|c|c|c|}
\hline 
& & $e_{12}$ & $e_{13}$ & $e_{14}$ & $e_{23}$ & $e_{24}$ & $e_{34}$ \\[1pt]
\hhline{|=|=#=|=|=|=|=|=|}
\multirow{6}{*}{$\mathfrak{s}_1$} & $e_{12}$ & 0 & 0 & 0 & 0 & 0 & 0 \\\hhline{~|-|-|-|-|-|-|-|}
& $e_{13}$ & & 0 & 0 & 0 & $-e_{123}$ & 0 \\\hhline{~|-|-|-|-|-|-|-|}
& $e_{14}$ & & & 0 & $e_{123}$ & 0 & $e_{134}$ \\\hhline{~|-|-|-|-|-|-|-|}
& $e_{23}$ & & & & 0 & $-e_{123}$ & 0 \\\hhline{~|-|-|-|-|-|-|-|}
& $e_{24}$ & & & & & $-2e_{124}$ & $-e_{134} + e_{234}$ \\\hhline{~|-|-|-|-|-|-|-|}
& $e_{34}$ & & & & & & 0 \\
\hline
\multirow{6}{*}{$\mathfrak{s}_2$} & $e_{12}$ & 0 & 0 & 0 & 0 & 0 & $2e_{123}$ \\\hhline{~|-|-|-|-|-|-|-|}
& $e_{13}$ &  & 0 & 0 & 0 & $-2e_{123}$ & $e_{123}$ \\\hhline{~|-|-|-|-|-|-|-|}
& $e_{14}$ & & & 0 & $2e_{123}$ & 0 & $e_{124}$ \\\hhline{~|-|-|-|-|-|-|-|}
& $e_{23}$ &  &  &  & 0 & $-e_{123}$ & 0 \\\hhline{~|-|-|-|-|-|-|-|}
& $e_{24}$ & & & & & $-2e_{124}$ & $-e_{134}$ \\\hhline{~|-|-|-|-|-|-|-|}
& $e_{34}$ & & & & & & $-2e_{234}$ \\
\hline
\multirow{6}{*}{$\mathfrak{s}_3$} & $e_{12}$ & 0 & 0 & 0 & 0 & 0 & $(1 + \alpha)e_{123}$ \\\hhline{~|-|-|-|-|-|-|-|}
& $e_{13}$ & & 0 & 0 & 0 & $-(1 + \beta)e_{123}$ & 0 \\\hhline{~|-|-|-|-|-|-|-|}
& $e_{14}$ & & & 0 & $(\alpha + \beta)e_{123}$ & $(\alpha - 1)e_{124}$ & $(\beta - 1)e_{134}$ \\\hhline{~|-|-|-|-|-|-|-|}
& $e_{23}$ & & & & 0 & 0 & 0 \\\hhline{~|-|-|-|-|-|-|-|}
& $e_{24}$ & & & & & 0 & $(\beta - \alpha)e_{234}$ \\\hhline{~|-|-|-|-|-|-|-|}
& $e_{34}$ & & & & & & 0 \\
\hline
\multirow{6}{*}{$\mathfrak{s}_4$} & $e_{12}$ & 0 & 0 & 0 & 0 & 0 & $2e_{123}$ \\\hhline{~|-|-|-|-|-|-|-|}
& $e_{13}$ & & 0 & 0 & 0 & $-(1 + \alpha)e_{123}$ & 0 \\\hhline{~|-|-|-|-|-|-|-|}
& $e_{14}$ & & & 0 & $(1 + \alpha) e_{123}$ & 0 & $(\alpha - 1)e_{134}$ \\\hhline{~|-|-|-|-|-|-|-|}
& $e_{23}$ & & & & 0 & $-e_{123}$ & 0 \\\hhline{~|-|-|-|-|-|-|-|}
& $e_{24}$ & & & & & $-2e_{124}$ & \parbox{1.8cm}{\vspace{2pt} $(\alpha - 1)e_{234} \\-e_{134}$ \vspace{2pt}} \\\hhline{~|-|-|-|-|-|-|-|}
& $e_{34}$ & & & & & & 0 \\
\hline
\multirow{6}{*}{$\mathfrak{s}_5$} & $e_{12}$ & 0 & 0 & 0 & 0 & $e_{123}$ & $(\alpha + \beta)e_{123}$ \\\hhline{~|-|-|-|-|-|-|-|}
& $e_{13}$ & & 0 & 0 & 0 & $-(\alpha + \beta)e_{123}$ & $e_{123}$ \\\hhline{~|-|-|-|-|-|-|-|}
& $e_{14}$ & & & 0 & $2\beta e_{123}$ & \parbox{1.7cm}{\vspace{3pt}$(\beta - \alpha)e_{124} \\- e_{134}$ \vspace{2pt}} & \parbox{1.8cm}{$(\beta - \alpha)e_{134} \\+ e_{124}$ \vspace{2pt}} \\\hhline{~|-|-|-|-|-|-|-|}
& $e_{23}$ & & & & 0 & 0 & 0 \\\hhline{~|-|-|-|-|-|-|-|}
& $e_{24}$ & & & & & $-2e_{234}$ & 0 \\\hhline{~|-|-|-|-|-|-|-|}
& $e_{34}$ &  & & & & & $-2e_{234}$ \\
\hline
\multirow{6}{*}{$\mathfrak{s}_6$} & $e_{12}$ & 0 & 0 & 0 & 0 & 0 & $e_{123}$ \\\hhline{~|-|-|-|-|-|-|-|}
& $e_{13}$ & & 0 & 0 & 0 & $e_{123}$ & 0 \\\hhline{~|-|-|-|-|-|-|-|}
& $e_{14}$ & & & 0 & 0 & $e_{124}$ & $-e_{134}$ \\\hhline{~|-|-|-|-|-|-|-|}
& $e_{23}$ & & & & $2e_{123}$ & $e_{124}$ & $e_{134}$ \\\hhline{~|-|-|-|-|-|-|-|}
& $e_{24}$ & & & & & 0 & $-2e_{234}$ \\\hhline{~|-|-|-|-|-|-|-|}
& $e_{34}$ & & & & & & 0 \\
\hline
\end{tabular} 
\endgroup
\hfill
\begin{tabular}[t]{|c|c||c|c|c|c|c|c|}
\hhline{|-|-|-|-|-|-|-|-|}
 & $e_{12}$ & 0 & 0 & 0 & 0 & $e_{123}$ & 0 \\\hhline{~|-|-|-|-|-|-|-|}
& $e_{13}$ & & 0 & 0 & 0 & 0 & $e_{123}$ \\\hhline{~|-|-|-|-|-|-|-|}
$\mathfrak{s}_7$ & $e_{14}$ & & & 0 & 0 & $-e_{134}$ & $e_{124}$ \\\hhline{~|-|-|-|-|-|-|-|}
& $e_{23}$ & & & & $2e_{123}$ & $e_{124}$ & $e_{134}$ \\\hhline{~|-|-|-|-|-|-|-|}
& $e_{24}$ & & & & & $-2e_{234}$ & 0 \\\hhline{~|-|-|-|-|-|-|-|}
& $e_{34}$ & & & & & & $-2e_{234}$ \\
\hline
\multirow{6}{*}{$\mathfrak{s}_8$} & $e_{12}$ & 0 & 0 & 0 & 0 & 0 & $(2 + \alpha)e_{123}$ \\\hhline{~|-|-|-|-|-|-|-|}
& $e_{13}$ & & 0 & 0 & 0 & $-(1 + 2\alpha)e_{123}$ & 0 \\\hhline{~|-|-|-|-|-|-|-|}
& $e_{14}$ & & & 0 & $(1 + \alpha)e_{123}$ & $-\alpha e_{124}$ & $-e_{134}$ \\\hhline{~|-|-|-|-|-|-|-|}
& $e_{23}$ & & & & $2e_{123}$ & $e_{124}$ & $e_{134}$ \\\hhline{~|-|-|-|-|-|-|-|}
& $e_{24}$ & & & & & 0 & $(\alpha - 1)e_{234}$ \\\hhline{~|-|-|-|-|-|-|-|}
& $e_{34}$ & & & & & & 0 \\
\hline
\multirow{6}{*}{$\mathfrak{s}_9$} & $e_{12}$ & 0 & 0 & 0 & 0 & $e_{123}$ & $3\alpha e_{123}$ \\\hhline{~|-|-|-|-|-|-|-|}
& $e_{13}$ & & 0 & 0 & 0 & $-3\alpha e_{123}$ & $e_{123}$ \\\hhline{~|-|-|-|-|-|-|-|}
& $e_{14}$ & & & 0 & $2\alpha e_{123}$ & $-\alpha e_{124} - e_{134}$ & $e_{124} - \alpha e_{134}$ \\\hhline{~|-|-|-|-|-|-|-|}
& $e_{23}$ & & & & $2e_{123}$ & $e_{124}$ & $e_{134}$ \\\hhline{~|-|-|-|-|-|-|-|}
& $e_{24}$ & & & & & $-2e_{234}$ & 0 \\\hhline{~|-|-|-|-|-|-|-|}
& $e_{34}$ & & & & & & $-2e_{234}$ \\
\hline
\multirow{6}{*}{$\mathfrak{s}_{10}$} & $e_{12}$ & 0 & 0 & 0 & 0 & 0 & $3 e_{123}$ \\\hhline{~|-|-|-|-|-|-|-|}
& $e_{13}$ & & 0 & 0 & 0 & $-3e_{123}$ & $e_{123}$ \\\hhline{~|-|-|-|-|-|-|-|}
& $e_{14}$ & & & 0 & $2e_{123}$ & $-e_{124}$ & $e_{124} - e_{134}$ \\\hhline{~|-|-|-|-|-|-|-|}
& $e_{23}$ & & & & $2e_{123}$ & $e_{124}$ & $e_{134}$ \\\hhline{~|-|-|-|-|-|-|-|}
& $e_{24}$ & & & & & 0 & 0  \\\hhline{~|-|-|-|-|-|-|-|}
& $e_{34}$ & & & & & & $-2e_{234}$ \\
\hline
\multirow{6}{*}{$\mathfrak{s}_{11}$} & $e_{12}$ & 0 & 0 & 0 & 0 & 0 & $2e_{123}$ \\\hhline{~|-|-|-|-|-|-|-|}
& $e_{13}$ & & 0 & 0 & 0 & $-e_{123}$ & 0 \\\hhline{~|-|-|-|-|-|-|-|}
& $e_{14}$ & & & 0 & $e_{123}$ & 0 & $-e_{134}$ \\\hhline{~|-|-|-|-|-|-|-|}
& $e_{23}$ & & & & $2e_{123}$ & $e_{124}$ & $e_{134}$ \\\hhline{~|-|-|-|-|-|-|-|}
& $e_{24}$ & & & & & 0 & $-e_{234}$ \\\hhline{~|-|-|-|-|-|-|-|}
& $e_{34}$ & & & & & & 0 \\
\hline
\multirow{6}{*}{$\mathfrak{s}_{12}$} & $e_{12}$ & 0 & 0 & 0 & 0 & 0 & $-2e_{124}$ \\\hhline{~|-|-|-|-|-|-|-|}
& $e_{13}$ & & 0 & $-e_{123}$ & 0 & $e_{124}$ & $-e_{134}$ \\\hhline{~|-|-|-|-|-|-|-|}
& $e_{14}$ & & & $-2e_{124}$ & $-e_{124}$ & 0 & $e_{234}$ \\\hhline{~|-|-|-|-|-|-|-|}
& $e_{23}$ & & & & 0 & $-e_{123}$ & $-e_{234}$ \\\hhline{~|-|-|-|-|-|-|-|}
& $e_{24}$ & & & & & $-2e_{124}$ & $-e_{134}$ \\\hhline{~|-|-|-|-|-|-|-|}
& $e_{34}$ & & & & & & 0 \\
\hline
\multirow{6}{*}{$\mathfrak{n}_{1}$} & $e_{12}$ & 0 & 0 & 0 & 0 & 0 & 0 \\\hhline{~|-|-|-|-|-|-|-|}
& $e_{13}$ & & 0 & 0 & 0 & 0 & $-e_{123}$ \\\hhline{~|-|-|-|-|-|-|-|}
& $e_{14}$ & & & 0 & 0 & 0 & $-e_{124}$ \\\hhline{~|-|-|-|-|-|-|-|}
& $e_{23}$ & & & & 0 & $e_{123}$ & 0 \\\hhline{~|-|-|-|-|-|-|-|}
& $e_{24}$ & & & & & $2e_{124}$ & $e_{134}$ \\\hhline{~|-|-|-|-|-|-|-|}
& $e_{34}$ & & & & & & $2e_{234}$ \\
\hline
\end{tabular}
\caption{Schouten brackets between basis elements of $\Lambda^2\mathfrak{g}$  for four-dimensional indecomposable Lie algebras.} \label{Tab:Schoten_2_2}
\end{table}
}
\end{landscape}

{\scriptsize
\begin{table}[h!]
\centering
\begin{tabular}{|c|c||c|c|c|c|}
\hline
&  & $e_{123}$ & $e_{124}$ & $e_{134}$ & $e_{234}$ \\\hline
\multirow{4}{*}{$\mathfrak{s}_1$} & $e_1$ & 0 & 0 & 0 & 0 \\
& $e_2$ & 0 & 0 & 0 & $-e_{123}$ \\
& $e_3$ & 0 & $-e_{123}$ & 0 & 0 \\
& $e_4$ & $e_{123}$ & $0$ & $e_{134}$ & $e_{134} + e_{234}$ \\\hline
\multirow{4}{*}{$\mathfrak{s}_2$} & $e_1$ & 0 & 0 & 0 & $-e_{123}$ \\
& $e_2$ & 0 & 0 & $e_{123}$ & $-e_{123}$ \\
& $e_3$ & 0 & $-e_{123}$ & $e_{123}$ & 0 \\
& $e_4$ & $3e_{123}$ & $2e_{124}$ & $e_{124} + 2e_{134}$ & $e_{134} + 2e_{234}$ \\
\hline
\multirow{4}{*}{$\mathfrak{s}_3$} & $e_1$ & 0 & 0 & 0 & $-e_{123}$ \\
& $e_2$ & 0 & 0 & $\alpha e_{123}$ & 0 \\
& $e_3$ & 0 & $-\beta e_{123}$ & 0 & 0 \\
& $e_4$ & $(1 + \alpha + \beta) e_{123}$ & $(1+ \alpha) e_{124}$ & $(1 + \beta) e_{134}$ & $(\alpha + \beta) e_{234}$ \\
\hline
\multirow{4}{*}{$\mathfrak{s}_4$} & $e_1$ & 0 & 0 & 0 & $e_{123}$ \\
& $e_2$ & 0 & 0 & $e_{123}$ & $-e_{123}$ \\
& $e_3$ & 0 & $-\alpha e_{123}$ & 0 & 0 \\
& $e_4$ & $(2 + \alpha) e_{123}$ & $2 e_{124}$ & $(1+ \alpha) e_{134}$ & $e_{134} + (1 + \alpha) e_{234}$ \\
\hline
\multirow{4}{*}{$\mathfrak{s}_5$} & $e_1$ & 0 & 0 & 0 & $-\alpha e_{123}$ \\
& $e_2$ & 0 & $e_{123}$ & $\beta e_{123}$ & 0 \\
& $e_3$ & 0 & $-\beta e_{123}$ & $e_{123}$ & 0 \\
& $e_4$ & $(\alpha + 2\beta) e_{123}$ & $(\alpha + \beta) e_{124} - e_{134}$ & $e_{124} + (\alpha + \beta) e_{134}$ & $2 e_{234}$ \\
\hline
\multirow{4}{*}{$\mathfrak{s}_6$} & $e_1$ & 0 & 0 & 0 & 0 \\
& $e_2$ & 0 & 0 & $e_{123}$ & $-e_{124}$ \\
& $e_3$ & 0 & $e_{123}$ & 0 & $-e_{134}$ \\
& $e_4$ & 0 & $e_{124}$ & $-e_{134}$ & 0 \\
\hline
\multirow{4}{*}{$\mathfrak{s}_7$} & $e_1$ & 0 & 0 & 0 & 0 \\
& $e_2$ & 0 & $e_{123}$ & 0 & $-e_{124}$ \\
& $e_3$ & 0 & 0 & $e_{123}$ & $-e_{134}$ \\
& $e_4$ & 0 & $-e_{134}$ & $e_{124}$ & 0 \\
\hline
\multirow{4}{*}{$\mathfrak{s}_8$} & $e_1$ & 0 & 0 & 0 & $-(1+\alpha) e_{123}$ \\
& $e_2$ & 0 & 0 & $e_{123}$ & $-e_{124}$ \\
& $e_3$ & 0 & $-\alpha e_{123}$ & 0 & $-e_{134}$ \\
& $e_4$ & $2(1 + \alpha) e_{123}$ & $(2 + \alpha) e_{124}$ & $(1 + 2\alpha) e_{134}$ & $(1 + \alpha) e_{234}$ \\
\hline
\multirow{4}{*}{$\mathfrak{s}_9$} & $e_1$ & 0 & 0 & 0 & $-2\alpha e_{123}$ \\
& $e_2$ & 0 & $e_{123}$ & $\alpha e_{123}$ & $-e_{124}$ \\
& $e_3$ & 0 & $-\alpha e_{123}$ & $e_{123}$ & $-e_{134}$ \\
& $e_4$ & $4 \alpha e_{123}$ & $3 \alpha e_{124} - e_{134}$ & $3\alpha e_{134} + e_{124}$ & $2\alpha e_{234}$ \\
\hline
\multirow{4}{*}{$\mathfrak{s}_{10}$} & $e_1$ & 0 & 0 & 0 & $-2e_{123}$ \\
& $e_2$ & 0 & 0 & $e_{123}$ & $-e_{124}$ \\
& $e_3$ & 0 & $-e_{123}$ & $e_{123}$ & $-e_{134}$ \\
& $e_4$ & $4 e_{123}$ & $3 e_{124}$ & $e_{124} + 3 e_{134}$ & $2 e_{234}$ \\
\hline
\multirow{4}{*}{$\mathfrak{s}_{11}$} & $e_1$ & 0 & 0 & 0 & $-e_{123}$ \\
& $e_2$ & 0 & 0 & $e_{123}$ & $-e_{124}$ \\
& $e_3$ & 0 & 0 & 0 & $-e_{134}$ \\
& $e_4$ & $2e_{123}$ & $2e_{124}$ & $e_{134}$ & $e_{234}$ \\
\hline
\multirow{4}{*}{$\mathfrak{s}_{12}$} & $e_1$ & 0 & 0 & $-e_{123}$ & $-e_{124}$ \\
& $e_2$ & 0 & 0 & $-e_{124}$ & $-e_{123}$ \\
& $e_3$ & $2e_{123}$ & $2e_{124}$ & $e_{134}$ & $e_{234}$ \\
& $e_4$ & 0 & 0 & $-e_{234}$ & $e_{134}$ \\
\hline
\multirow{4}{*}{$\mathfrak{n}_{1}$} & $e_1$ & 0 & 0 & 0 & 0 \\
& $e_2$ & 0 & 0 & 0 & $e_{123}$ \\
& $e_3$ & 0 & 0 & $-e_{123}$ & 0 \\
& $e_4$ & 0 & 0 & $-e_{124}$ & $-e_{134}$ \\
\hline
\end{tabular}
\caption{Schouten brackets between basis elements of $\mathfrak{g}$ and $\Lambda^3\mathfrak{g}$ for real four-dimensional indecomposable Lie algebras.} \label{Tab:Schoten_1_3}
\end{table}
}

\newpage

\end{document}